\documentclass[12pt,reqno]{siamart171218}

\usepackage{graphicx}  
\usepackage{amsmath}  
\usepackage{color}
\usepackage{tikz}
\usepackage{pgfplots}

\newcommand{\R}{{\mathbb R}}\newcommand{\N}{{\mathbb N}}

\newcommand{\Z}{{\mathbb Z}}

\let\epsilon\varepsilon
\let\theta\vartheta
\let\hat\widehat

\DeclareMathOperator{\sech}{sech}

\usepackage{amsfonts}
\usepackage{graphicx}
\usepackage{epstopdf}
\usepackage{algorithmic}
\ifpdf
  \DeclareGraphicsExtensions{.eps,.pdf,.png,.jpg}
\else
  \DeclareGraphicsExtensions{.eps}
\fi

\newsiamremark{remark}{Remark}

\headers{Generalized breathers in time-periodic nonlinear lattices}{C. Chong, D. E. Pelinovsky, G. Schneider}

\title{On the Existence of Generalized Breathers and Transition Fronts in Time-Periodic Nonlinear Lattices}

\author{Christopher Chong\thanks{Department of Mathematics, Bowdoin College, Brunswick, ME 04011, USA \\e-mail: \email{cchong@bowdoin.edu}.}
\and Dmitry E. Pelinovsky\thanks{Department of Mathematics and Statistics, McMaster University, Hamilton, Ontario, Canada, L8S 4K1 \\e-mail: \email{dmpeli@math.mcmaster.ca} }
\and Guido Schneider \thanks{Institut f\"ur Analysis, Dynamik und Modellierung, Universit\"at Stuttgart, 70569 Stuttgart, Germany \\ e-mail: \email{guido.schneider@mathematik.uni-stuttgart.de}}}
\usepackage{amsopn}

\makeatletter
\newcommand*{\addFileDependency}[1]{
  \typeout{(#1)}
  \@addtofilelist{#1}
  \IfFileExists{#1}{}{\typeout{No file #1.}}
}
\makeatother

\begin{document}

\maketitle

\begin{abstract}
We prove the existence of a class of time-localized and space-periodic 
breathers (called $q$-gap breathers) in nonlinear lattices with time-periodic coefficients. These $q$-gap breathers are the counterparts to the classical space-localized and time-periodic breathers found in space-periodic systems. Using normal form transformations, we establish rigorously the existence of such
solutions with oscillating tails (in the time domain) that can be made arbitrarily small, but finite. Due to the presence of the oscillating tails, these solutions are coined generalized $q$-gap breathers. Using a multiple-scale analysis, we also derive a tractable amplitude equation that describes
the dynamics of breathers in the limit of small amplitude. In the presence of damping, we demonstrate
the existence of transition fronts that connect the trivial state to the time-periodic ones. The analytical results are corroborated by systematic numerical simulations.
\end{abstract}

\vspace{-.3cm}

\section{Introduction}

The classical discrete breather is a fundamental coherent structure of nonlinear lattices. They can
be found in many fields, ranging from photonics, electrical circuits, condensed matter physics, molecular biology, chemistry, and phononics \cite{flach_discrete_2008}.
Breathers are relevant for applications, such as information storage and transfer in the context of photonic crystals \cite{photonic1}, but are also rich mathematically and have inspired countless numerical and analytical studies \cite{pgk:2011,Dmitriev_2016}.
The discrete breather is localized in space and periodic in time  with temporal frequency
lying within a frequency gap \cite{flach_discrete_2008}.   Spatially periodic media can have frequency gaps, and hence, discrete breathers are possible in such systems \cite{Huang2}.

If breathers can be found in the frequency gap of spatially periodic media,
what can be found in the wavenumber gap of temporally periodic media? While
this question is a natural one to ask, it has only been very recently addressed.
In the context of a photonic time crystal, it was formally shown in \cite{Superluminal}
that structures that are localized in time and periodic in space can be found in the wavenumber
bandgap of temporally periodic media. The structure reported on had same defining features as the classic breather, but with the role of space and time switched. Such solutions are called  $q$-gap breathers, where $q$ stands for the wavenumber.

In the presence of damping, so-called transition fronts
are possible in a $q$-gap, which connect
the trivial state to time-periodic ones. $q$-gap breathers
and transition fronts were studied numerically and experimentally
in the context of a nonlinear phononic lattice in \cite{chong2}. 
The experimental platform therein was based on the one developed in \cite{chong1}, where bifurcations of time-periodic solutions were studied.


It is the purpose of this paper to establish rigorously the existence of $q$-gap
breathers and transition fronts and to provide a tractable analytical approximation of their dynamics.
$q$-gap breathers are a new type of structure, and are distinct from $q$-breathers, which are localized in wavenumber and periodic in time \cite{Qbreather}. Temporal localization can also be achieved via other mechanisms, including zero-wavenumber gain modulation instability \cite{ZeroGain} and nonlinear resonances \cite{FrequencyComb2,MicroBreather}.
Integrable equations admit such solutions explicitly, 
e.g., the Akhmediev breathers of the nonlinear Schr\"odinger (NLS) equation \cite{AkhmedievNLS} and its discrete counterpart, the Ablowitz-Ladik
lattice \cite{AkhmedievAL}. A feature that distinguishes $q$-gap breathers from other temporally localized structures, like the ones just described, is the fact that the underlying wavenumber lies in a $q$-gap.

Wavenumber bandgaps for the (possible) existence of $q$-gap breathers can be found in
a wide class of temporally periodic lattices.  Indeed, there have been
many recent advances in experimental platforms for time-varying systems, including
photonic \cite{soljacic_optimal_2002,wang_optical_2008,wen_tunable_2020,Rechtsman2021},
electric \cite{powell_multistability_2008,kozyrev_parametric_2006,powell_asymmetric_2009}, and phononic examples \cite{reyes-ayona_observation_2015, trainiti_time-periodic_2019,marconi_experimental_2020,nassar_nonreciprocity_2020, Kim2023}.
Controllable temporal localization has potential applications in
the creation of phononic frequency combs \cite{FrequencyComb1} (see also \cite{FrequencyComb2,MicroBreather}), energy
harvesting \cite{harvesting2,harvesting3}, or acoustic signal processing \cite{Hartmann2007}.
The alternate mechanism for temporal localization that $q$-gap breathers afford and the wide availability
of platforms in which they may be implemented suggest the potential utility of $q$-gap breathers in  photonic, phononic, electrical, and even chemical or biological applications \cite{chong2}.

\subsection{Model equations and physical motivation} \label{sec:model}

The mathematical model for the present study is a time-periodic
nonlinear lattice,
\begin{equation} 
\underline{m} \ddot{u}_n + c \dot{u}_n + k(t) u_n  = F(u_{n+1}-u_n) - F(u_n-u_{n-1})
\label{FPU-pert}
\end{equation}
with mass $ \underline{m}$, damping parameter $ c \geq 0 $,  
the time-periodic modulation of the spring parameter $ k(t) = k(t+T) $ for a 
period $ T > 0 $, and the inter-particle force $F$.  Assuming Dirichlet boundary conditions $ u_0(t) = u_{N+1}(t) = 0 $ for some integer $N$, we have a $ 2N $-dimensional dynamical system obtained from \eqref{FPU-pert} at $ n = 1, 2, \ldots, N$. We use $U := (u_1,u_2,\dots,u_N)$ for further references in the main results.

We will consider a polynomial form of the inter-particle force
\begin{equation} 
F(w) = K_2 w - K_3 w^2 + K_4 w ^3, \qquad K_2 >0
\label{FPUpoly}
\end{equation}
in which  Eq.~\eqref{FPU-pert} corresponds to the classical
Fermi-Pasta-Ulam-Tsingou (FPUT) lattice 
if $k(t)=0$ and $ c = 0 $ \cite{fermi_studies_1955,FPUreview}. The FPUT lattice is a central equation in the study
of nonlinear waves \cite{Vainstein}, partly due to its relevance as a model
in phononic, electrical, and biological systems (among others), its mathematical
richness \cite{berman_fermi-pasta-ulam_2005}, and its place in history
as the first test-bed for numerical simulations \cite{fpupop}.

One concrete motivation for studying system \eqref{FPU-pert} with a time-periodic stiffness term $k(t) = k(t+T)$ is that it describes
an array of repelling magnets surrounded by time modulated coils. It was in this setting that $q$-gap breathers
were observed experimentally \cite{chong2}. In this case $F(w)$ models
the repulsive force of the magnets, and is given by 
\begin{equation} 
F(w) = -\frac{a_1}{(d+w)^{a_2}}, \quad 
\label{FPUmagnet}
\end{equation}
where $d,a_1,a_2 > 0$ are material parameters. Using the Taylor expansion of
Eq.~\eqref{FPUmagnet} at $w = 0$ gives a correspondence to the FPUT
model with $F$ given by Eq.~\eqref{FPUpoly} with
\begin{equation} 
\label{Taylor}
K_2 = \frac{a_2 a_1}{d^{a_2 + 1}}, \quad 
K_3 = \frac{a_2 (a_2+ 1) a_1}{2 d^{a_2 + 2}}, \quad 
K_4 = \frac{a_2 (a_2 + 1) (a_2+ 2) a_1}{6 d^{a_2+ 3}}.
\end{equation}

For $k(t) = k(t+T)$, we will
use a specific choice for illustrations that is motivated 
by the experimental set-up of \cite{chong2}. In particular, we consider a piecewise constant 
time-periodic parameter function $k(t)$ in the form: 
\begin{equation} \label{spec2}
k(t) = \left\{ \begin{array}{cl}  k_a ,& t \in [0,\tau_d T), \\
k_b ,& t \in [\tau_d T,T), \end{array} \right.
\end{equation}
for a $ \tau_d \in [0,1] $ and where $k_a,k_b$ are the so-called modulation amplitude
parameters and $\tau_d$ is the duty-cycle. Using the rescaling 
$$
u_n(t) \rightarrow \frac{K_2}{K_3} u_n\left(\sqrt{\frac{K_2}{\underline{m}}} t \right)
$$ 
leads to the normalized parameter values with 
$\underline{m},K_2,K_3 \rightarrow 1$. We note that the results of this paper
are applicable for more general parameter choices and time-periodic coefficients $k(t) = k(t+T)$ and to  lattices with $K_j(t) = K_j(t+T)$ for $ j =2,3,\ldots $.

\begin{figure} 
    \centerline{
   \begin{tabular}{@{}p{0.5\linewidth}@{}p{0.5\linewidth}@{}}
     \rlap{\hspace*{5pt}\raisebox{\dimexpr\ht1-.1\baselineskip}{\bf (a)}}
 \includegraphics[height=6cm]{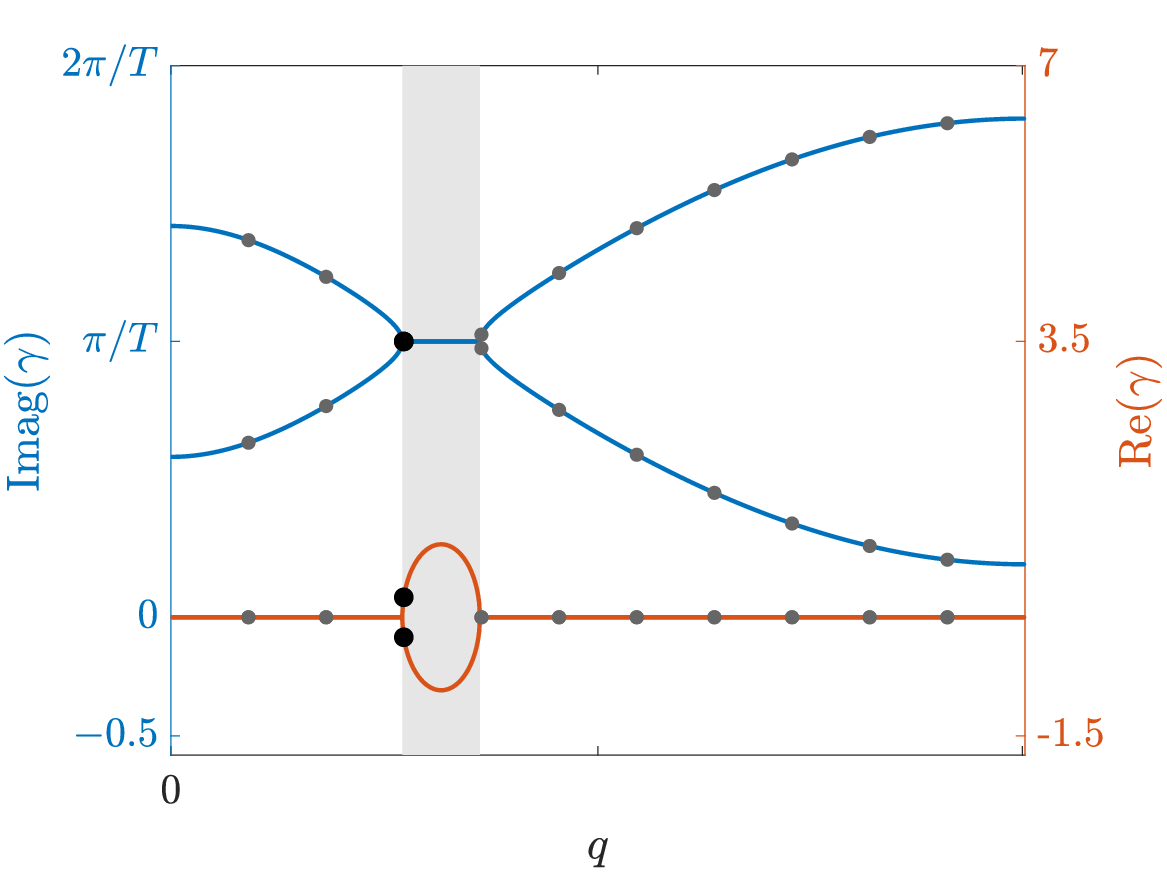}  &
   \rlap{\hspace*{5pt}\raisebox{\dimexpr\ht1-.1\baselineskip}{\bf (b)}}
 \includegraphics[height=6cm]{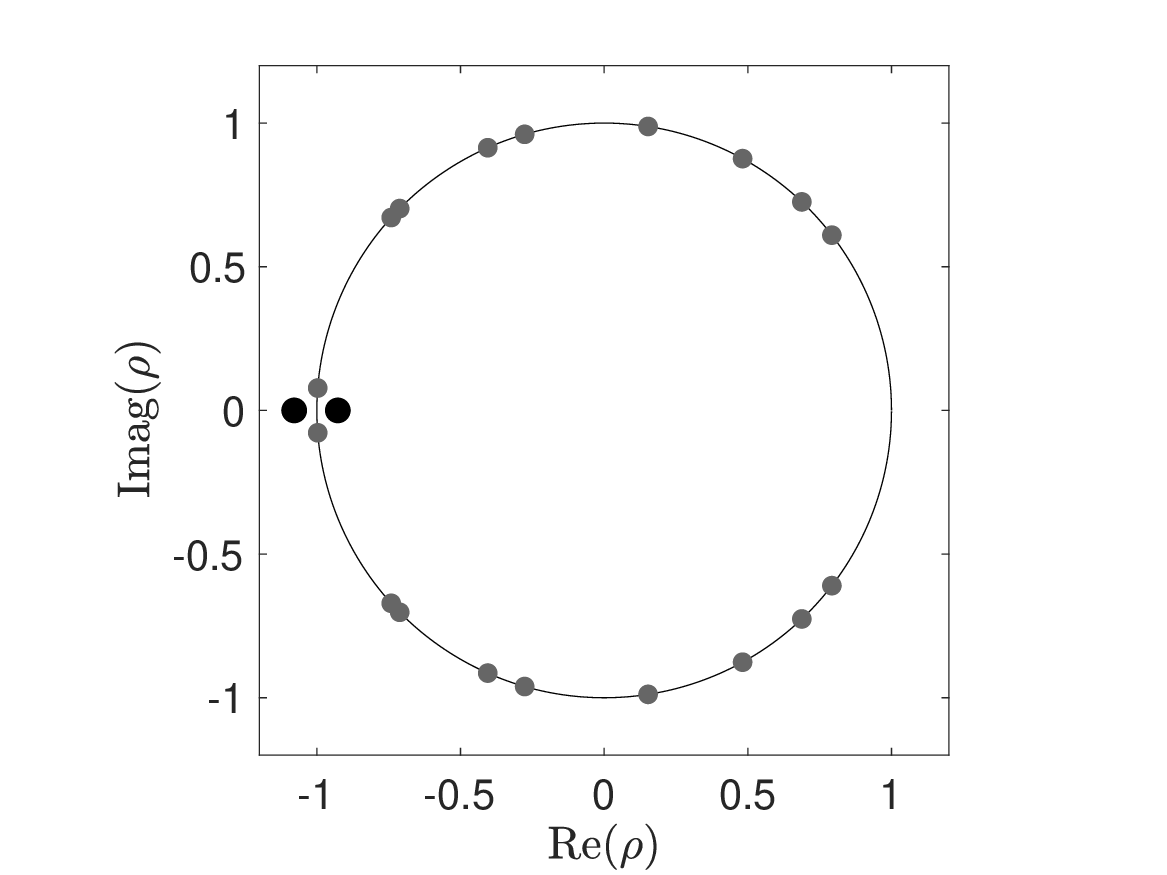} 
  \end{tabular}
  }
\centerline{
   \begin{tabular}{@{}p{0.5\linewidth}@{}p{0.5\linewidth}@{} }
     \rlap{\hspace*{5pt}\raisebox{\dimexpr\ht1-.1\baselineskip}{\bf (c)}}
 \includegraphics[height=6cm]{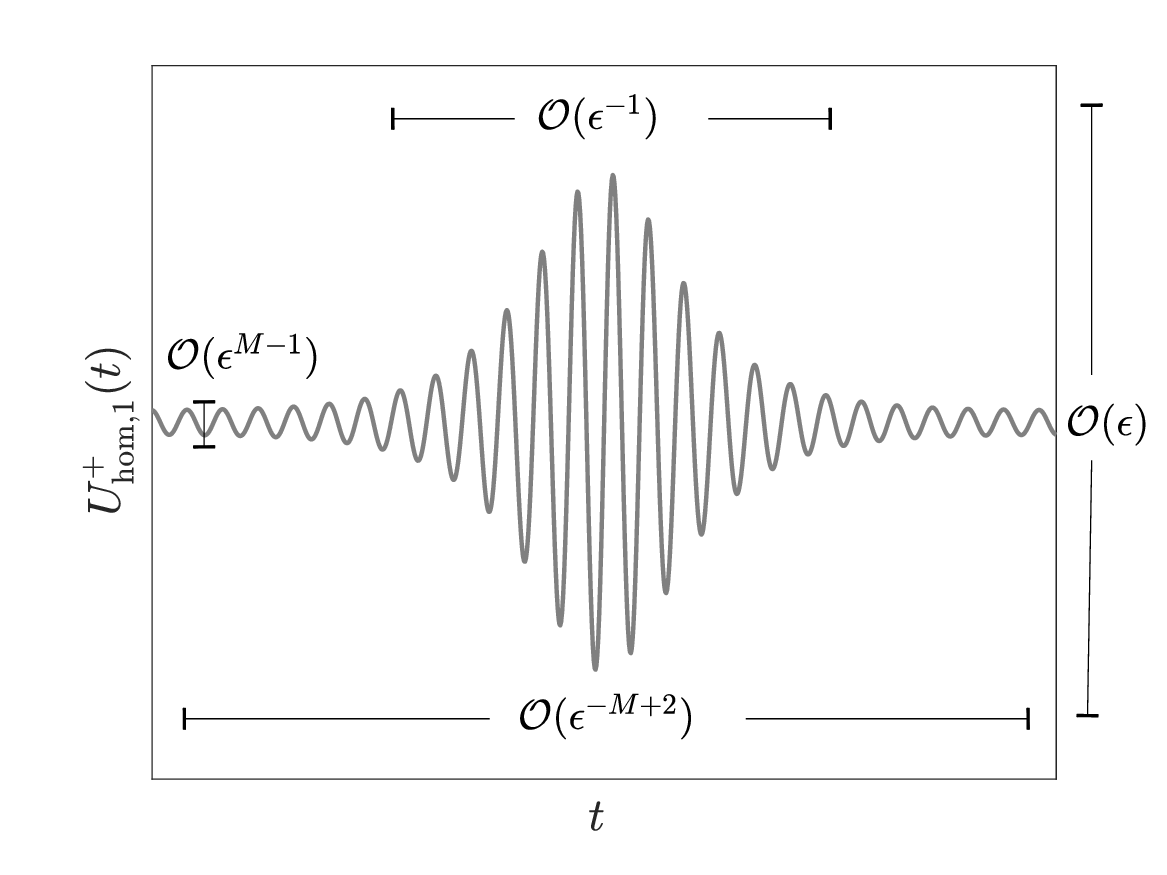}  &
   \rlap{\hspace*{5pt}\raisebox{\dimexpr\ht1-.1\baselineskip}{\bf (d)}}
 \includegraphics[height=6cm]{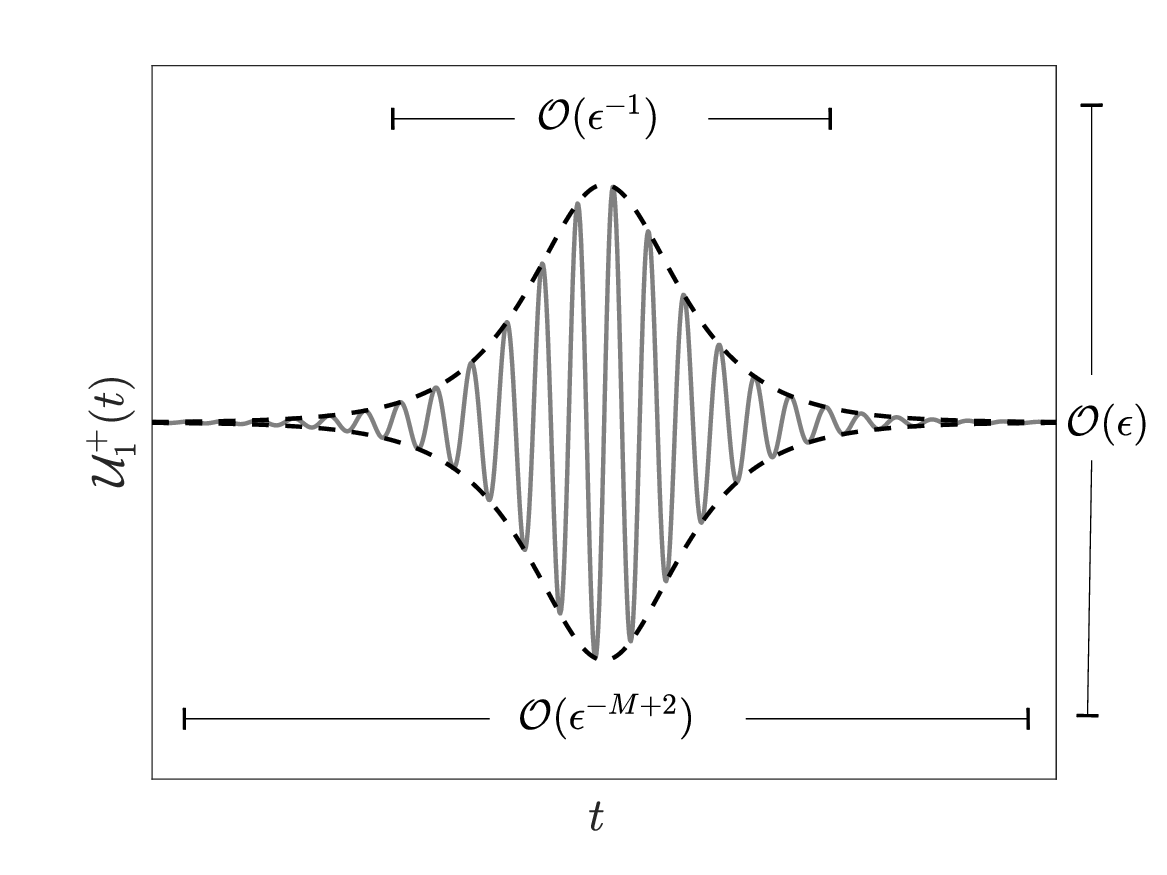} 
  \end{tabular}}
   \caption{ \textbf{(a)}
   The real (red) and imaginary (blue) part of the Floquet exponent $\gamma$
   as a function of Fourier wavenumber $q\in[0,\pi]$ in the infinite lattice. 
   The gray shaded region corresponds to the wavenumber bandgap.  The gray markers
   correspond to the Floquet exponents with a lattice size of $N=10$.  The
    $m_0=3$ exponents lie in wavenumber bandgap (larger black markers).
    The parameter values are $\underline{m}=K_2=1$, $c=0$, $T=1/0.37$, $\tau_d=0.5$, $k_a = 0.6$ and $k_b = 0.8$.
   \textbf{(b)} Floquet multipliers corresponding to panel (a).  One of the
    $m_0=3$ multipliers (larger black markers) has modulus exceeding unity.
   \textbf{(c)} Conceptualization of a generalized $q$-gap breather. The first component of $U^{+}_{\rm hom} = (u_1,u_2,\dots,u_N)$ is shown. 
   \textbf{(d)} The analytical approximation of the solution shown in panel (c), namely $\mathcal{U}^{+}$. 
   The first component of $\mathcal{U}^{+}$ is shown.
    }
   \label{fig:breather_idea}
\end{figure}

\subsection{Summary of main results}

We will develop rigorous proofs of the existence of 
oscillating homoclinic  solutions for $ c = 0 $ and heteroclinic solutions for small damping $ c>0 $ of Eq.~\eqref{FPU-pert} with time-periodic stiffness $k(t)$. Since the tails of homoclinic solutions 
have small oscillations that do not vanish at infinity, the solutions can also be thought of as generalized $q$-gap breathers. Similar
nomenclature has been adopted of the description of classical breathers with non-zero tails in space-time continuous systems \cite{Groves} and with spatially periodic coefficients \cite{DPS24}. See \cite{Fukuizumi2022} for discussion of how the interchange of time and space variables affects derivation and justification of the homoclinic solutions. 

Before stating the main theorems of the paper, let us describe intuitively the generalized $q$-gap breathers for $c=0$.
In the presence of time-periodic stiffness $k(t)$, 
some wavenumbers may fit into the gap in the dispersion relationship, as seen in Fig.~\ref{fig:breather_idea}(a).
The corresponding Floquet multipliers are shown in panel (b) (details on the Floquet theory follow
in Sec.~\ref{sec2}). Unlike the situation for frequency gaps in (linear) spatially periodic media,
exponential growth of Fourier modes occurs if the wavenumber is inside 
the gap of the dispersion relation since the Floquet exponent in the gap has positive real part, or equivalently, 
the Floquet multiplier has modulus exceeding unity. Due to this (parametric) instability,
initializing  Eq.~\eqref{FPU-pert} with such a Fourier mode will initially lead
to growth. However, as the amplitude increases, the nonlinearity of the system comes
into play, and, as we shall prove later, has a localizing affect on the dynamics,
see Fig.~\ref{fig:breather_idea}(c). This solution, however, cannot decay to zero.
This is due to the presence of neutrally stable modes (i.e., the Floquet
multipliers lying on the unit circle). During the dynamic evolution, all of the Fourier
modes will couple (due to the nonlinearity). The presence of the neutrally stable
modes causes the small oscillations, as seen in the tails of Fig.~\ref{fig:breather_idea}(c).
From a dynamical systems point of view, the trivial state has one unstable direction, one stable direction, and $ 2 N-2 $ neutral directions. 
While genuine homoclinic solutions could be found in the intersection of the associated one-dimensional stable and unstable manifolds, it cannot be expected that such an intersection 
exists in an $  2 N $-dimensional phase space. Thus, only homoclinic solutions with small oscillating ripples for $ |t| \to \infty $ exist. These solutions lie in the intersection of the $ 2 N -1 $-dimensional  center-stable manifold with 
the $ 2 N -1 $-dimensional  center-unstable manifold of the origin, for which we use the time-reversibility of the system \eqref{FPU-pert} with $k(t)$ given by \eqref{spec2}.   
The distance the wavenumber of the unstable Fourier mode is to the edge of the gap
defines a small parameter $ \varepsilon > 0  $ (how exactly is detailed
in Sec. \ref{sec2}). With normal form transformations for time-periodic systems (details in Secs.~\ref{sec3}-\ref{sec7}) it can be shown that the oscillating ripples can be made arbitrarily small, i.e., of order $ \mathcal{O}(\varepsilon^{M-1}) $ at the time scale of $\mathcal{O}(\varepsilon^{-M+2})$ with some $  M \in \N$ arbitrarily large but fixed, see Fig.~\ref{fig:breather_idea}(c). Using a multiple-scale analysis (details in Sec.~\ref{sec8}), one can derive an explicit approximation of a genuine homoclinic orbit, see Fig.~\ref{fig:breather_idea}(d), which agrees 
with the leading order of the generalized breather's profile. 

We are now ready to present the main theorem on the existence of two homoclinic orbits with oscillating ripples, i.e., generalized $q$-gap breathers. The assumptions {\bf (Spec)} and {\bf (Coeff)} are described in Secs.~\ref{sec2} and~\ref{sec6} respectively. The small parameter $\varepsilon$ is defined in {\bf (Spec)} and explicitly 
in the representation for $k(t) = k(t+T)$.

\begin{theorem}
	\label{theorem-1}
	Assume the spectral condition {\bf (Spec)} and the normal form coefficient condition {\bf (Coeff)} are satisfied. Then 
	for every $M \in \mathbb{N}$ with $ M \geq 3 $
	there exists an $ \varepsilon_0 > 0 $ and $ C_0 > 0 $ such that for all $ \varepsilon \in (0,\varepsilon_0) $  the system \eqref{FPU-pert} with the time-periodic coefficient \eqref{spec2} and $c = 0$ possesses two generalized homoclinic solutions  
	$ U_{\rm hom}^{\pm} \in C^1([-\varepsilon^{-M+2}, \varepsilon^{-M+2}] ,\R^{N}) $ 
	satisfying
	$$ 
	\sup_{t \in [-\varepsilon^{-M+2}, \varepsilon^{-M+2}]} 
	\| U^{\pm}_{\rm hom}(t) - \mathcal{U}^{\pm}(t)  \| + 
		\| (U^{\pm}_{\rm hom})'(t) - (\mathcal{U}^{\pm})'(t)  \| \leq  C_0 \varepsilon^{M-1}
	$$ 
where $\mathcal{U}^{\pm}(t) : \mathbb{R} \to \mathbb{R}^N$ satisfy
$ \lim\limits_{|t| \to \infty} \| \mathcal{U}^{\pm}(t) \| + \| (\mathcal{U}^{\pm})'(t) \| = 0 $ and can be approximated as 
$$
(\mathcal{U}^{\pm})_n(t) = \pm \varepsilon A(\varepsilon t) g(t) \sin(q_{m_0} n) + \mathcal{O}(\varepsilon^2), 
$$
where $g(t + T) = -g(t)$ and $A(\tau) = \alpha {\rm sech}(\beta \tau)$ are uniquely defined, real-valued functions with some $\alpha, \beta > 0$, 
see Eq.~\eqref{A-soliton} below.
\end{theorem}

In case of small damping $ c = \mathcal{O}(\varepsilon) $ each of the two homoclinic orbits $U_{\rm hom}^{\pm}$ of Theorem \ref{theorem-1} breaks up. There exist two non-zero anti-periodic solutions $\mathcal{U}_{\rm per}^{\pm}(t + T) = -\mathcal{U}_{\rm per}^{\pm}(t)$ with a  $ 2 N $-dimensional stable manifold for $ c > 0 $.

Therefore, as can be seen by counting the dimensions, the one-dimensional unstable manifold from the zero equilibrium intersects the $ 2 N $-dimensional stable manifolds of one the two non-zero anti-periodic solutions $ \mathcal{U}_{\rm per}^{\pm} $ transversally. In contrast to the oscillating homoclinic orbits,  the heteroclinic orbits have no oscillating ripples as $ t \to -\infty $ and converge to the orbits of the anti-periodic solutions $\mathcal{U}_{\rm per}^{\pm}$ as $t \to +\infty$, see Figure \ref{fig3}.

\begin{figure}[htbp] 
   \centering
   \includegraphics[width=4in]{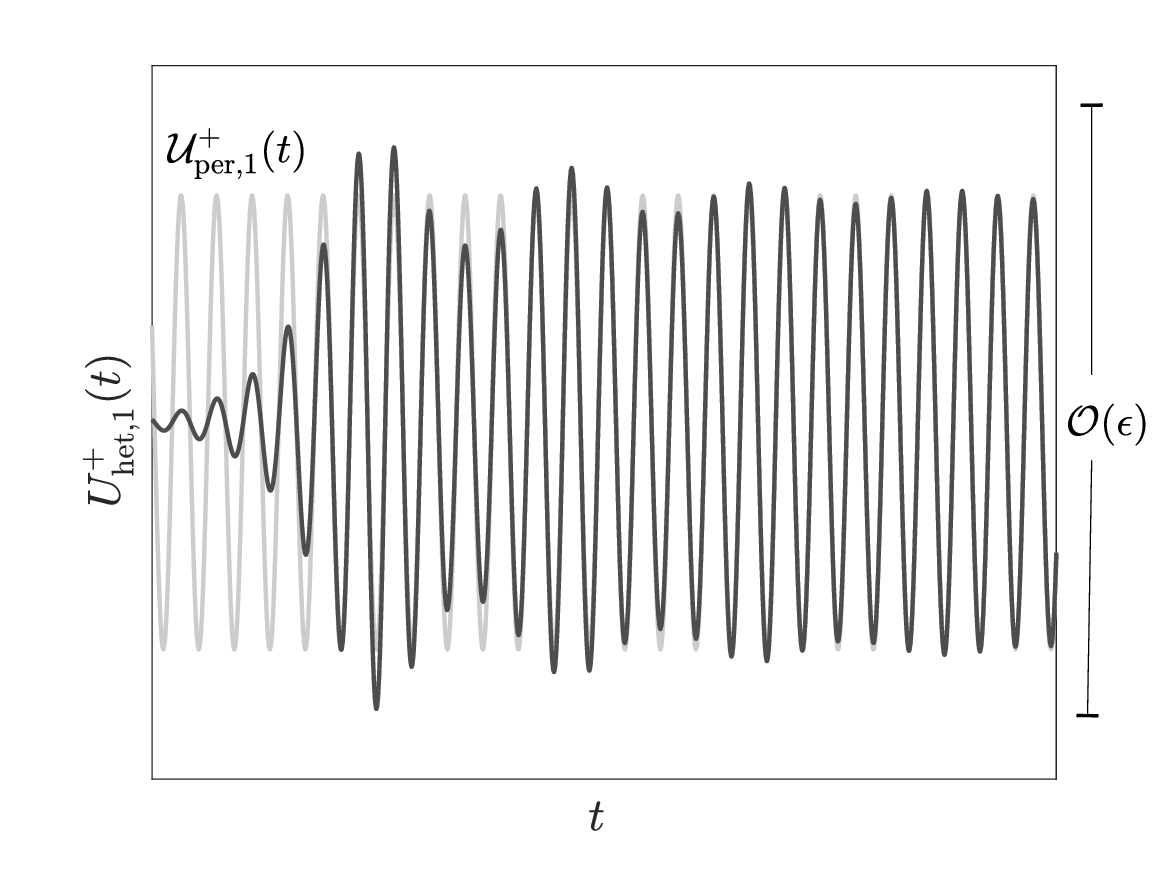} 
   \qquad 
   \caption{
   Example of the transition front $U^{+}_{\rm het} = (u_1,u_2,\dots,u_N)$, where the first component is shown. 
   The first component of the anti-periodic solution $\mathcal{U}^{+}_{\rm per}$ is also shown
   as the light gray line.  }
   \label{fig3}
\end{figure}

Existence of the anti-periodic solutions $\mathcal{U}_{\rm per}^{\pm}$ is guaranteed by the following theorem.

\begin{theorem}
	\label{theorem-2a}
	Assume the spectral condition {\bf (Spec)} and the normal form coefficient condition {\bf (Coeff)}. 
	Fix $ \widetilde{c} \geq  0 $.
	Then there exists an $ \varepsilon_0 > 0 $ and $ C_0 > 0 $ 
	such that for 
	all $ \varepsilon \in (0,\varepsilon_0) $ the system \eqref{FPU-pert} with the time-periodic coefficient \eqref{spec2} and $c = 
	\widetilde{c} \varepsilon > 0$ possesses anti-periodic solutions $\mathcal{U}_{\rm per}^{\pm}$ such that $\mathcal{U}_{\rm per}^{\pm}(t+T) = -\mathcal{U}_{\rm per}(t)$ and 
		$$ 
	\sup_{t \in \mathbb{R} } 
	\| \mathcal{U}^{\pm}_{\rm per}(t)  \| + 
	\| (\mathcal{U}^{\pm}_{\rm per})'(t)  \| \leq  C_0 \varepsilon.
	$$ 
\end{theorem}

The following theorem  presents the main result on the existence of the heteroclinic orbits (transition fronts) between the trivial solution $0$ and the anti-periodic solutions $\mathcal{U}_{\rm per}^{\pm}$.

\begin{theorem}
	\label{theorem-2}
	Assume the spectral condition {\bf (Spec)} and the normal form coefficient condition {\bf (Coeff)}. 
	Fix $ \widetilde{c} > 0 $.
	Then there exists an $ \varepsilon_0 > 0 $ such that for 
	all $ \varepsilon \in (0,\varepsilon_0) $ the system \eqref{FPU-pert} with the time-periodic coefficient \eqref{spec2} and  $c = 
	\widetilde{c} \varepsilon > 0$ possesses 
	two heteroclinic solutions  $ U_{\rm het}^{\pm} \in C^1(\R,\R^{N}) $ such that 
	$$ 
	\lim_{t \to - \infty} U^{\pm}_{\rm het}(t) = 0, \qquad 
	 	\lim_{t \to - \infty} (U^{\pm}_{\rm het})'(t) = 0
	$$ 
	and 
	$$ 
	\lim_{t \to  +\infty} \inf_{t_0 \in [0,T]} \| U^{\pm}_{\rm het}(t) - \mathcal{U}_{\rm per}^{\pm}(t+t_0) \| + \| (U^{\pm}_{\rm het})'(t) - (\mathcal{U}_{\rm per}^{\pm})'(t+t_0) \| = 0,
	$$
where $\mathcal{U}_{\rm per}^{\pm}$ are the anti-periodic   solutions of  \eqref{FPU-pert} from Theorem 
\ref{theorem-2a}.
\end{theorem}

\begin{remark}
	Because of the quadratic nonlinearity in the FPUT system \eqref{FPU-pert}, the homoclinic,  anti-periodic, and heteroclinic orbits in Theorems \ref{theorem-1}-\ref{theorem-2} are not related by the sign reflection even though the leading orders obtained from the cubic normal form are related by the sign reflection, see Eq.  \eqref{amplitude-nonlinear}.  However, if $K_3 = 0$, these orbits are related by the sign reflection up to any orders due to the symmetry of the FPUT system \eqref{FPU-pert}.
\end{remark}

The article is organized as follows. Section \ref{sec2} presents the Floquet and spectral analysis of the linearized FPUT system and introduce assumption ({\bf Spec}). Preparations for the normal form transformations are described in Section \ref{sec3}. Normal form transformations are described in Section \ref{sec4}. The proof of Theorem \ref{theorem-1} is given in Section \ref{sec6}, where assumption ({\bf Coeff}) is introduced. Section \ref{sec7} contains the proof of 
Theorems \ref{theorem-2a} and \ref{theorem-2}. 
A multiple-scale analysis is carried out
in Section \ref{sec8}, which provides tractable approximations for both breathers
and fronts. It also allows verification of ({\bf Coeff}) through direct computation. Numerical illustrations of the main results are described in Section \ref{secdisc}. Section~\ref{theend} concludes the paper with a summary and brief discussions. 

\medskip 

{\bf Funding:} This work was partially supported by the National Science Foundation under grant number DMS-2107945 (C. Chong). D. E. Pelinovsky is partially supported by AvHumboldt Foundation as Humboldt Research Award. G. Schneider is partially supported by the Deutsche Forschungsgemeinschaft 
DFG through the cluster of excellence 'SimTech' under EXC 2075-390740016.

\section{The linearized system}
\label{sec2}

The linearized FPUT system at the trivial (zero) equilibrium is given by 
\begin{equation} 
\underline{m} \ddot{u}_n + c \dot{u}_n + k(t) u_n  = K_2(  u_{n+1}- 2 u_n+ u_{n-1})
\label{FPU-lin}
\end{equation}
for $ n = 1,2, \ldots , N $ with Dirichlet boundary conditions $ u_0(t) = u_{N+1}(t) = 0 $. The linearized system \eqref{FPU-lin} is solved by a linear superposition of the discrete Fourier sine modes:
$$
u_n(t) = \sum_{m=1}^{N} \hat{u}_m(t) \sin(q_m n), \quad q_m := \frac{\pi m}{N+1}, \quad 1 \leq m \leq N.
$$
The $m$-th Fourier mode has the amplitude $\hat{u}_m(t)$ for which the linear FPUT equation \eqref{FPU-lin} transforms to the linear Schr\"{o}dinger 
equation
\begin{equation}
\label{Schr}
\mathcal{L} \hat{u}_m = K_2 \omega^2(q_m) \hat{u}_m, 
\end{equation}
where
\begin{align*}
\mathcal{L} &:= -\underline{m} \partial_t^2 - c \partial_t - k(t), \quad k(t+T) = k(t)
\end{align*}
and
\begin{align*}
\omega^2(q) := 4 \sin^2\left(\frac{q}{2}\right), \quad q \in [0,\pi].
\end{align*}
We will review solutions of the spectral problem \eqref{Schr} by using the Floquet theory and the spectral theory of the Schr\"{o}dinger operators. 

\subsection{Floquet theory} \label{sec:Floquet}

To obtain the monodromy matrix associated to
\eqref{Schr} for general
time-periodic coefficients $k(t)$, one may resort to numerical
computation or perturbation analysis \cite{kovacic_mathieus_2018}. However, 
in the case of piecewise constant $k(t)$ as in \eqref{spec2},
this can be done explicitly \cite{Centurion,Rapti_2004} (see also \cite{chong2}). For the convenience of the readers, we summarize the relevant results.

Let $\lambda_m := K_2 \omega^2(q_m)$ for each $1 \leq m \leq N$, 
and define $s_a, s_b > 0$  by 
\begin{equation}
\label{notation-s}
s_{a,b} := \sqrt{\frac{\lambda_m + k_{a,b}}{\underline{m}} - \frac{c^2}{4 \underline{m}^2}}. 
\end{equation}

Note that $s_a, s_b$ also depend on $m = 1,2,\dots,N$ but the index $m$ is dropped from the notation for simplicity. We obtain the exact solution of \eqref{Schr}:
\begin{equation} \label{eq:bloch}
\hat{u}_m(t) = \left\{ \begin{array}{cl}  e^{-\frac{ct}{2 \underline{m}}} \left[ A_0 \cos(s_a t) + B_0 \sin(s_a t) \right], & t \in [0,\tau_d T), \\
e^{-\frac{ct}{2 \underline{m}}} \left[ C_0 \cos(s_b (t-\tau_d T)) + D_0 \sin(s_b (t-\tau_d T)) \right], & t \in [\tau_d T,T), \end{array} \right.
\end{equation}
with some constants $A_0, B_0, C_0, D_0$. By $C^1$-continuity across $t = \tau_d T$, 
we obtain 
\begin{equation} \label{eq:M1}
\left[ \begin{array}{c} C_0 \\ D_0 \end{array} \right] = \left[ \begin{array}{cc} 
\cos(s_a \tau_d T) & \sin(s_a \tau_d T) \\ 
-\frac{s_a}{s_b} \sin(s_a \tau_d T) & 
\frac{s_a}{s_b} \cos(s_a \tau_d T)
\end{array} \right] 
\left[ \begin{array}{c} A_0 \\ B_0 \end{array} \right].
\end{equation}
The monodromy matrix $J$ is obtained as a mapping 
\begin{equation}
\label{mapping}
\left\{ \begin{array}{l} \hat{u}_m(0) = A_0, \\ 
\hat{u}_m'(0) = s_a B_0 - \frac{c}{2\underline{m}} A_0 \end{array} \right. 
\qquad \Rightarrow \qquad
\left\{ \begin{array}{l} \hat{u}_m(T) = e^{-\frac{cT}{2 \underline{m}}} A_1, \\ 
\hat{u}_m'(T) = e^{-\frac{cT}{2 \underline{m}}} \left[ s_a B_1  - \frac{c}{2\underline{m}} A_1 \right], \end{array} \right. 
\end{equation}
with 
\begin{align}
\left[ \begin{array}{c} A_1 \\ B_1 \end{array} \right] &= \left[ \begin{array}{cc} 
\cos(s_b (1-\tau_d) T) & \sin(s_b (1-\tau_d) T) \\ 
-\frac{s_b}{s_a} \sin(s_b (1-\tau_d) T) & 
\frac{s_b}{s_a} \cos(s_b (1-\tau_d) T)
\end{array} \right]  
\left[ \begin{array}{c} C_0 \\ D_0 \end{array} \right] =: J \left[ \begin{array}{c} A_0 \\ B_0 \end{array} \right].
\label{eq:monodromy}
\end{align}
Since $\det(J) = 1$ and 
\begin{align}
\label{trace-J}
{\rm trace}(J) &= 2 \cos(s_a \tau_d T) \cos(s_b (1-\tau_d) T) - 
\frac{s_a^2 + s_b^2}{s_a s_b} \sin(s_a \tau_d T) \sin(s_b (1-\tau_d) T),
\end{align}
the eigenvalues $\rho_1$ and $\rho_2$ of $J$ satisfies 
$$
\rho_1 \rho_2 = 1, \quad \rho_1 + \rho_2 = {\rm trace}(J),
$$
with only three possibilities:
\begin{itemize}
	\item ${\rm trace}(J) > 2$ implies $0 < \rho_1 < 1 < \rho_2 = \rho_1^{-1}$, 
	\item $-2 \leq {\rm trace}(J) \leq 2$ implies $\rho_1 = \overline{\rho}_2 \in \mathbb{C}$ with $|\rho_{1,2}| = 1$, 
	\item ${\rm trace}(J) < -2$ implies $\rho_2 = \rho_1^{-1} < -1 < \rho_1 < 0$.
\end{itemize}
The Floquet exponents of the mapping \eqref{mapping} are given by 
$ \gamma_{1,2} = \upsilon_{1,2} - \frac{c}{2\underline{m}} $, where 
\begin{itemize}
	\item ${\rm trace}(J) > 2$ implies $\upsilon_{1,2} = \pm \log(\rho_2)/T$, 
	\item $-2 \leq {\rm trace}(J) \leq 2$ implies $\upsilon_{1,2} = \pm i \arg(\rho_1)/T$,
	\item ${\rm trace}(J) < -2$ implies $\upsilon_{1,2} = i \pi/T \pm \log(|\rho_2|)/T$.
\end{itemize}
If $c = 0$, the trivial solution $ U = 0 $ is spectrally stable if all Floquet exponents are purely imaginary. This corresponds to 
the case with $-2 \leq {\rm trace}(J) \leq 2$. Figure \ref{fig:FMs}(a) shows
the Floquet multipliers $\rho$ in the critical case where ${\rm trace}(J) = -2$
with $m=3$. The corresponding Floquet exponents $\gamma = \upsilon$ are shown in Fig. \ref{fig:FMs}(b).

\begin{figure} 
   \centerline{
   \begin{tabular}{@{}p{0.33\linewidth}@{}p{0.33\linewidth}@{}p{0.33\linewidth}@{} }
     \rlap{\hspace*{5pt}\raisebox{\dimexpr\ht1-.1\baselineskip}{\bf (a)}}
 \includegraphics[height=4cm]{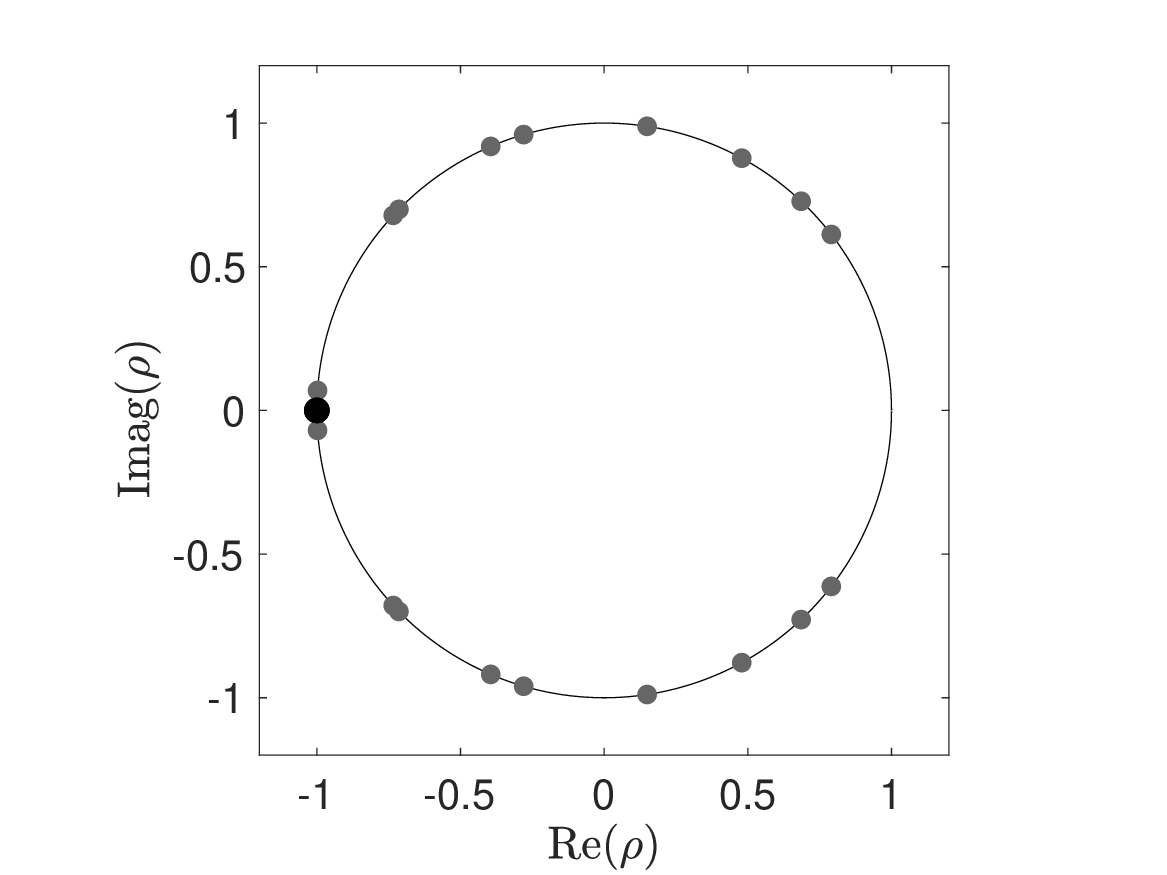}  &
  \rlap{\hspace*{5pt}\raisebox{\dimexpr\ht1-.1\baselineskip}{\bf (b)}}
 \includegraphics[height=4cm]{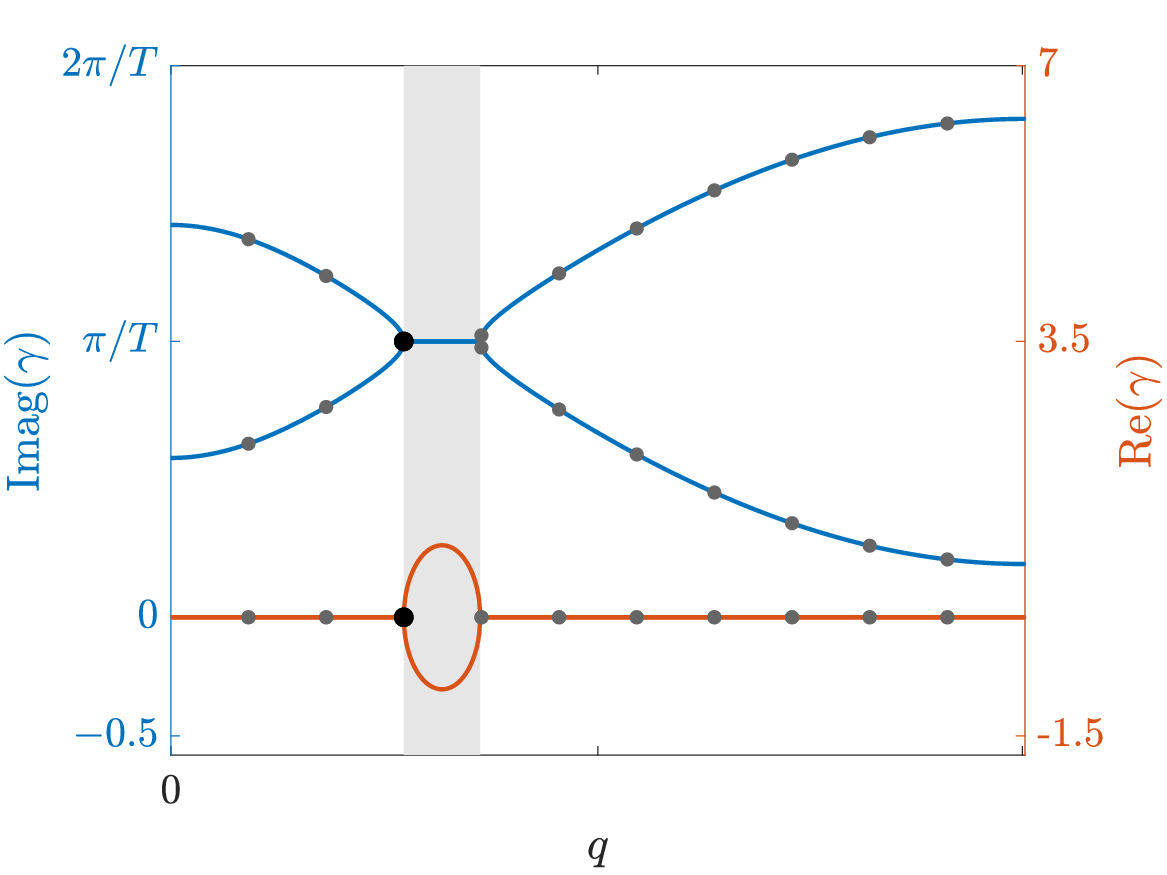} &
   \rlap{\hspace*{5pt}\raisebox{\dimexpr\ht1-.1\baselineskip}{\bf (c)}}
 \includegraphics[height=4cm]{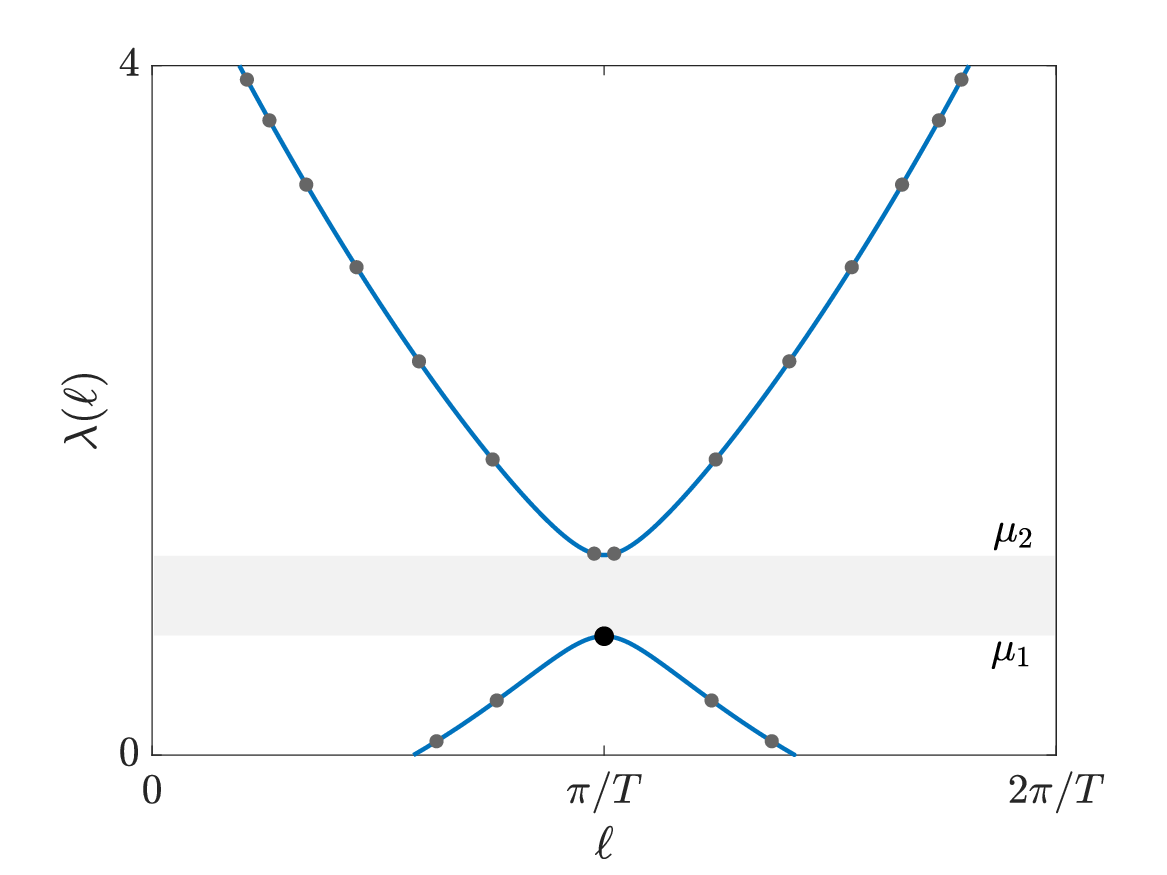} 
  \end{tabular}
  }
   \caption{Bifurcation scenario for the parameter set $\underline{m}=K_2=1$, $c=0$, $T=1/0.37$ and $\tau_d=0.5$. With the critical modulation amplitude parameters $k_a^0 = 0.5$ and $k_b^0 = 0.79$, the critical Fourier mode
   is $m_0 = 3$.
 \textbf{(a)} Plot of the Floquet multipliers, $\rho$, in the complex plane (the unit circle is shown for visual aid). The
    $m_0=3$ multiplier lies exactly at -1 on the unit circle (larger black marker).
 \textbf{(b)} The real (red) and imaginary (blue) part of the Floquet exponent $\gamma$
   as a function of Fourier wavenumber $q\in[0,\pi]$ in the infinite lattice. 
   The gray shaded region corresponds to the wavenumber bandgap.  The gray markers
   correspond to the Floquet exponents with a lattice size of $N=10$.  The
    $m_0=3$ exponent lies exactly at the left edge of the wavenumber bandgap (larger black marker).
   \textbf{(c)} Spectral bands (blue curves) of the Schr\"odinger operator as a function of $\ell = \mathrm{imag}(\gamma)$, the imaginary part of the Floquet exponent. The gray dots show the corresponding values
   in the finite lattice with $N=10$.  The eigenvalues $\mu_1$ and $\mu_2$ define
   the edges of the band gap, shown as the gray shaded region. The
   $m_0=3$ mode lies exactly at the top of the first spectral band (larger black marker). 
 }
   \label{fig:FMs}
\end{figure}

\subsection{Spectral theory}

Let us review the spectral properties of the Schr\"{o}dinger operator 
$$
\mathcal{L} : H^2(\R) \subset L^2(\R) \to L^2(\R)
$$
with a $T$-periodic coefficient $k(t) = k(t+T)$ in the particular case of $c = 0$. The spectrum of $\mathcal{L}$ is purely continuous and consists of bands disjoint from each other by some gaps:
\begin{equation}
\label{spectrum-Schr}
\sigma(\mathcal{L}) = [\nu_0,\mu_1] \cup [\mu_2,\nu_1] \cup [\nu_2,\mu_3] \cup [\mu_4,\nu_3] \cup \cdots
\end{equation}
where $\{ \nu_j \}_{j=0}^{\infty}$ are eigenvalues of $\mathcal{L} f = \nu f$ with periodic boundary conditions $f(t + T) = f(t)$ and 
$\{ \mu_j \}_{j = 1}^{\infty}$ are eigenvalues of $\mathcal{L}g = \mu g$ with anti-periodic boundary conditions $g(t+T) = -g(t)$. Eigenvalues $\{ \nu_j \}_{j=0}^{\infty}$ correspond to ${\rm trace}(J) = 2$ of the Floquet theory, 
whereas eigenvalues $\{ \mu_j \}_{j=0}^{\infty}$ correspond to ${\rm trace}(J) = -2$ and the spectral bands correspond to $-2 \leq {\rm trace}(J) \leq 2$. 

\begin{remark}
Note that $\lambda$ is parameterized by the wavenumber $q$ 
in $\lambda = K_2 \omega^2(q)$, which in turn determines the Floquet exponent $\gamma$. Within the spectral bands, the corresponding Floquet exponents
are purely imaginary and are of the form $\gamma = i \ell$. An example plot showing the dependence of $\lambda$ on $\ell$
is shown in Fig.~\ref{fig:FMs}(c). For this example there is one band gap, which is the shaded region 
in the figure. The band gap edges are given by $\mu_1$ (the bottom of the gap) and $\mu_2$ (the top of the gap).
This representation of the spectrum will be useful later when we derive an amplitude
equation for the description of the envelope of the breather in Sec.~\ref{sec8}. In particular, the concavity of $\lambda(\ell)$
will play a key role.
\end{remark}

We define {\em the bifurcation} for a particular stiffness 
$k(t) \equiv k_0(t)$, for which $\mathcal{L}\equiv \mathcal{L}_0$, 
if there exists an integer $1 \leq m_0 \leq N$ such that $\lambda_{m_0} = K_2 \omega^2(q_{m_0})$ coincides with the end point 
of $\sigma(\mathcal{L}_0)$ and $\lambda_m = K_2 \omega^2(q_m)$ for $m \neq m_0$ are located inside $\sigma(\mathcal{L}_0)$. For the bifurcation shown in  Figure \ref{fig:FMs}, $\lambda_{m_0}$ coincides with $\mu_1$, for which the two bifurcating Floquet exponents are given by $\gamma_0 = i \frac{\pi}{T}$ with $\ell_0 := \frac{\pi}{T}$. This bifurcation corresponds to $\lambda_1''(\ell_0) < 0$ and 
\begin{equation}
\label{bifurcating-mode}
[0,4K] \subset [\nu_0,\nu_1], \quad K_2 \omega^2(q_{m_0}) = \mu_1.
\end{equation}
Equivalently, we can obtain a bifurcation 
when $\lambda_{m_0}$ coincides with $\mu_2$ 
with $\lambda_2''(\ell_0) > 0$ and 
\begin{equation}
\label{bifurcating-mode-alternative}
[0,4K] \subset [\nu_0,\nu_1], \quad K_2 \omega^2(q_{m_0}) = \mu_2.
\end{equation}
The two bifurcating Floquet exponents $\gamma_0$ correspond to the Floquet multiplier $\rho$ at $ -1 $. All other Floquet exponents $\gamma$ are assumed to be on the imaginary axis bounded away from $ 0 $ and $ \gamma_0 $. The corresponding Floquet multipliers $\rho$ are on the unit circle bounded away from $ +1 $ and $ -1 $. 

The bifurcation in terms of Floquet multipliers is shown
in Fig.~\ref{fig:FMs}(a), whereas Fig.~\ref{fig:FMs}(b) shows the bifurcation
in terms of Floquet exponents. The spectral bands of the Schr\"odinger operator
$\mathcal{L}_0$ are shown in Fig.~\ref{fig:FMs}(c). Notice that the discrete mode $m_0 = 3$ lies exactly
at $(\ell,\lambda) = (\ell_0, \mu_1)$, where $\ell_0 = \frac{\pi}{T}$.

\subsection{Spectral assumption and defining the small parameter $\epsilon$}

With the linear theory in hand, we can now specify the spectral assumption as follows. 

{\bf (Spec)} {\it There exists a periodic coefficient $k_0(t+T) = k_0(t)$, for which all Floquet exponents lie on the imaginary  axis. With the exception of two exponents at $\frac{i \pi}{T}$, they are assumed to be simple and non-zero. For small $ \varepsilon > 0 $ we assume that the two Floquet exponents at $\frac{i \pi}{T}$ split symmetrically from the imaginary axis along the real axis to the order of $ \mathcal{O}(\varepsilon) $. }

We can define $\varepsilon$ more explicitly by using the decomposition
\begin{equation}
\label{k-varepsilon}
 k(t) = k_0(t) + \delta \varepsilon^2, 
\end{equation}
where $\delta$ is a proper sign factor. For $c = 0$, the small parameter $\epsilon$ is related to the distance of the critical Floquet exponents 
$\gamma = \upsilon$ from the imaginary axis in the following way. The real part of the Floquet exponent $\gamma$ which depends on $\varepsilon$ is given by 
$$
\mathrm{Re}(\gamma) = \frac{1}{T} \mathrm{cosh}^{-1}\left( -\frac{1}{2} \mathrm{trace}(J) \right).
$$ 
Since we know from ({\bf Spec}) that $\mathrm{trace}(J) = -2$ with $\epsilon=0$ (for $k(t) = k_0(t)$), a series expansion of the real part
of the Floquet exponent about $\epsilon=0$ yields $\mathrm{Re}(\gamma) = \mathcal{O}(\epsilon)$, where we used the fact that $\mathrm{cosh}^{-1}(1+w) \approx \sqrt{2w}$ and  $\mathrm{trace}(J) = -2 + \mathcal{O}(\varepsilon^2)$. In Sec.~\ref{sec:NLS_lin} we will show that
\begin{equation}
\label{gamma-varepsilon}
\mathrm{Re}(\gamma) = \epsilon  \frac{\sqrt{2}}{\sqrt{|\lambda''(\ell_0)|}} + \mathcal{O}(\epsilon^2),
\end{equation}
where $\lambda(\ell)$ is the corresponding band of $\mathcal{L}_0$ at $\ell_0 = \frac{\pi}{T}$ and the sign of $\delta$ is selected to be the opposite of the sign of $\lambda''(\ell_0)$. 

\begin{remark}
One can relate the small parameter $\epsilon$ to the distance of the bifurcating wavenumber $q_{m_0}$ to the band edge in the following way. For a fixed $\epsilon$, suppose the wavenumber bandgap is $[q_\ell ,q_r]$, where the left edge $q_\ell$ and right edge $q_r$ depend on $\epsilon$ and can be found by solving trace($J$)$=-2$ with $\lambda = K_2 \omega^2(q)$.  Suppose that the critical wavenumber $q_{m_0}$ coincides with the left band edge
at the bifurcation point, (i.e., $q_{m_0} = q_\ell$ when $\epsilon=0$). Then, for $\epsilon>0$, the distance 
to the band edge is $\Delta q =q_\ell - q_{m_0}$. By inspection of \eqref{notation-s} and \eqref{trace-J}, if one knows the critical values of $k_a^0$ and $k_b^0$ then
$\Delta q$ can be determined by solving
$\lambda(\ell(q_{m_0})) =\lambda(\ell(q_{m_0} + \Delta q)) + \delta \epsilon^2$
which yields 
\begin{equation} \label{eq:Dq}
     \Delta q = q_{m_0} - 2 \sin^{-1}\left(\sqrt{ \sin^2(q_{m_0}/2) - \frac{\delta \epsilon^2}{4K_2}    }   \right) = \frac{\delta \epsilon^2}{\partial_q\lambda(\ell(q_{m_0}))} + \mathcal{O}(\varepsilon^4).
\end{equation}
Thus $\Delta q = \mathcal{O}(\epsilon^2)$.
\end{remark}

\section{Normal form transformations}
\label{sec3}

The FPUT system \eqref{FPU-pert} consists of $ N $ oscillators with Dirichlet  boundary conditions. By augmenting the vector $U(t) \in \mathbb{R}^N$ with $U'(t) \in \mathbb{R}^N$ as the vector $V(t) \in \mathbb{R}^{2N}$, we rewrite the $ 2 N $-dimensional time-periodic system in the abstract form:
\begin{equation} \label{Msyst}
\dot{V}(t) = Q(t) V(t) + N(V(t)),
\end{equation}
with the time-periodic coefficient matrix $ Q(t) = Q(t+T) \in \R^{2 N \times 2 N}$ being piecewise continuous on $[0,T]$ for a period $ T > 0 $ 
and the nonlinear function $N(V) : \mathbb{R}^{2N} \to \mathbb{R}^{2N}$ being smooth at $V = 0$ with $N(0) = 0$ and $D_V N(0) = 0$. The solutions of the linear system
\begin{equation}
\label{linear-system}
\dot{V}(t) = Q(t) V(t) 
\end{equation}
are, according to Floquet's theorem, of the form
\begin{equation}
\label{Floquet}
V(t) = P(t) e^{\Lambda t} V(0)
\end{equation}
with a $ T $-time-periodic matrix function $ P(t) = P(t+T) \in \R^{2 N \times 2 N}$ and a  time-independent matrix $ \Lambda \in \R^{2 N \times 2 N} $, eigenvalues of which coincide with Floquet exponents in Section \ref{sec2}.

\begin{remark}
Eigenvalues $\gamma$ of the matrix $\Lambda$ are uniquely defined in the strip: 
\begin{equation}
\label{stripe}
-\frac{\pi}{T} < {\rm Imag}(\gamma) \leq \frac{\pi}{T}. 
\end{equation}
Eigenvalues of $\Lambda$ are generally complex-valued but we use the presentation \eqref{Floquet} with real $P(t)$ and real $\Lambda$ for convenience of the normal form transformations. For example, if $\gamma = \alpha \pm i \beta$ are two complex-conjugate eigenvalues of $\Lambda$, then the canonical form for the corresponding block of $\Lambda \in \R^{2 N \times 2 N} $ is 
$$
\left[ \begin{array}{cc} \alpha & \beta \\ -\beta & \alpha \end{array} \right].
$$
\end{remark}

\begin{remark}
As preparations for the normal form transformation, we can consider the time-periodic system \eqref{Msyst} on the double period $ 2 T $. The advantage of this approach is that the bifurcating Floquet exponents in the stripe \eqref{stripe} correspond to zero Floquet exponents in the $2T$-periodic system. The solution \eqref{Floquet} of the linear system \eqref{linear-system} 
can be rewritten in the form:
\begin{equation}
\label{Floquet-double}
V(t) = \widetilde{P}(t) e^{\widetilde{\Lambda} t} V(0),
\end{equation}
with the time-periodic matrix function $\widetilde{P}(t) = \widetilde{P}(t + 2 T) \in \R^{2 N \times 2 N}$ and the time-independent matrix $\widetilde{\Lambda} \in \R^{2 N \times 2 N} $. According to the assumption  ({\bf Spec}) at the bifurcation point, $\widetilde{\Lambda}$ has a double zero eigenvalue and all other (purely imaginary) eigenvalues of $\widetilde{\Lambda}$ are simple and bounded away from $ 0 $.
\end{remark}

\begin{remark}
For the normal form transformation in Section \ref{sec4} 
we need the following property of $ \widetilde{P}(t) $.
Comparing the two representations of the 
fundamental matrix solution
$$ 
\Phi(t) = P(t) e^{\Lambda t} =  \widetilde{P}(t) e^{\widetilde{\Lambda} t},
$$ 
we obtain
$$ 
\widetilde{P}(t) = P(t)  e^{\pi i t/T} = \sum_{m \in \Z} P_m e^{2 \pi im t /T} e^{\pi i t/T},
$$ 
with $ P_m $ being constant $ 2 N \times 2 N $-matrices.
\end{remark}

We now transform the system \eqref{Msyst} on the double period to a convenient form for which  the linear part is autonomous in $t$. Let $ V(t) = \widetilde{P}(t) W(t) $, then $W(t) \in \mathbb{R}^{2N}$ satisfies the time-periodic system:
\begin{equation} \label{Weq}
\dot{W}(t) = \widetilde{\Lambda} W(t) + \widetilde{P}(t)^{-1} N(\widetilde{P}(t) W(t) ).
\end{equation}
We define the projection $ \Pi_0 $ on the subspace associated with the double zero eigenvalue of $  \widetilde{\Lambda} $ by 
\begin{equation} \label{proj0}
\Pi_0 = \frac{1}{2 \pi i} \int_{\Gamma_0} (\lambda I -  \widetilde{\Lambda})^{-1} d \lambda,
\end{equation}
where $ \Gamma_0 $ is a closed curve surrounding the origin in the $\lambda$ plane counter-clockwise. The projection on the two other $ (2N-2) $ eigenvalues of $ \widetilde{\Lambda} $ on the imaginary axis is defined by $ \Pi_h = I- \Pi_0 $. The range of $ \Pi_0 $ is two-dimensional and the range of $ \Pi_h $ is $(2N-2)$-dimensional.

We apply these projections on system \eqref{Weq} 
and find for $ W_0 = \Pi_0 W $ and $ W_h = \Pi_h W $ that 
\begin{eqnarray} \label{Weq0}
\dot{W}_0(t) & = &  \Lambda_0 W_0(t) + N_0(W_0,W_h),\\
\dot{W}_h(t) & = &  \Lambda_h W_h(t) + N_h(W_0,W_h) + H(W_0),
\label{Weqh}
\end{eqnarray}
where we have introduced $ \Pi_0  \widetilde{\Lambda} = \Lambda_0 \Pi_0 $,  $ \Pi_h  \widetilde{\Lambda} = \Lambda_h \Pi_h $,  
\begin{align*}
	N_0(W_0,W_h) & := \Pi_0  \widetilde{P}(t)^{-1} N( \widetilde{P}(t) W(t) ),\\
	N_h(W_0,W_h) + H(W_0) & := \Pi_h  \widetilde{P}(t)^{-1} N( \widetilde{P}(t) W(t) ).
\end{align*}
The splitting into $ N_h(W_0,W_h) + H(W_0) $ is justified with 
$ N_h(W_0,0) = 0 $. System \eqref{Weq0}-\eqref{Weqh} is extended by the additional equation $ \dot{\varepsilon} = 0 $, where $\varepsilon$ is the bifurcation parameter in ({\bf Spec}). 

\begin{remark}
In the context of the time-periodic system \eqref{Weq} on the double period, we recall that $ W_0 $ represents the modes associated 
to the two Floquet exponents which split from the double zero and leave the imaginary axis and that $ W_h $ represents the modes associated 
to the other $ (2 N - 2) $  Floquet exponents which stay 
on the imaginary axis for small bifurcation parameter $\varepsilon$.
\end{remark}

We use the normal form transformations to reduce the order 
of $ H(W_0) $ in terms of powers of $ \| W_0 \| $. 
\begin{lemma}
For every $ M \geq 2 $ there exists an $ \varepsilon_0 > 0 $ such that
for all $ \varepsilon \in (0,\varepsilon_0) $ 
there exists a change of coordinates 
$$
W_{0,M}  = W_0, \qquad 
W_{h,M}  = W_h + G(W_0)
$$
such that system \eqref{Weq0}-\eqref{Weqh} transforms into 
\begin{eqnarray} \label{Wteq0}
\dot{W}_0 & = &  \Lambda_0 W_0 + N_{0,M}(W_0,W_{h,M}),\\
\dot{W}_{h,M} & = &  \Lambda_h W_{h,M} + N_{h,M}(W_0,W_{h,M})+ H_M(W_0),
\label{Wteqh}
\end{eqnarray}
with 
$N_{h,M}(W_0,0) = 0 $ 
and
$ H_M(W_0)= \mathcal{O}(\|W_0\|^M) $.
\end{lemma}
\begin{proof}
We set $ W_{h,2} = W_h $ and then 
inductively 
$$
W_{h,n+1} = W_{h,n} + G_n(W_0) , 
$$
with $  G_n(W_0) $ being a $ n $-linear mapping in $ W_0 $. 
After the transformations  we have a system of the form 
\begin{eqnarray} \label{wq0}
\dot{W}_0 & = & \Lambda_0 W_0 + N_{0,n+1}(W_0,W_{h,n+1}) ,\\
\dot{W}_{h,n+1} & = & \Lambda_h W_{h,n+1} + N_{h,n+1}(W_0,W_{h,n+1}) + H_{n+1}(W_0),
\label{wh0}
\end{eqnarray}
with
\begin{eqnarray*} 
N_{0,n+1}(W_0,W_{h,n+1}) & = & N_{0,n}(W_0,W_{h,n} - G_n(W_0)) ,\\
N_{h,n+1}(W_0,W_{h,n+1}) & = & N_{h,n}(W_0,W_{h,n} - G_n(W_0)) \\ && - (D_{W_0} G_n(W_0)) ( N_{0,n}(W_0,W_{h,n}) - N_{0,n}(W_0,0)), \\
H_{n+1}(W_0) & = & H_{n}(W_0) +  \Lambda_h G_n(W_0) \\ && - (D_{W_0} G_n(W_0)) ( \Lambda_0 W_0 + N_{0,n}(W_0,0) ),
\end{eqnarray*}
and so $ N_{h,n+1}(W_0,0) = 0 $.
In order to have $ H_{n+1}(W_0)  =  \mathcal{O}(\|W_0\|^{n+1}) $ if $ H_{n}(W_0)  
=  \mathcal{O}(\|W_0\|^{n}) $
we have to choose $ G_n $ such that 
\begin{equation} \label{elim1}
H_{n,n}(W_0) +  \Lambda_h G_n(W_0)  - (D_{W_0} G_n(W_0)) \Lambda_0 W_0 = 0  ,
\end{equation}
where $ H_{n,n}(W_0)  $ is the $ n $-linear part of $
H_{n}(W_0) $, i.e. $ H_{n}(W_0) - H_{n,n}(W_0)  = \mathcal{O}(\|W_0\|^{n+1}) $.
Since 
$$ 
H_{n,n}(W_0)(t) = H_{n,n}(W_0)(t+2 T) 
$$ 
we also have to choose  $ G_{n}(W_0)(t) = G_{n}(W_0)(t+2 T) $.
Using Fourier series $$ G_n(W_0)(t) = \sum_{m \in \Z} G_n(W_0)[m] e^{ \pi i m t/T} $$ it is easy to see that \eqref{elim1} can be solved w.r.t. $ G_n $ since 
according to the assumption ({\bf Spec}), none of the eigenvalues $\lambda_j$ of $\Lambda_h$ are located at $ 0 $ and $\Lambda_0$
has a double zero eigenvalue.  As a result, the term $ H(W_0) $ can be made arbitrarily small  in terms of powers of $ \| W_0 \| $.  
\end{proof}

\begin{remark}
The sequence of normal form transformations is not convergent and so we stop after $ M-1 $ transformations with some large but fixed $M \in \mathbb{N}$, for which we have $ H_{M}(W_0) = \mathcal{O}(\|W_0\|^{M}) $. Note that the minimum of $H_{M}(W_0)$ is attained for $ M = \mathcal{O}(1/\| W_0 \|) $-many transformations, after which $ H_M $ is exponentially small  in terms  of $ \| W_0 \| $, cf. \cite{Lombardi}.
\end{remark}

\section{Normal form transformations for the reduced system}
\label{sec4}

If we ignore the terms  $ H_{M}(W_0)  $ in \eqref{Wteq0}-\eqref{Wteqh},  
then $ \{ W_{h,M}  = 0 \} $ is an invariant subspace for  \eqref{Wteq0}-\eqref{Wteqh}.
In this two-dimensional subspace 
the reduced system is obtained by setting $ W_{h,M}  = 0 $ in the first equation.
So we consider the two-dimensional ODE
\begin{equation} 
\label{reduced}
\dot{W}_0 = \Lambda_0 W_0 + N_{0,M} (W_0,0).
\end{equation}
At the bifurcation point 
$ \Lambda_0 $ posseses a double eigenvalue  $ \lambda =  0 $ with algebraic multiplicity two and geometric multiplicity one.
Thus, we have a Jordan-block of size two. The eigenvector of $ \widetilde{\Lambda} $ is denoted with $ \varphi_1 $ and the generalized eigenvector with $ \varphi_2 $, i.e. $ \widetilde{\Lambda}  \varphi_1 = 0 $
and $ \widetilde{\Lambda}  \varphi_2 =  \varphi_1 $. If we introduce coordinates $ A $, $ B $ by 
\begin{equation}
\label{A-def}
W_0 = A \varphi_1 + B \varphi_2 
\end{equation}
we can rewrite \eqref{reduced} as the following two-dimensional system 
\begin{subequations}\label{ABeqneu}
\begin{eqnarray} 
\label{Aeqneu}
\dot{A} & = & B  + f_A(t,A,B), \\
\dot{B} & =  & \widetilde{\varepsilon}^2 A + f_B(t,A,B), \label{Beqneu}
\end{eqnarray}
\end{subequations}
where $ f_A $ and $ f_B $ stand for real-valued  nonlinear terms which are of the form
$$ 
f_A(t,A,B) = \sum_{n = 2}^{\infty} \sum_{j=0}^n \sum_{m \in \Z} f_{A,n,j,m} A^j B^{n-j} e^{2 im\pi t/T} e^{i(n-1) \pi t/T}
$$
and similarly for $ f_B $, where $ f_{A,n,j,m} $ and $ f_{B,n,j,m} $ are independent  of time. To simplify 
notations, we introduce here the normalized small parameter $\widetilde{\varepsilon}$ for the distance of the two Floquet exponents from the imaginary axis. According to the expansion \eqref{gamma-varepsilon}, 
we have 
\begin{equation}
\label{varepsilon-def}
\widetilde{\varepsilon} = \varepsilon \frac{\sqrt{2}}{\sqrt{|\lambda''(\ell_0)|}} + \mathcal{O}(\varepsilon^2). 
\end{equation}
For analyzing \eqref{ABeqneu} we simplify this system by eliminating 
a number of terms by another normal form transformation

\begin{lemma}
There exists an $ \tilde{\varepsilon}_0 > 0 $ such that
for all $ \tilde{\varepsilon} \in (0,\tilde{\varepsilon}_0) $ 
there exists a change of coordinates 
$$
A  = A_3 + F_A(A_3,B_3)  , \qquad B  = B_3 + F_B(A_3,B_3), 
$$
with $ F_A $ and $ F_B $ polynomials not containing linear terms,
such that system \eqref{Aeqneu}-\eqref{Beqneu} transforms into 
\begin{align*}
\dot{A}_3 & = B_3 + f_{A,3,3,-1,2} A_3^3 + f_{A,3,2,-1,2}  A_3^2 B_3 + f_{A,3,1,-1,2}  A_3 B_3^2  + f_{A,3,0,-1,2}  B_3^3 + \mathcal{O}(|A_3|^4+|B_3|^4)
, \\
\dot{B}_3 & =  \widetilde{\varepsilon}^2 A_3 + f_{B,3,3,-1,2} A_3^3 + f_{B,3,2,-1,2}  A_3^2 B_3 + f_{B,3,1,-1,2}  A_3 B_3^2  + f_{B,3,0,-1,2}  B_3^3 + \mathcal{O}(|A_3|^4+|B_3|^4)
,
\end{align*}
with real-valued  coefficients $ f_{A,n,j,m,2} $ and $ f_{B,n,j,m,2} $.
\end{lemma}

\begin{proof}
It is well known that all terms which have a pre-factor which is oscillating in time can be eliminated by a normal form transform 
or equivalently by averaging, cf. \cite{GH83,Sanders}. The technique is elaborated in the Normal Form Theorem III of \cite[Theorem III.13]{Iooss}. 
For the quadratic terms there is no term which has a pre-factor  which is constant in time and so all quadratic terms can be eliminated by  a transformation 
\begin{align*}
A & = A_2 +  \sum_{j=0}^2 \sum_{m \in \Z} g_{A,2,j,m,1} A_2^j B_2^{2-j} e^{2 im\pi t/T} e^{i \pi t/T}, \\
B & = B_2 +  \sum_{j=0}^2 \sum_{m \in \Z} g_{B,2,j,m,1} A_2^j B_2^{2-j}e^{2 im\pi t/T} e^{i \pi t/T}.
\end{align*}
By suitably choosing the coefficients $ g_{A,2,j,m,1} $ and $ g_{B,2,j,m,1} $
we find
\begin{align*}
\dot{A}_2 & =  B_2 + \sum_{n = 3}^{\infty} \sum_{j=0}^n \sum_{m \in \Z} f_{A,n,j,m,1} A_2^j B_2^{n-j}  e^{2 im\pi t/T} e^{i(n-1) \pi t/T}, \\
\dot{B}_2 & =  \widetilde{\varepsilon}^2 A_2 + \sum_{n = 3}^{\infty} \sum_{j=0}^n \sum_{m \in \Z} f_{B,n,j,m,1} A_2^j B_2^{n-j} e^{2 im\pi t/T} e^{i(n-1) \pi t/T} ,
\end{align*}
with new time-independent coefficients $ f_{A,n,j,m,1} $ and $ f_{B,n,j,m,1} $.
For simplifying the cubic terms we make a near identity transformation
\begin{align*}
A_2 & = A_3 +  \sum_{j=0}^3 \sum_{m \in \Z}g_{A,3,j,m,2} A_3^j B_3^{3-j} e^{2 im\pi t/T} , \\
B_2 & = B_3 +  \sum_{j=0}^3 \sum_{m \in \Z} g_{B,3,j,m,2} A_3^j B_3^{3-j} e^{2 im\pi t/T} .
\end{align*}
Again by suitably choosing the coefficients $ g_{A,3,j,m,2} $ and $ g_{B,3,j,m,2}  $
we find
\begin{align*}
\dot{A}_3 & = B_3 + f_{A,3,3,-1,2} A_3^3 + f_{A,3,2,-1,2}  A_3^2 B_3 + f_{A,3,1,-1,2}  A_3 B_3^2  + f_{A,3,0,-1,2}  B_3^3
\\ & \quad  +  \sum_{n = 4}^{\infty} \sum_{j=0}^n \sum_{m \in \Z} f_{A,n,j,m,2} A_3^j B_3^{n-j} e^{2 im\pi t/T} e^{i(n-1) \pi t/T}, \\
\dot{B}_3 & =  \widetilde{\varepsilon}^2 A_3 + f_{B,3,3,-1,2} A_3^3 + f_{B,3,2,-1,2}  A_3^2 B_3 + f_{B,3,1,-1,2}  A_3 B_3^2  + f_{B,3,0,-1,2}  B_3^3
\\ & \quad +  \sum_{n = 4}^{\infty} \sum_{j=0}^n \sum_{m \in \Z} f_{B,n,j,m,2} A_3^j B_3^{n-j} e^{2 im\pi t/T} e^{i(n-1) \pi t/T},
\end{align*}
with new coefficients $ f_{A,n,j,m,2} $ and $ f_{B,n,j,m,2} $.  
\end{proof}

\section{Proof of Theorem \ref{theorem-1}}
\label{sec6}

Here we obtain the oscillating homoclinic solutions with small tails  for $ c = 0 $. The bifurcating solutions scale as $ A_3(t) =  \widetilde{\varepsilon} \widetilde{A}(\tau) $
and $ B_3(t) =  \widetilde{\varepsilon}^2 \widetilde{B}(\tau) $, with $ \tau = \widetilde{\varepsilon} t $.
For the rescaled variables we find 
\begin{subequations}\label{ampl1}
\begin{eqnarray} 
\partial_{\tau}\widetilde{A} & = &  \widetilde{B}  + \mathcal{O}(\widetilde{\varepsilon}), \\
\partial_{\tau}\widetilde{B} & = & \widetilde{A} + f_{B,3,3,-1,2} \widetilde{A}^3 + \mathcal{O}(\widetilde{\varepsilon}). 
\end{eqnarray}
\end{subequations}
Ignoring the terms of order $ \mathcal{O}(\widetilde{\varepsilon}) $ we find two homoclinic solutions to the origin, see the left panel of Figure \ref{fig:homoclinic}, if the following sign condition holds.

{\bf (Coeff)} {\it Assume that 
$$  f_{B,3,3,-1,2} < 0 .$$} 

\begin{remark}
	It is shown in Section \ref{sec8} that {\bf (Coeff)} can generally be satisfied   either at the bifurcation \eqref{bifurcating-mode} or \eqref{bifurcating-mode-alternative}.
\end{remark}

\begin{remark}
	The truncated system \eqref{ampl1}, with $\mathcal{O}(\widetilde{\varepsilon})$ terms neglected,  admits an explicit solution 
	\begin{equation}
	\label{stat-NLS}
	\left\{ \begin{array}{l} 
	\widetilde{A}_{\rm hom,0}(\tau) = \sqrt{2 |f_{B,3,3,-1,2}|^{-1}} {\rm sech}(\tau-\tau_0), \\
	\widetilde{B}_{\rm hom,0}(\tau) = -\sqrt{2 |f_{B,3,3,-1,2}|^{-1}} {\rm tanh}(\tau-\tau_0) {\rm sech}(\tau-\tau_0),
	\end{array} \right.
	\end{equation}
	where $f_{B,3,3,-1,2} < 0$ and $\tau_0 \in \mathbb{R}$ is to our disposal. 
\end{remark}

In the invariant subspace $ \{ W_{h,M}  = 0 \}$ the homoclinic orbits persist 
in the reduced system \eqref{reduced} if the following  reversibility condition 
holds.

{\bf (Rev)} {\it Assume that there exists $t_0 \in [0,T]$ such that $k(t-t_0) = k(t_0-t)$.}  

\begin{remark} \label{reversibility-k}
	Condition {\bf (Rev)} is satisfied for $ k(t) $ defined in \eqref{spec2}
with either $t_0 = \frac{1}{2} \tau_d T$ or $t_0 = \frac{1}{2} (\tau_d+1) T$.
There exist many other possible time-periodic coefficients $ k(t) $ 
satisfying the reversibility condition \eqref{reversibility-k} which lead to the same kind of solutions as in Theorems \ref{theorem-1}. 
\end{remark}

As a consequence of the reversibility,
if $ t \mapsto (\widetilde{A}(t) ,\widetilde{B}(t))  $ is a solution, 
so is 
$$ 
t \mapsto (\widetilde{A}(2t_0-t) ,- \widetilde{B}(2t_0-t)). 
$$ 
Hence, the one-dimensional unstable manifold of the origin of the time $ 2 T $-mapping transversely  intersects the fixed space 
of reversibility $ \{ ( \widetilde{A} ,\widetilde{B} ) : \widetilde{B} = 0 \} $ and continue from the upper half to the lower half of the phase plane, see the right panel of Figure \ref{fig:homoclinic}. Hence, by extending 
$ (\widetilde{A}(t) ,\widetilde{B}(t))_{t \leq t_0} $ by its mirror picture 
$ (\widetilde{A}(2t_0-t) ,-\widetilde{B}(2t_0-t))_{t \geq t_0} $
at the fixed space of reversibility we constructed 
a homoclinic orbit 
$$
(A,B)(t)  = (\widetilde{\varepsilon} \widetilde{A}_{\rm hom}(\varepsilon t) ,\widetilde{\varepsilon}^2 \widetilde{B}_{\rm hom}(\widetilde{\varepsilon} t))
$$
for the reduced system \eqref{ABeqneu}, where
$\widetilde{A}_{\rm hom} = \widetilde{A}_{\rm hom,0} + \mathcal{O}(\widetilde{\varepsilon})$ and $\widetilde{B}_{\rm hom} = \widetilde{B}_{\rm hom,0} + \mathcal{O}(\widetilde{\varepsilon})$ with the leading-order solution given by \eqref{stat-NLS}. 
In the original variables the homoclinic orbit is denoted with $ \mathcal{W} $ 
and corresponds to the truncation $ W_{h,M} = 0 $ for a given $M \in \mathbb{N}$.

\begin{figure}[htbp] 
	\centerline{
		\begin{tabular}{@{}p{0.5\linewidth}@{}p{0.5\linewidth}@{}}
			\rlap{\hspace*{5pt}\raisebox{\dimexpr\ht1-.1\baselineskip}{\bf (a)}}
			\includegraphics[height=6cm]{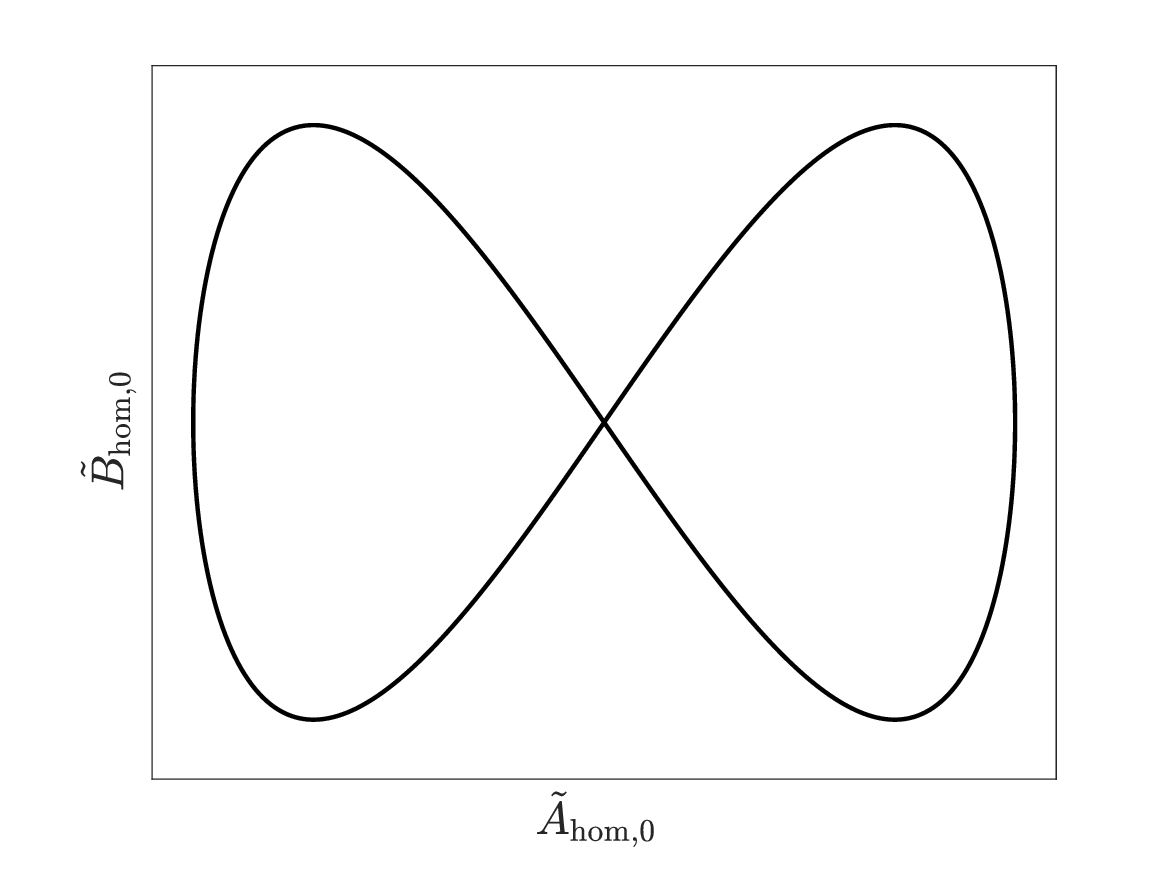}  &
			\rlap{\hspace*{5pt}\raisebox{\dimexpr\ht1-.1\baselineskip}{\bf (b)}}
			\includegraphics[height=6cm]{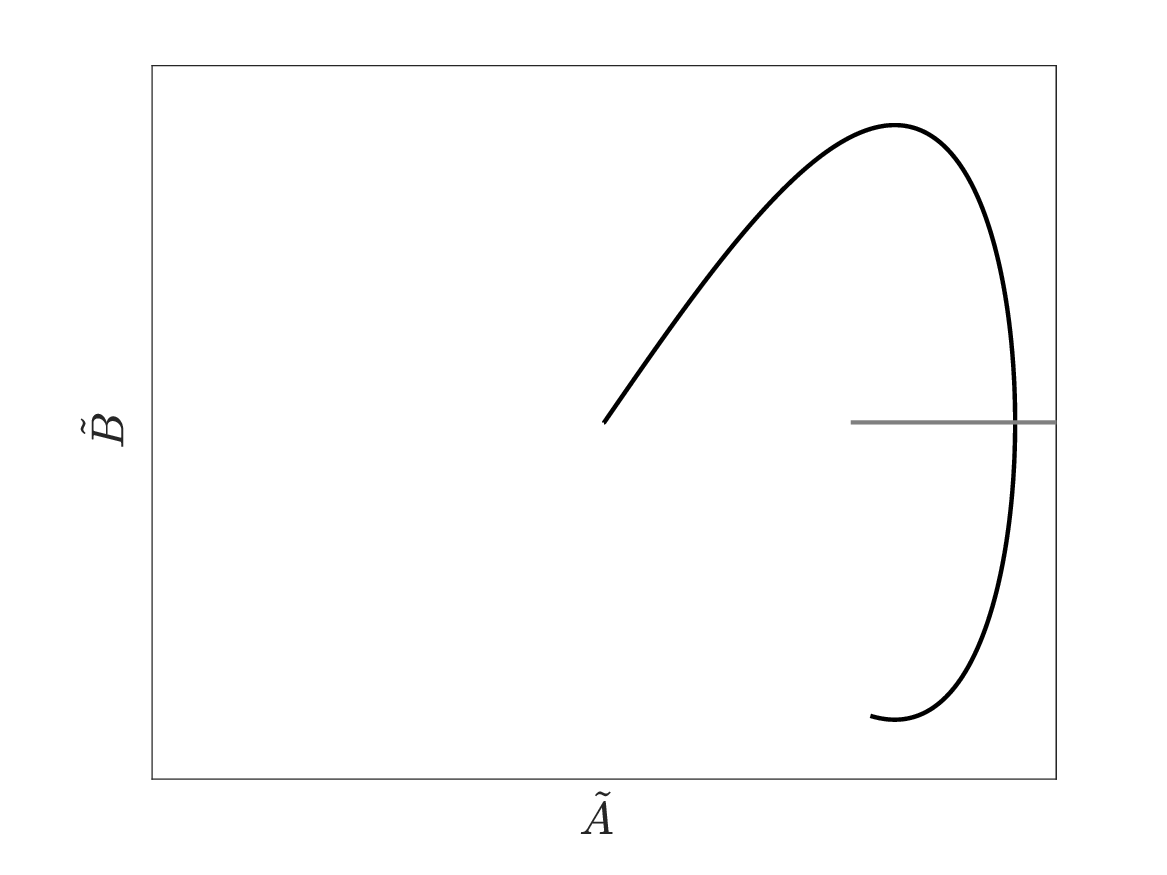} 
		\end{tabular}
	}
	\caption{\textbf{(a)} Two homoclinic solutions
 of Eq.~\eqref{ampl1} with $\mathcal{O}(\widetilde{\varepsilon})$ terms neglected.
 \textbf{(b)} A sketch of the transversal intersection of the unstable manifold with the fixed space of reversibility.}
	\label{fig:homoclinic}
\end{figure}

The rest of this section contains the proof of Theorem \ref{theorem-1}, which we rewrite in notations of Sections \ref{sec3} and \ref{sec4} as follows.
 
\begin{theorem}\label{th61}
	Assume the validity of {\bf (Spec)}, {\bf (Coeff)}, and {\bf (Rev)}. Then there exists an $ \widetilde{\varepsilon}_0 > 0 $ and $ C_0 > 0 $ such that for 
	all $ \widetilde{\varepsilon} \in (0,\widetilde{\varepsilon}_0) $ and every $M \in \mathbb{N}$ with $M \geq 3$ the system \eqref{Weq0}-\eqref{Weqh} possesses a generalized homoclinic solution  
	$ W_{\rm hom} : [-\widetilde{\varepsilon}^{-M+2}, \widetilde{\varepsilon}^{- M+2}] \to \R^{2N} $ 
	with 
	$$ 
	\sup_{t \in [-\widetilde{\varepsilon}^{- M+2}, \widetilde{\varepsilon}^{- M+2}]} 
	\| W_{\rm hom} (t) - \mathcal{W}(t)  \| \leq  
	C_0 \widetilde{\varepsilon}^{M-1}
	$$ 
	with $ \lim\limits_{|t| \to \infty} \mathcal{W}(t) = 0 $. Moreover, for $ \mathcal{W}(t) = (\mathcal{W}_0(t), \mathcal{W}_h(t)) $ we have 
	$$ 
	\sup_{t \in [-\widetilde{\varepsilon}^{- M+2}, \widetilde{\varepsilon}^{- M+2}]}  \| \mathcal{W}_h(t) \| \leq C_0 \widetilde{\varepsilon}^2 
	$$ 
	and
	$$ 
	\sup_{t \in [-\widetilde{\varepsilon}^{- M+2}, \widetilde{\varepsilon}^{- M+2}]}  
	\|\mathcal{W}_0(t) - \widetilde{\varepsilon} A_{\rm hom,0}(\widetilde{\varepsilon}^2 t) \varphi_1 \|
	\leq C_0 \widetilde{\varepsilon}^2,
	$$ 
	with $ A_{\rm hom,0} $ given by \eqref{stat-NLS}.
\end{theorem}

\begin{proof}
To prove persistence of $\mathcal{W}$ if the terms $ H_{M}(W_0) = \mathcal{O}(\|W_0\|^{M}) $ are taken into account, we use again the reversibility. 
Obviously, it is impossible that the one-dimensional 
unstable manifold transversally intersects the $ N $-di\-men\-sional fixed space of reversibility for 
the full system \eqref{Wteq0}-\eqref{Wteqh}. Therefore, it can only be expected that the 
homoclinic solutions persist as solutions 
with small tails for $ |t| \to \infty $. This can rigorously be shown by intersecting the $ 2 N-1 $-dimensional center-unstable 
manifold with the fixed space of reversibility. Obviously this intersection  is a transversal intersection.

On the center-unstable manifold the solutions converge towards the center manifold for $ t \to - \infty $ with some exponential rate. However, the solutions on the center manifold can grow slowly, hence it remains to obtain bounds for such  solutions.
In a first step we apply another normal form transformation 
$$
W_{0,M} = \widetilde{W}_{0,M} + Q_0(W_0,W_{h,M} ), \qquad 
W_{h,M} = \widetilde{W}_{h,M} + Q_h(W_0,W_{h,M} )
$$
to eliminate the bilinear terms which 
are linear in $ W_0 $ and linear in $ W_{h,M} $ from the full system \eqref{Wteq0}-\eqref{Wteqh},
where the 
$$ 
Q_{0,h}(W_0,W_{h,M} ) = \sum_{m \in \Z} Q_{0,h}(W_0,W_{h,M} )[m] e^{2 im\pi t/T} e^{i \pi t/T} 
$$
are bilinear mappings in their arguments.  
With some slight abuse of notation we skip the tildes and reconsider 
the  system \eqref{Wteq0}-\eqref{Wteqh} but now with $ N_{0,M} $ and $ N_{h,M} $
additionally satisfying
\begin{eqnarray*}
\|N_{0,M}(W_0,W_{h,M}) \| & \leq & C (\| W_0 \|^3 + \| W_0 \|^2 \| W_{h,M} \| + \| W_{h,M} \|^2 ) ,\\
\|N_{h,M}(W_0,W_{h,M}) \| & \leq & C (\| W_0 \|^2 \| W_{h,M} \| + \| W_{h,M} \|^2 ) .
\end{eqnarray*}
The transformations are possible due to the spectral assumption {\bf (Spec)}.

In a second step we introduce the deviation $ \widetilde{W}_{0} $ from the homoclinic orbit $ \mathcal{W}_0 $ by
$ W_0 =  \mathcal{W}_0 + \widetilde{W}_{0} $. 
The subsequent estimates on the deviation $ \widetilde{W}_{0} $ have already been carried out in a number of papers,  cf.  \cite{Groves,DPS24}. 
We use the cutoff functions to estimate the solutions on $[-\xi_0,\xi_0]$ with a suitable chosen large $\xi_0$ as $\tilde{\varepsilon} \to 0$.

The homoclinic orbit can be estimated 
pointwise by $ \| \mathcal{W}_0(t)  \| \leq \widetilde{\varepsilon} q (\widetilde{\varepsilon} t) $ with a smooth $ q \in L^1(\R) $.
We denote the stable part  of $ \widetilde{W}_0 $ by $ \widetilde{W}_{0,s} $ and the projection on the stable part by $ \Pi_s $. We find for a large $ \xi_0 $ determined below that
\begin{align*}
\| \widetilde{W}_{0,s}(t) \| & = \left\| \int_{-\xi_0}^t e^{\Lambda_0 (t-\tau)} \Pi_0  \left[ N_{0,M}(\mathcal{W}_0 + \widetilde{W}_{0},W_{h,M}) - 
N_{0,M}(\mathcal{W}_0,0) \right](\tau) d\tau \right\| \\
& \leq C Y(t) \int_{-\infty}^t e^{- \widetilde{\varepsilon} (t-\tau)} \widetilde{\varepsilon}^2 q^2(\widetilde{\varepsilon} \tau) d\tau  + 
C Y(t)^2 \int_{-\infty}^t e^{- \widetilde{\varepsilon} (t-\tau)} d\tau  \\
& \leq C \widetilde{\varepsilon} Y(t) + C \widetilde{\varepsilon}^{-1} Y(t)^2,
\end{align*}
where 
$$ 
Y(t) := \sup_{\tau \in [-\xi_0,t]} (\|   \widetilde{W}_{0}(\tau) \| + \|   W_{h,M}(\tau) \|).
$$
If the solution is in the fixed space of reversibility at $ t = 0 $ we find  
\begin{align*}
\| W_{h,M}(t) \| & = \left\| \int_0^t e^{\Lambda_h (t-\tau)}  (N_{h,M}(W_0,W_{h,M})(\tau) + H_{M}(W_0)(\tau)) d\tau \right\| \\
& \leq C Y(t)  \int_0^t \widetilde{\varepsilon}^2 q^2(\widetilde{\varepsilon} \tau) d\tau 
+ C |t| Y(t)^2  + C  \int_0^t \widetilde{\varepsilon}^{M} q^{M}(\widetilde{\varepsilon} \tau) d\tau  
\\
& \leq C  \widetilde{\varepsilon} Y(t) 
+ C |t| Y(t)^2  + C   \widetilde{\varepsilon}^{M-1} .
\end{align*}
Summarizing the estimates yields 
$$ 
Y(\xi_0)  \leq C \widetilde{\varepsilon} Y(\xi_0) + C \widetilde{\varepsilon}^{-1} Y(\xi_0)^2 + C |\xi_0| Y(\xi_0)^2  + C   \widetilde{\varepsilon}^{M-1}
$$ 
and so $ Y(\xi_0) \leq  C   \varepsilon^{M-1} $ for $ \xi_0 \leq \widetilde{C} \varepsilon^{-M+2} $ and $M \geq 3$. This completes the proof of the theorem in view of the normal form transformations. 
\end{proof}

\section{Proof of Theorem \ref{theorem-2a} and Theorem \ref{theorem-2}}
\label{sec7}

Here we consider the case of small damping $ c> 0 $ where reversibility no longer holds. For consistency of our analysis, we assume that  $ c = \widetilde{c} \widetilde{\varepsilon} $ with $  \widetilde{c}  = \mathcal{O}(1) > 0 $ fixed. With the same analysis as for \eqref{ABeqneu} and \eqref{ampl1}
we end up in $ \{ W_{h,M}  = 0 \}$ at the 
reduced system for the rescaled variables which is now given by 
\begin{subequations}\label{ampl1_damped}
\begin{eqnarray} 
\partial_{\tau}\widetilde{A} & =&  \widetilde{B}  + \mathcal{O}(\widetilde{\varepsilon}), \\
\partial_{\tau}\widetilde{B} & =& \widetilde{A} -  \frac{\widetilde{c}}{\underline{m}} \widetilde{B}+ f_{B,3,3,-1,2} \widetilde{A}^3 + \mathcal{O}(\widetilde{\varepsilon}),
\end{eqnarray}
\end{subequations}
where we used the expression $\gamma = \upsilon - \frac{c}{2\underline{m}} $ for the Floquet exponents of the mapping 
\eqref{mapping} in Section \ref{sec:Floquet}. Ignoring the terms of order $ \mathcal{O}(\widetilde{\varepsilon}) $ we first find two 
fixed points $ (\widetilde{A},\widetilde{B})_{\pm} = (\pm 1/ \sqrt{|f_{B,3,3,-1,2}}|,0) $ if $  f_{B,3,3,-1,2} < 0 $. 
Since the fixed points are hyperbolic they persist as fixed points of the time-$2T$-mapping  
even if the terms of order $ \mathcal{O}(\widetilde{\varepsilon}) $ are included.
This proves Theorem \ref{theorem-2a}.

\begin{figure}[htbp] 
	\centerline{
		\begin{tabular}{@{}p{0.5\linewidth}@{}p{0.5\linewidth}@{}}
			\rlap{\hspace*{5pt}\raisebox{\dimexpr\ht1-.1\baselineskip}{\bf (a)}}
			\includegraphics[height=6cm]{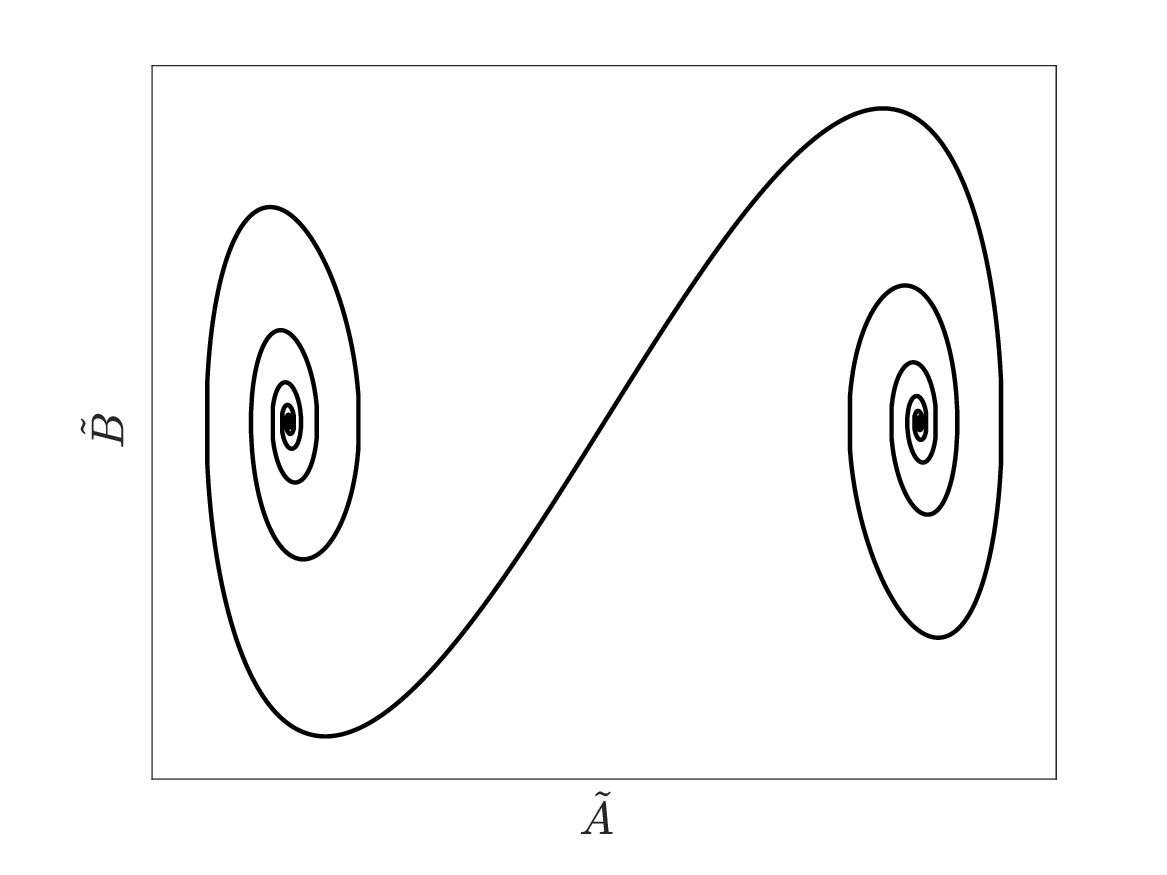}  &
			\rlap{\hspace*{5pt}\raisebox{\dimexpr\ht1-.1\baselineskip}{\bf (b)}}
			\includegraphics[height=6cm]{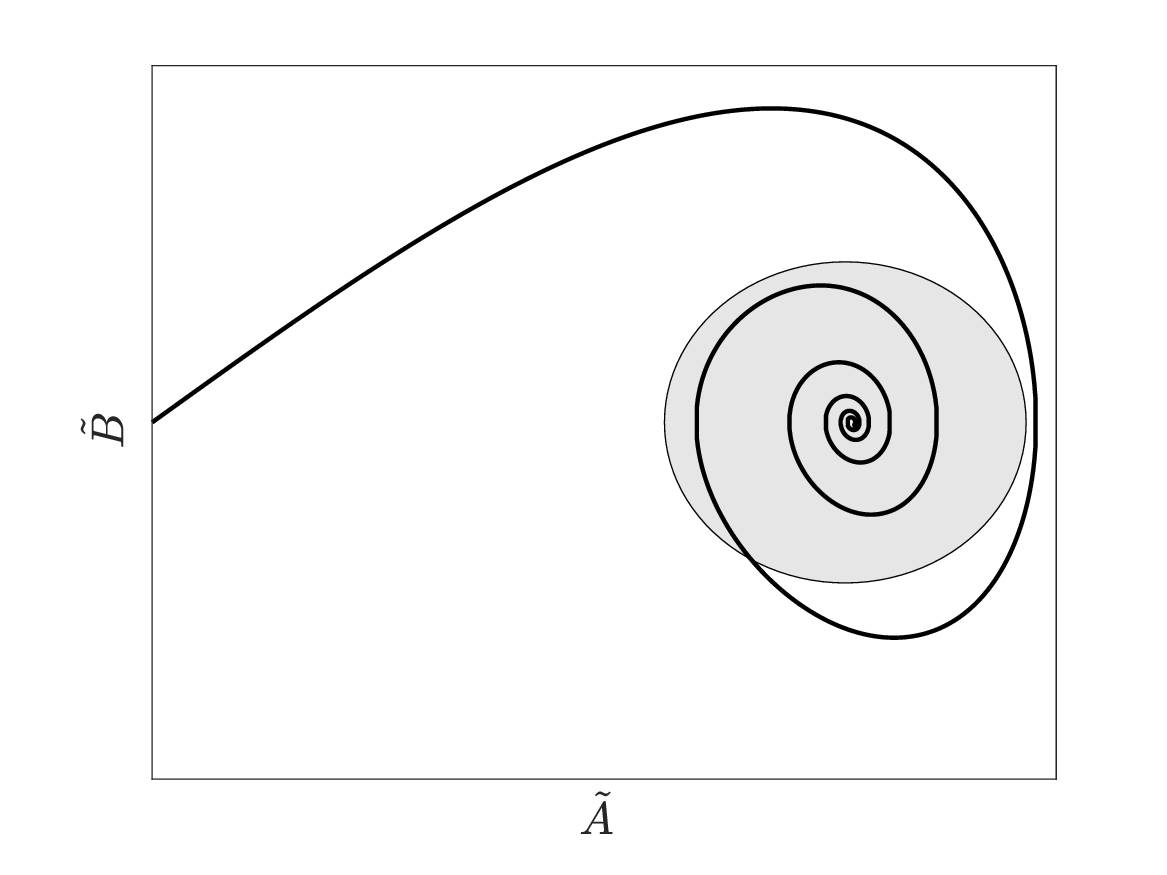} 
		\end{tabular}
	}
   \caption{\textbf{(a)} Two heteroclinic solutions for small $ c > 0 $ of Eq.~\eqref{ampl1_damped} with $\mathcal{O}(\widetilde{\varepsilon})$ terms neglected.
   \textbf{(b)} The heteroclinic solution of Eq.~\eqref{ampl1_damped} in the intersection of the one-dimensional unstable manifold from the zero equilibrium point (in black) and the stable manifold from the positive fixed point (in gray).
   }
   \label{fig:heteroclinic}
\end{figure}

Next we find two heteroclinic solutions
connecting  the origin with the two fixed points $ (\widetilde{A},\widetilde{B})_{\pm} $.
In order to prove the persistence of these heteroclinic solutions we use that  the one-dimensional unstable manifold of the origin transversely  intersects the two-dimensional stable manifold of the two other fixed points 
of the reduced system. These heteroclinic solutions are shown on Figure \ref{fig:heteroclinic}.

To prove the persistence of heteroclinic solutions 
if the terms $ H_{M}(W_0) = \mathcal{O}(|W_0|^{M}) $ are taken into account 
we use that the stable manifold of the two other fixed points is $ 2 N $-dimensional 
and so the one-dimensional unstable manifold of the origin transversally intersects 
the $ 2 N $-dimensional stable manifold of the two anti-periodic solutions 
from Theorem \ref{theorem-2a}. Theorem \ref{theorem-2} is proven by using the transformations 
of Sections \ref{sec3} and \ref{sec4}, with the decomposition 
$W = (W_0, W_h)$.

\section{Multiple-scale analysis and checking the assumption {\bf (Coeff)}}
\label{sec8}

For the system \eqref{FPU-pert} with \eqref{FPUpoly} we already checked 
the spectral condition {\bf (Spec)} used in Theorems \ref{theorem-1}-\ref{theorem-2}. Here we establish the validity of the normal form coefficient condition {\bf (Coeff)}.

We do so by formally deriving the reduced systems \eqref{ampl1} and \eqref{ampl1_damped} 
via a multiple-scale analysis, which will yield an explicit and 
convenient formula for $f_{B,3,3,-1,2}$ after adjusting the notations. 
Truncating \eqref{ampl1_damped} with $\mathcal{O}(\widetilde{\varepsilon})$ terms neglected yields the scalar equation:
\begin{equation}
\label{stat-NLS-eq}
\partial_{\tau}^2\widetilde{A}  + \frac{\widetilde{c}}{\underline{m}} \partial_{\tau} \widetilde{A} =  \widetilde{A} + f_{B,3,3,-1,2} \widetilde{A}^3
\end{equation}

\begin{remark}
\label{rem-normal-form}   
The scalar equation \eqref{stat-NLS-eq} is recovered in Eq. \eqref{amplitude-nonlinear} below 
with the correspondence between $\varepsilon$ and $\widetilde{\varepsilon}$ given by \eqref{varepsilon-def}. Note that the definition $A = \tilde{\varepsilon} \widetilde{A}(\tilde{\varepsilon} t)$, see Eq. \eqref{A-def}, is different from the definition $A(\varepsilon t)$, see Eq. \eqref{A-g-def} below. For Eq. \eqref{stat-NLS-eq}, the amplitude $\widetilde{A}$ is introduced based on the Floquet theorem, the decomposition into subspaces, and the normal form transformations. For Eq. \eqref{amplitude-nonlinear}, the amplitude $A$ is introduced directly by the perturbation expansions in powers of $\varepsilon$, see Eq. \eqref{eq:nls_ansatz}.
\end{remark}

We define as in \eqref{k-varepsilon},
\begin{equation}
\label{bif-k(t)}
k(t) = k_0(t) + \delta \varepsilon^2,
\end{equation}
where $k_0(t)$ is the potential for the bifurcation \eqref{bifurcating-mode} or \eqref{bifurcating-mode-alternative}, $\varepsilon$ is a small parameter for the asymptotic expansion, and $\delta$ is a proper sign factor.

\subsection{Bloch modes of the unperturbed linear problem}

We start with the linear problem for $c = 0$. In order to construct the asymptotic expansion, we define the 
Bloch function  $v(t) = e^{i \ell t} f_j(\ell,t)$ for the $j$-th spectral band $\lambda_j(\ell)$ of $\mathcal{L}_0 v = \lambda v$, 
where $f_j(\ell,t+T) = f_j(\ell,t)$. The parameter $\ell$ is defined in the Brillouin zone $[0,\frac{2\pi}{T})$ 
with $\ell = 0$ and $\ell = \frac{\pi}{T}$ being at the ends of each spectral band corresponding to periodic and anti-periodic solutions respectively. 

Let $L^2_{\rm per}$ be the Hilbert space of $2\pi$-periodic functions with the inner product $\langle \cdot, \cdot \rangle$ and the induced norm $\| \cdot \|_{L^2_{\rm per}}$. The $L^2_{\rm per}$-normalized 
Bloch function $f_j(\ell,t)$ is a $T$-periodic solution 
of the spectral problem 
\begin{equation}
\label{Bloch-1}
\left[ -\underline{m} (\partial_t + i \ell)^2 - k_0(t) \right] f_j(\ell,t) = \lambda_j(\ell) f_j(\ell,t),
\end{equation}
which can be differentiated in $\ell$ as
\begin{align}
& \left[ -\underline{m} (\partial_t + i \ell)^2 - k_0(t) - \lambda_j(\ell) \right] \partial_{\ell} f_j(\ell,t) 
 = 2 i \underline{m} (\partial_t + i \ell) f_j(\ell,t)  + \lambda_j'(\ell) f_j(\ell,t) 
\label{Bloch-2} 
\end{align}
and
\begin{align}
& \left[ -\underline{m} (\partial_t + i \ell)^2 - k_0(t) - \lambda_j(\ell) \right] \partial^2_{\ell} f_j(\ell,t) = 4 i \underline{m} (\partial_t + i \ell) \partial_{\ell} f_j(\ell,t)  \notag \\
& \qquad + 2 \lambda_j'(\ell) \partial_{\ell} f_j(\ell,t) + \left[ \lambda_j''(\ell) -2 \underline{m}\right] f_j(\ell,t),
\label{Bloch-3}
\end{align}
Projecting to $f(\ell,t)$ in $L^2_{\rm per}$ yields from \eqref{Bloch-2} and \eqref{Bloch-3}:
\begin{align}
\lambda_j'(\ell) &= {\color{red} - } 2\underline{m} \langle f_j(\ell,\cdot), i \partial_t f_j(\ell,\cdot) \rangle {\color{red} + } 2 \underline{m} \ell, \label{projection-1} 
\end{align}
and
\begin{align}
\lambda_j''(\ell) & + 2 \lambda_j'(\ell) \langle f_j(\ell,\cdot), \partial_{\ell} f_j(\ell,\cdot) \rangle = -4\underline{m} \langle f_j(\ell,\cdot), i \partial_t \partial_{\ell} f_j(\ell,\cdot) \rangle \notag \\
& \qquad + 4\underline{m} \langle f_j(\ell,\cdot), \partial_{\ell} f_j(\ell,\cdot) \rangle 
\langle f_j(\ell,\cdot),i \partial_t f_j(\ell,\cdot) \rangle + 2 \underline{m}, \label{projection-2}
\end{align}
where the normalization condition $\| f_j(\ell,\cdot) \|_{L^2_{\rm per}} = 1$ has been used.

\begin{remark}
If $\ell_0 = \frac{\pi}{T}$ for the bifurcating mode and the band gap $(\mu_1,\mu_2)$ has a non-zero width, 
then necessarily, $\lambda_j'(\ell_0) = 0$. If the bifurcation \eqref{bifurcating-mode} occurs at $\mu_1$, then 
$j = 1$ and $\lambda_1''(\ell_0) < 0$. See Fig.~\ref{fig:FMs}(c)
for an example. If the bifurcation \eqref{bifurcating-mode-alternative} occurs at $\mu_2$, then $j = 2$ and $\lambda_2''(\ell_0) > 0$.
See Fig.~\ref{fig:simulation_K30_hard}(a).
\end{remark} 

\subsection{Perturbation of the linear problem} \label{sec:NLS_lin}

Let us consider asymptotic solutions of the linear equation 
\begin{equation}
\label{lin-theory}
(\mathcal{L}_0 - K_2 \omega^2(q_{m_0}) - \delta \varepsilon^2 ) \hat{u}(t) = 0,
\end{equation}
which follows from \eqref{Schr} and \eqref{bif-k(t)} at $m = m_0$. As in \eqref{bifurcating-mode}, we take 
$\mu_1 = K_2 \omega^2(q_{m_0})$ for which $\ell_0 = \frac{\pi}{T}$ is selected in the first spectral band $\{ \lambda_1(\ell) \}_{\ell \in [0,\frac{2\pi}{T})}$. The set-up for the second spectral band $\{ \lambda_2(\ell) \}_{\ell \in [0,\frac{2\pi}{T})}$ when $\mu_2 = K_2 \omega^2(q_{m_0})$ as in \eqref{bifurcating-mode-alternative} 
is essentially identical (see
remark \ref{rm:delta}). Expanding 
$$
\hat{u}(t) = A(\varepsilon t) e^{i \ell_0 t} f_1(\ell_0,t) + \varepsilon 
B(\varepsilon t) e^{i \ell_0 t} g_1(\ell_0,t) + \varepsilon^2 e^{i \ell_0 t} C(\epsilon t) h_1(\ell_0,t) + \mathcal{O}(\varepsilon^3),
$$
with $A,B,C$ and $g_1,h_1$ to be determined, 
we obtain for $\mu_1 = K_2 \omega^2(q_{m_0})$ at the order of $\mathcal{O}(\varepsilon) $ that 
\begin{align*}
B \left[ -\underline{m} (\partial_t + i \ell_0)^2 - k_0(t) - \mu_1 \right] g_1(\ell_0,t) = 2m A' (\partial_t + i \ell_0) f_1(\ell_0,t), 
\end{align*}
Since $\lambda_1(\ell_0) = \mu_1$ and $\lambda_1'(\ell_0) = 0$, comparing with \eqref{Bloch-2} yields 
$$
g_1(t) = \partial_{\ell} f_1(\ell_0,t) \quad \mbox{\rm and} \quad 
B(\varepsilon t) = -iA'(\varepsilon t).
$$ 
At the order of $\mathcal{O}(\varepsilon^2)$, we obtain the linear inhomogeneous equation, 
\begin{align}
C \left[ -\underline{m} (\partial_t + i \ell_0)^2 - k_0(t) - \mu_1 \right] h_1(t) &= 
-2\underline{m} i A'' (\partial_t + i \ell_0) \partial_{\ell} f_1(\ell_0,t) 
\notag \\
& \quad + \underline{m} A'' f_1(\ell_0,t) + \delta A f_1(\ell_0,t). 
\label{second-order-eq}
\end{align}
Comparing \eqref{second-order-eq} with \eqref{Bloch-3} yields $h_1(t) = \partial_{\ell}^2 f_1(\ell_0,t)$ and $C(\varepsilon t) = -\frac{1}{2} A''(\varepsilon t)$ if and only if $A(\tau)$, $\tau := \varepsilon t$ satisfies the amplitude equation 
\begin{equation}
\label{amplitude-linear}
\frac{1}{2} \lambda_1''(\ell_0) A''(\tau) + \delta A(\tau) = 0.
\end{equation}
Alternatively, the amplitude equation \eqref{amplitude-linear} can be obtained by projecting the linear inhomogeneous equation \eqref{second-order-eq} to $f_1(\ell_0,t)$ in $L^2_{\rm per}$ and using equations \eqref{projection-1} and \eqref{projection-2} with $\lambda_1'(\ell_0) = 0$.

Since $\delta$ and $\lambda_1''(\ell_0)$ have opposite signs,
$A(\tau)$ of Eq.~\eqref{amplitude-linear} will experience exponential growth with rate 
\begin{equation}\label{varepsilon-again}
\varepsilon \frac{\sqrt{2}}{\sqrt{|\lambda_1''(\ell_0)|}},
\end{equation}
which is the leading order for Eq.~\eqref{varepsilon-def}. This is an approximation
of the real part of the Floquet exponent associated
to the $m_0$-th mode of the
linear problem with $k(t) = k_0(t) + \delta \epsilon^2$, see Eq.~\eqref{gamma-varepsilon}
and Fig.~\ref{fig:simulation_K30}(c) for example.

\begin{remark} \label{rm:delta}
If $\lambda_1''(\ell_0) < 0$ for the bifurcating mode \eqref{bifurcating-mode}, then $\delta = +1$ is selected from the condition that $\mu_1 = K_2 \omega^2(q_{m_0})$ is inside the band gap of the perturbed operator $\mathcal{L} = \mathcal{L}_0 - \delta \varepsilon^2$. On the other hand, if $\lambda_2''(\ell_0) > 0$ for the bifurcating mode  \eqref{bifurcating-mode-alternative}, then $\delta = -1$ is selected from the condition that $\mu_2 = K_2 \omega^2(q_{m_0})$ is inside the band gap of $\mathcal{L}$.
\end{remark}

\subsection{Perturbation of the nonlinear problem}

Let us now consider asymptotic solutions of the nonlinear equation 
\eqref{FPU-pert} with \eqref{bif-k(t)} and $c = \varepsilon \tilde{c}$, 
where $\tilde{c} \geq 0$ is fixed and $\delta = -{\rm sign}(\lambda_1''(\ell_0))$. As in \eqref{bifurcating-mode}, we take $\mu_1 = K_2 \omega^2(q_{m_0})$, for which $\ell_0 = \frac{\pi}{T}$ is selected in the first spectral band $\{ \lambda_1(\ell) \}_{\ell \in [0,\frac{2\pi}{T})}$. To simplify computations, we will use expansions in terms of real-valued functions only. Expanding 
\begin{equation}\label{eq:nls_ansatz}
u_n(t) = \varepsilon U_n^{(1)}(t) + \varepsilon^2 
U_n^{(2)}(t) + \varepsilon^3 U_n^{(3)}(t) + \mathcal{O}(\varepsilon^4),
\end{equation}
we select the leading order in the form 
\begin{equation} 
\label{A-g-def}
U_n^{(1)}(t) =  A(\varepsilon t) g_1(t) \sin(q_{m_0} n),
\end{equation}
where the amplitude $A(\varepsilon t)$ is real according to \eqref{stat-NLS-eq} and the eigenfunction 
	\begin{equation} \label{g1}
	g_1(t) := \left[ e^{i \ell_0 t} f_1(\ell_0,t) + 
	e^{-i \ell_0 t } \bar{f}_1(\ell_0,t) \right] 
\end{equation}
is a real-valued, $T$-anti-periodic solution of $\mathcal{L}_0 g_1 = \mu_1 g_1$ satisfying $g_1(t+T) = - g_1(t)$. At the order of $\mathcal{O}(\varepsilon^2) $ we obtain
\begin{align*}
\mathcal{L}_0 U_n^{(2)} + K_2( U_{n+1}^{(2)} - 2 U_n^{(2)} + U_{n-1}^{(2)}) = H_n^{(2)},
\end{align*}
with 
\begin{align*}
H_n^{(2)} &= 2\underline{m} \partial_{\tau} \partial_t U_n^{(1)} + \tilde{c} \partial_t U^{(1)}_n +  K_3 \left[ (U_{n+1}^{(1)} - U_n^{(1)})^2 - (U_n^{(1)} - U_{n-1}^{(1)})^2 \right],
\end{align*}
where $\tau = \varepsilon t$ and $K_3$ is the coefficient of the quadratic term in \eqref{FPU-pert}. In the explicit form, we obtain 
\begin{align*}
H_n^{(2)} &= [2 \underline{m} A'(\tau)  + \tilde{c}A(\tau)  ]g_1'(t) \sin(q_{m_0} n)  + K_3  A(\tau)^2 g_1(t)^2 F_n^{(2)},
\end{align*}
with
\begin{align*}
F_n^{(2)} &= \left[ \sin(q_{m_0} (n+1)) - \sin(q_{m_0} n) \right]^2 - 
\left[ \sin(q_{m_0}n) - \sin(q_{m_0}(n-1)) \right]^2  \\
&= - 2 \sin(q_{m_0}) (1-\cos(q_{m_0})) \sin(2 q_{m_0} n).
\end{align*}
The solution for $U_n^{(2)}(t)$ can be written in the form
\begin{align*}
U_n^{(2)}(t) &= \left[A'(\tau) + \frac{\tilde{c}}{2\underline{m}} \right]h_1(t) \sin(q_{m_0} n) + 
	K_3 A(\tau)^2 h_2(t) \sin(2 q_{m_0} n),  
\end{align*}
where $h_1$ and $h_2$ are solutions of the linear inhomogeneous equations:
\begin{align}  \label{eq:h1}
(\mathcal{L}_0 - K_2 \omega^2(q_{m_0})) h_1 &= 2 \underline{m} g_1'(t), \\ 
(\mathcal{L}_0 - K_2 \omega^2(2 q_{m_0})) h_2 &= -2 \sin(q_{m_0}) (1 - \cos(q_{m_0})) g_1(t)^2.\label{eq:h2}
\end{align}
It follows from the linear theory that the real, $T$-anti-periodic solution for $h_1(t)$ exists in the form:
$$
h_1(t) = - i e^{i \ell_0 t} \partial_{\ell} f_1(\ell_0,t) + i e^{-i \ell_0 t} \partial_{\ell} \bar{f}_1(\ell_0,t).
$$
There exists a unique $T$-periodic solutions for $h_2(t)$ if and only if the non-resonance condition is met: 
\begin{equation}
\label{non-resonance}
K_2 \omega^2(2 q_{m_0}) \notin \cup_{j = 1}^{\infty} \lambda_j(0).
\end{equation}
This non-resonance condition is satisfied if the spectral assumption ({\bf Spec}) is satisfied. 

At the order of $\mathcal{O}(\varepsilon^3)$, we obtain 
\begin{align*}
\mathcal{L}_0 U_n^{(3)} + K_2( U_{n+1}^{(3)} - 2 U_n^{(3)} + U_{n-1}^{(3)}) = H_n^{(3)},
\end{align*}
with 
\begin{align*}
H_n^{(3)} &= \delta U_n^{(1)} + 2 m \partial_{\tau} \partial_t U_n^{(2)} 
+ m \partial_{\tau}^2 U_n^{(1)} +  \tilde{c} \partial_t U_n^{(2)} + \tilde{c} \partial_{\tau} U_n^{(1)}
\\
& \quad + 2K_3 \left[ (U_{n+1}^{(1)} - U_n^{(1)}) (U_{n+1}^{(2)} - U_n^{(2)}) - (U_n^{(1)} - U_{n-1}^{(1)}) (U_n^{(2)} - U_{n-1}^{(2)}) \right] \\
& \quad 
- K_4 \left[ (U_{n+1}^{(1)} - U_n^{(1)})^3 - (U_n^{(1)} - U_{n-1}^{(1)})^3 \right],
\end{align*}
where $K_4$ is the coefficient of the cubic term in \eqref{FPU-pert}. 

By using Euler's formulas, we obtain 
\begin{align*}
& \left[ \sin(q_{m_0} (n+1)) - \sin(q_{m_0} n) \right]^3 - \left[ \sin(q_{m_0}n) - \sin(q_{m_0}(n-1)) \right]^3  \\
& \qquad = -3 (1-\cos(q_{m_0}))^2 \sin(q_{m_0} n) \\
& \qquad \qquad +\frac{1}{2} 
\left[ 1 - 3 \cos(q_{m_0}) + 3 \cos(2 q_{m_0}) - \cos(3 q_{m_0}) \right] \sin(3 q_{m_0} n)
\end{align*}
and 
\begin{align*}
&  \left[ \sin(q_{m_0} (n+1)) - \sin(q_{m_0} n) \right] 
\left[ \sin(2 q_{m_0} (n+1)) - \sin(2 q_{m_0} n) \right] \\
& \qquad \qquad - \left[ \sin(q_{m_0}n) - \sin(q_{m_0}(n-1)) \right] \left[ \sin(2 q_{m_0}n) - \sin(2 q_{m_0}(n-1)) \right]
 \\
& \qquad = - \left[ 2 \sin(q_{m_0}) - \sin(2 q_{m_0}) \right] \sin(q_{m_0} n) \\
& \qquad \qquad - \left[ \sin(q_{m_0}) + \sin(2 q_{m_0}) - \sin(3 q_{m_0}) \right] \sin(3 q_{m_0} n)
\end{align*}
Hence, we obtain
\begin{align*}
H_n^{(3)} = D_1(t) \sin(q_{m_0} n) +  D_2(t) \sin(2 q_{m_0} n) + D_3(t) \sin(3 q_{m_0} n),
\end{align*}
where we are only interested to write explicitly 
	the coefficient for the bifurcating mode:
\begin{align*}
D_1(t) &= \delta A(\tau) g_1(t) + m A''(\tau) g_1(t) + 2 m A''(\tau) h_1'(t) \\
& \quad { + \tilde{c} \left[ A'(\tau) g_1(t) + 2 A'(\tau) h_1'(t) \right] + \frac{\tilde{c}^2}{2 \underline{m}} A(\tau) h_1'(t) } \\
& \quad + 3 K_4 (1-\cos(q_{m_0}))^2 
A(\tau)^3 g_1(t)^3 - 2 K_3^2 \left[ 2 \sin(q_{m_0}) - \sin(2 q_{m_0}) \right] 
A(\tau)^3 g_1(t) h_2(t).
\end{align*}
Projecting $D_1(t)$  to $g_1(t)$
gives the amplitude equation for $A(\tau)$:
\begin{equation}
\label{amplitude-nonlinear}
\frac{1}{2} \lambda_1''(\ell_0) 
\left[ A''(\tau) + \frac{\tilde{c}}{\underline{m}} A'(\tau) 
	+ \frac{\tilde{c}^2}{4 \underline{m}^2} A(\tau) \right] 
+ \left[ \delta - \frac{\tilde{c}^2}{4 \underline{m}} \right] A(\tau) + \chi 
A(\tau)^3 = 0, 
\end{equation}
where 
\begin{equation} \label{lambda_pp_again}
\lambda_1''(\ell_0) = 2m + 4m \frac{\langle g_1, h_1' \rangle}{\| g_1 \|^2}    
\end{equation}
and
\begin{align} \label{eq:chi}
\chi &= 3 K_4 (1-\cos(q_{m_0}))^2  \frac{\langle g_1^2, g_1^2 \rangle}{\| g_1 \|^2}
- 2K_3^2 \left[ 2 \sin(q_{m_0}) - \sin(2 q_{m_0}) \right] \frac{\langle g_1^2, h_2 \rangle}{\| g_1\|^2}.
\end{align}
Since the linear part of \eqref{amplitude-nonlinear} should be identical to the linear amplitude equation \eqref{amplitude-linear}, the new formula for $\lambda_1''(\ell_0)$ must be identical to the previous equation \eqref{projection-2} with $\lambda_1'(\ell_0) = 0$. The expression for $\chi$ is defined in terms of real quantities only.

\begin{remark}
Equation \eqref{amplitude-nonlinear}  for $\tilde{c}= 0$ is analogous to the stationary NLS equation that can be derived in the context of spatially periodic media for the description of breathers \cite{Huang2}. A similar equation
was derived in \cite{Superluminal} for a (space-time continuous) photonic time crystal. Eq.~\eqref{amplitude-nonlinear} is equivalent to \eqref{stat-NLS-eq} with the correspondence \eqref{varepsilon-def} and the appropriate definitions of amplitudes $A$ and $\widetilde{A}$. 
The coefficient $f_{B,3,3,-1,2}$ is constant proportional to $\chi/\lambda_1''(\ell_0)$.
\end{remark}

\begin{remark}
	The coefficient in front of the second derivative in the amplitude equation 
	\eqref{amplitude-nonlinear} comes from an expansion of the imaginary parts
	of the spectral curves at the spectral gaps, see Fig. \ref{fig:FMs}(c). The coefficient in front of the second derivative
	changes sign at every spectral boundary. Since the
	coefficient in front of the
	cubic coefficient in \eqref{amplitude-nonlinear} does not change
	sign, the homoclinic and heteroclinic solutions exist as bifurcating solutions at every spectral gap associated with the anti-periodic eigenfunctions. In other words, if $\chi < 0$, then we pick 
	the bifurcating mode at \eqref{bifurcating-mode}  with 
	$\lambda_1''(\ell_0) < 0$ and $\delta = +1$. If $\chi > 0$, 
	then we pick the bifurcating mode at \eqref{bifurcating-mode-alternative} 
	with $\lambda_2''(\ell_0) > 0$ and  $\delta = -1$.
\end{remark}

\section{Comparison with numerical simulations}
\label{secdisc}

We now conduct a number of numerical simulations to illustrate the main results of the paper. We start with the simplest case, and work up in complexity.

Equation \eqref{amplitude-nonlinear} with $\tilde{c}=0$ and $\delta = -{\rm sign}(\lambda_1''(\ell_0)) = -{\rm sign}(\chi)$ has the following homoclinic solution
\begin{equation}
    A(\tau) = \sqrt{ \frac{2}{|\chi|} } \sech\left(\sqrt{\frac{2}{|\lambda''_1(\ell_0)|}} \tau\right)
    \label{A-soliton}
\end{equation}
See Fig.~\ref{fig:homoclinic}(a) for an example plot of this solution in the $(A,A')$
phase plane. Returning to ansatz Eq.~\eqref{eq:nls_ansatz}, we have the following leading order approximation 
in terms of the original lattice variables,
\begin{equation}\label{eq:soliton}
    u_n(t) =   \epsilon \sqrt{ \frac{2}{|\chi|} } \sech\left(\sqrt{\frac{2}{|\lambda''_1(\ell_0)|}} \varepsilon t \right)   g_1(t) \sin(q_{m_0} n).   
\end{equation}

To make practical use of this approximation, 
the first step is to identify
the bifurcation value $k_0$ (namely the critical modulation amplitude values $k_a^0$ and $k_b^0$) and critical mode number $m_0$
such that $\mathrm{trace}(J)=-2$ where
$J$ is the monodromy matrix defined in Eq.~\eqref{eq:monodromy}.
This corresponds
to the bifurcation scenario shown in Fig.~\ref{fig:FMs}. In this case,
the Floquet exponent corresponding to $m_0$ will be purely imaginary and will be of the form
$i \ell_0 = \frac{i \pi}{T} $. The corresponding Bloch mode $e^{i \ell_0 t} f_1(\ell_0,t)$ is obtained by solving Eq.~\eqref{Bloch-1}. This can be done
explicitly, see Sec.~\ref{sec:Floquet} or the appendix of \cite{chong2}
for details. Next, we compute $\lambda''_1(\ell_0)$
using Eq.~\eqref{projection-2} with $j=1$ which depends only on $f_1$
and its derivatives, which can also be computed explicitly. Equivalently,
one can determine $\lambda''_1(\ell_0)$ using Eq.~\eqref{lambda_pp_again}.

\subsection{Examples with $c=0$ and $K_3=0$}

The nonlinear coefficient $\chi$ can be computed from $g_1$ if $K_3 = 0$,
i.e., if there is no quadratic nonlinearity (the case of $K_3 \neq 0$ is discussed below). To compute $\chi$ in this case, we substitute $g_1(t)= 2 {\rm Re}\left( e^{i \ell_0 t} f_1(\ell_0,t) \right) $ into Eq.~\eqref{eq:chi} and evaluate.

As our first example, we chose parameters that correspond to the spectral
picture in Fig.~\ref{fig:FMs} and $K_3=0$ and $K_4 = -0.8$.
In this case the critical mode
$m_0 = 3$ lies at the top of the first spectral band, namely $(\ell,\lambda) = (\pi/T,\mu_1)$.
This demonstrates that the spectral condition \textbf{(Spec)} is satisfied 
and we choose $\delta = +1$. It can be seen from Fig.~\ref{fig:FMs}(c), or via direct calculation,
that $\lambda''(\ell_0) < 0$. By choosing $K_4< 0 $, we have that
$\chi<0$, and thus \textbf{(Coeff)} is satisfied. 

\begin{figure}[htbp] 
	\centerline{
		\begin{tabular}{@{}p{0.33\linewidth}@{}p{0.33\linewidth}@{}p{0.33\linewidth}@{} }
			\rlap{\hspace*{5pt}\raisebox{\dimexpr\ht1-.1\baselineskip}{\bf (a)}}
			\includegraphics[height=4cm]{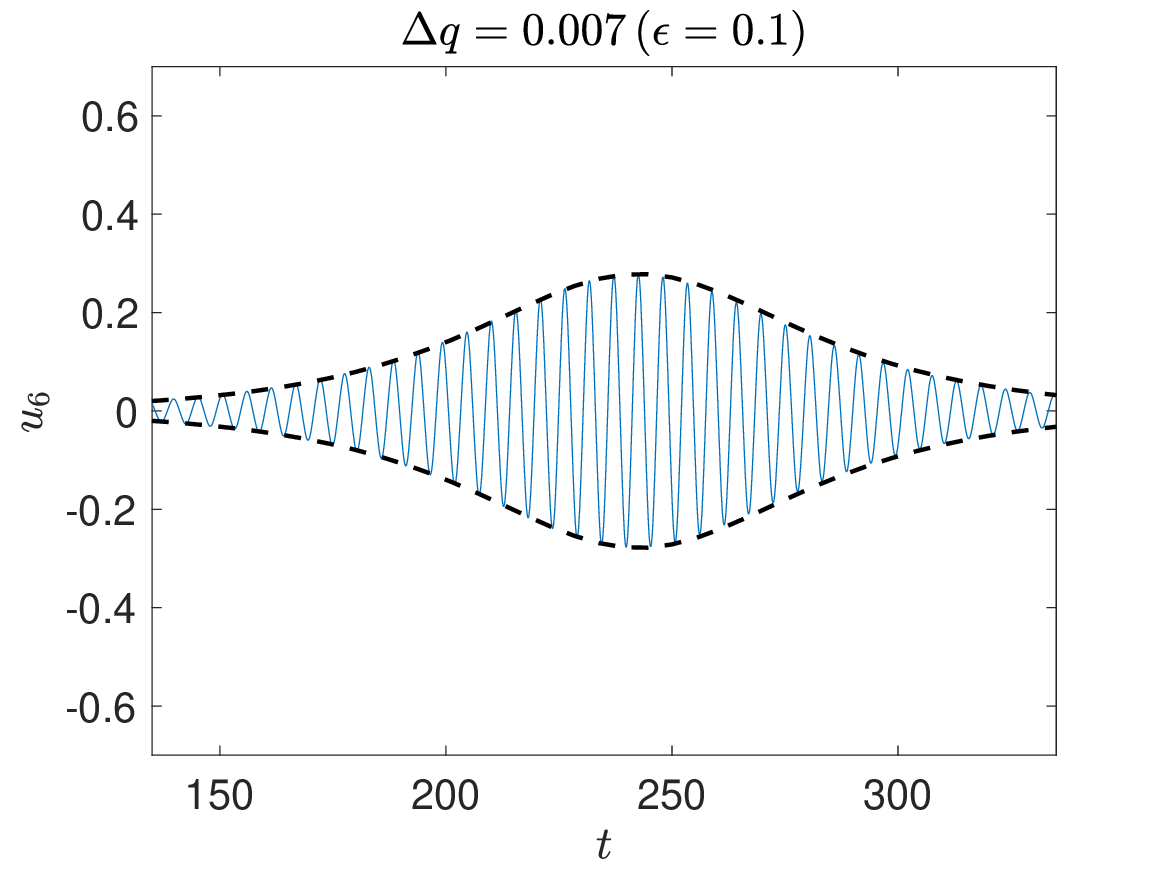}  &
			\rlap{\hspace*{5pt}\raisebox{\dimexpr\ht1-.1\baselineskip}{\bf (b)}}
			\includegraphics[height=4cm]{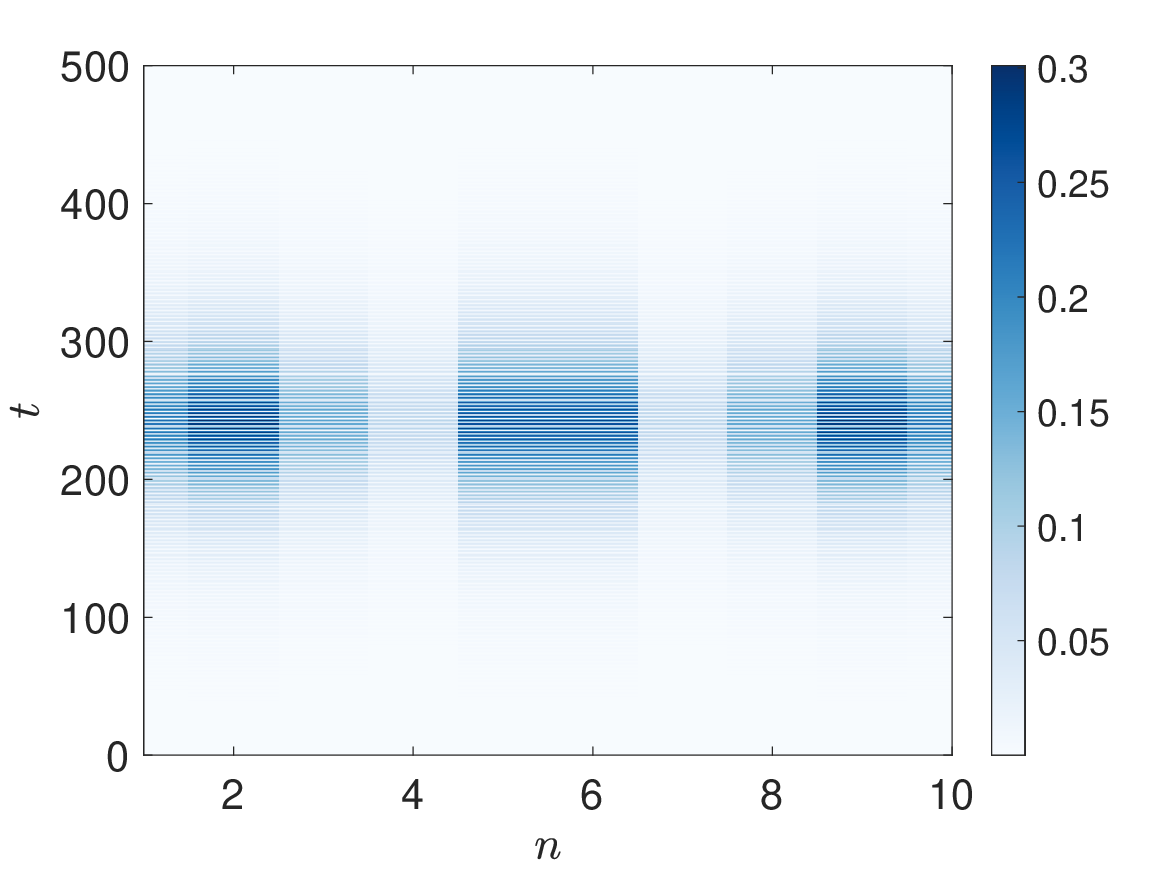} &
			\rlap{\hspace*{5pt}\raisebox{\dimexpr\ht1-.1\baselineskip}{\bf (c)}}
			\includegraphics[height=4cm]{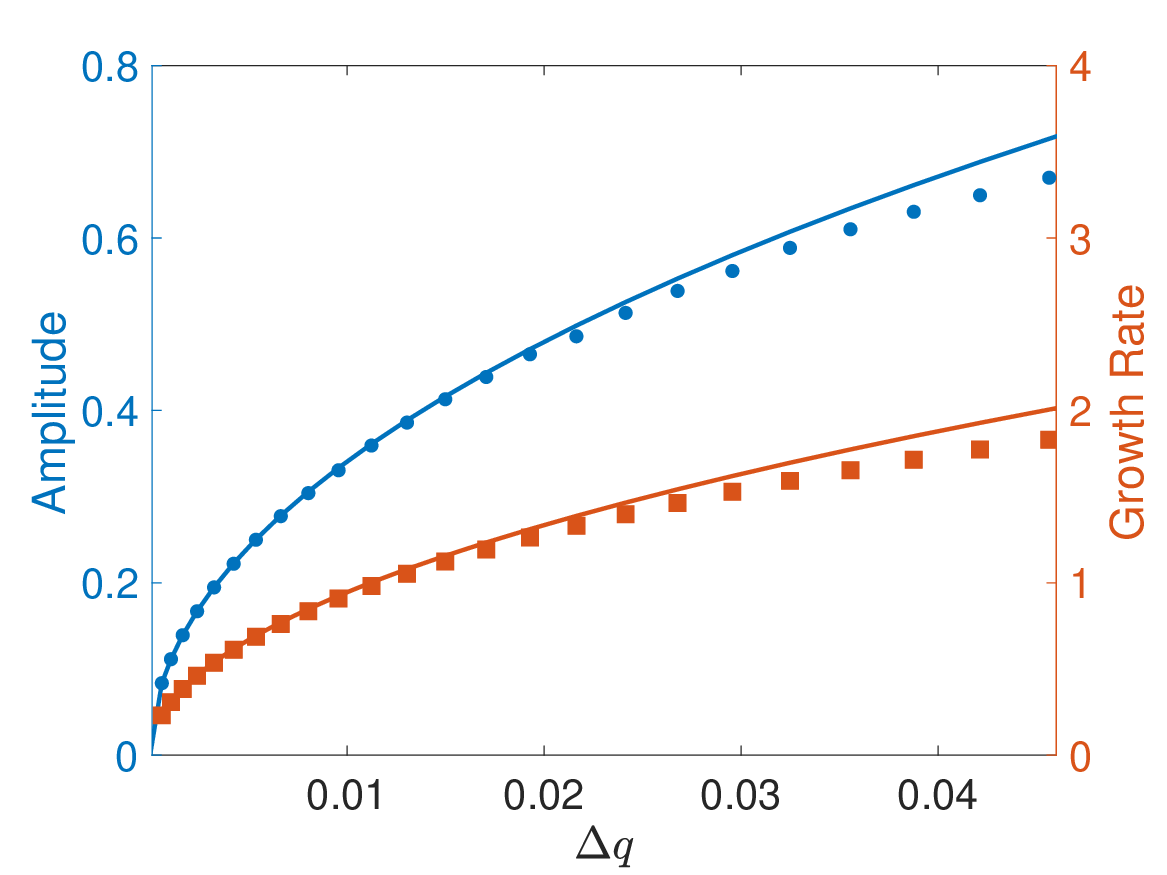} 
		\end{tabular}
	}
	
	\caption{A generalized $q$-gap breather bifurcating from $\mu_1$. The parameter values are $c=0$, $T=1/0.37$, $K_3=0$, $K_4 = -0.8$ and $K_2 = \underline{m} = 1$. The critical modulation amplitude parameters are $k_a^0 = 0.5$ and $k_b^0 =0.79$. The critical mode number 
		is $m_0=3$, which lies at the top of the first spectral band, namely $\lambda(q_3) = \mu_1$, (see Fig.~\ref{fig:FMs}(a)).
		\textbf{(a)} Numerical simulation with initial value $u_n(0) = 10^{-4} \sin(q_3 n)$ and $\epsilon = 0.1$. The displacement of the 6th particle is shown as a function of time. The modulation amplitude is $k_a = k_a^0 + \epsilon^2 = 0.06$, $k_b = k_b^0 + \epsilon^2 = 0.8$. The dashed line shows an approximation of the envelope given by Eq.~\eqref{eq:soliton}
		with $\chi =  -0.6530$, $\lambda''(\ell_0) = -21.0222$ and $\epsilon = 0.1$.
		For this value of $\epsilon$, the distance of the wavenumber to the edge of
		the gap is $\Delta q = q_3 - q_\ell = 0.007$.
		\textbf{(b)} Intensity plot of the solution shown in panel (a). Color intensity corresponds to $|u_n|$.
		\textbf{(c)} Plot
		of the amplitude of the breather ($\mathrm{max}_t u_6(t)$) for the numerical simulation (blue dots)
		and prediction based on Eq.~\eqref{eq:soliton} (blue line) as a function of $\Delta q$.
		The real part of the Floquet exponent corresponding to $q_3$ (solid red squares) and asymptotic approximation \eqref{varepsilon-again} (red line) are also shown, which indicate the growth (decay) rate of the breather.
	}
	\label{fig:simulation_K30}
\end{figure}

To generate the generalized $q$-gap breather, we keep all parameters fixed,
but select $\epsilon=0.1$ and $k_a = k_a^0 + \epsilon^2 = 0.06$, $k_b = k_b^0 + \epsilon^2 = 0.8$. With these parameter values, the $m_0=3$ mode lies
in the spectral gap. The corresponding Floquet multipliers and exponents and are shown
in Fig.~\ref{fig:breather_idea}(a,b). We initialize the numerical simulation with
$$
u_n(0) = 10^{-4}\sin(q_3 n) \quad \mbox{\rm and} \quad \dot{u}_n=0.
$$ 
For initial data
with such small amplitude, the dynamics will initially be nearly linear,
and hence the solution will grow exponentially with rate given
by Re($\gamma$), which is the real part the Floquet exponent
associated to mode $m=3$ (see the larger black dot in Fig.~\ref{fig:breather_idea}(a)).
According to Eq.~\eqref{amplitude-linear}, an approximation of this growth rate
is given by \eqref{varepsilon-again} with $\epsilon = 0.1$. As the amplitude increases in the dynamic evolution, the affect of the
nonlinearity comes into play, which will cause the solution
to experience decay, such that the resulting waveform is localized
in time. The time series of the $u_6(t)$ node is shown in Fig.~\ref{fig:simulation_K30}(a). The temporal localization occurs
uniformly throughout the lattice, as seen in the intensity plot of
Fig.~\ref{fig:simulation_K30}(b). By construction, the wavenumber $q_3$
lies in a wavenumber bandgap. Thus, the solution shown in Fig.~\ref{fig:simulation_K30}(a,b) is a generalized $q$-gap breather. The analytical prediction based on Eq.~\eqref{eq:soliton}
is shown as the dashed-line in panel (a). For the sake of clarity,
only the envelope of the approximation is shown, which is simply
a plot of the local maximums (and minimums) of Eq.~\eqref{eq:soliton}.

\begin{figure}[htbp]
	\centerline{
		\begin{tabular}{@{}p{0.5\linewidth}@{}p{0.5\linewidth}@{}}
			\rlap{\hspace*{5pt}\raisebox{\dimexpr\ht1-.1\baselineskip}{\bf (a)}}
			\includegraphics[height=6cm]{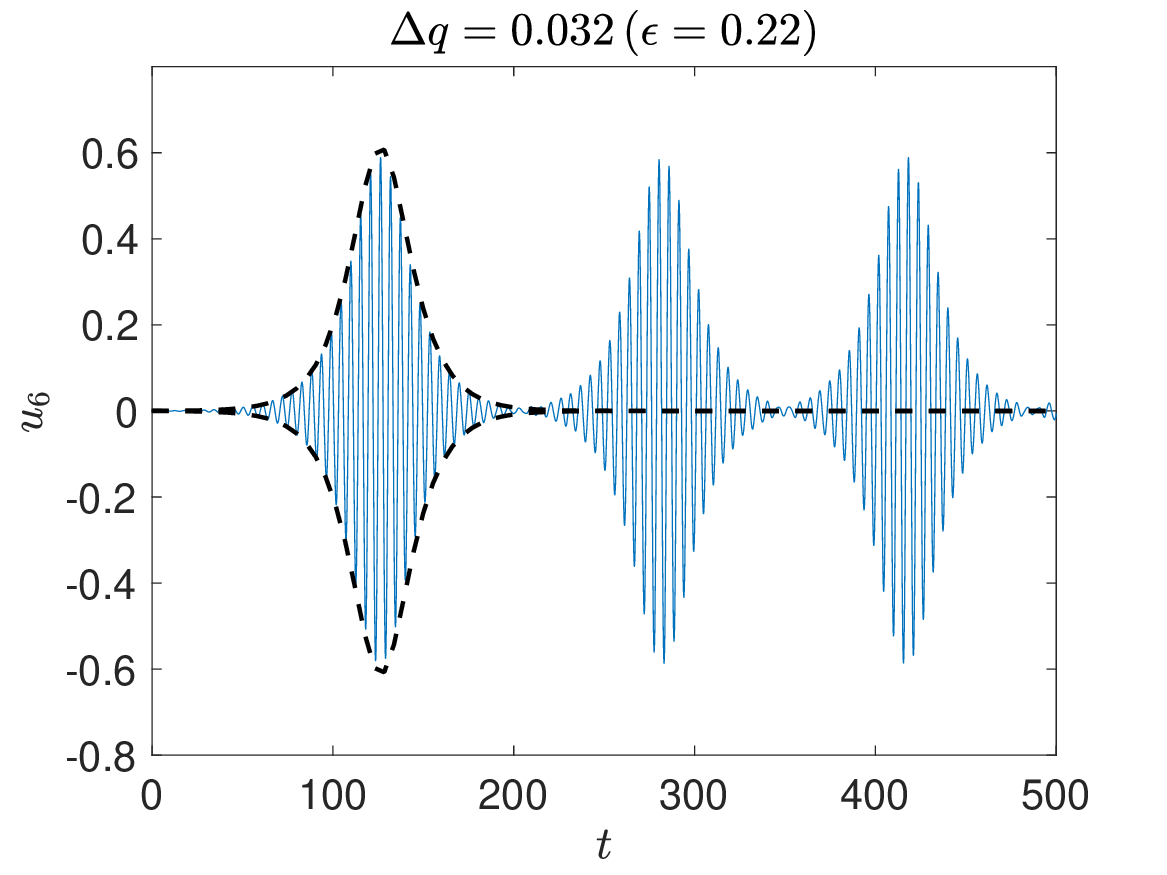}  &
			\rlap{\hspace*{5pt}\raisebox{\dimexpr\ht1-.1\baselineskip}{\bf (b)}}
			\includegraphics[height=6cm]{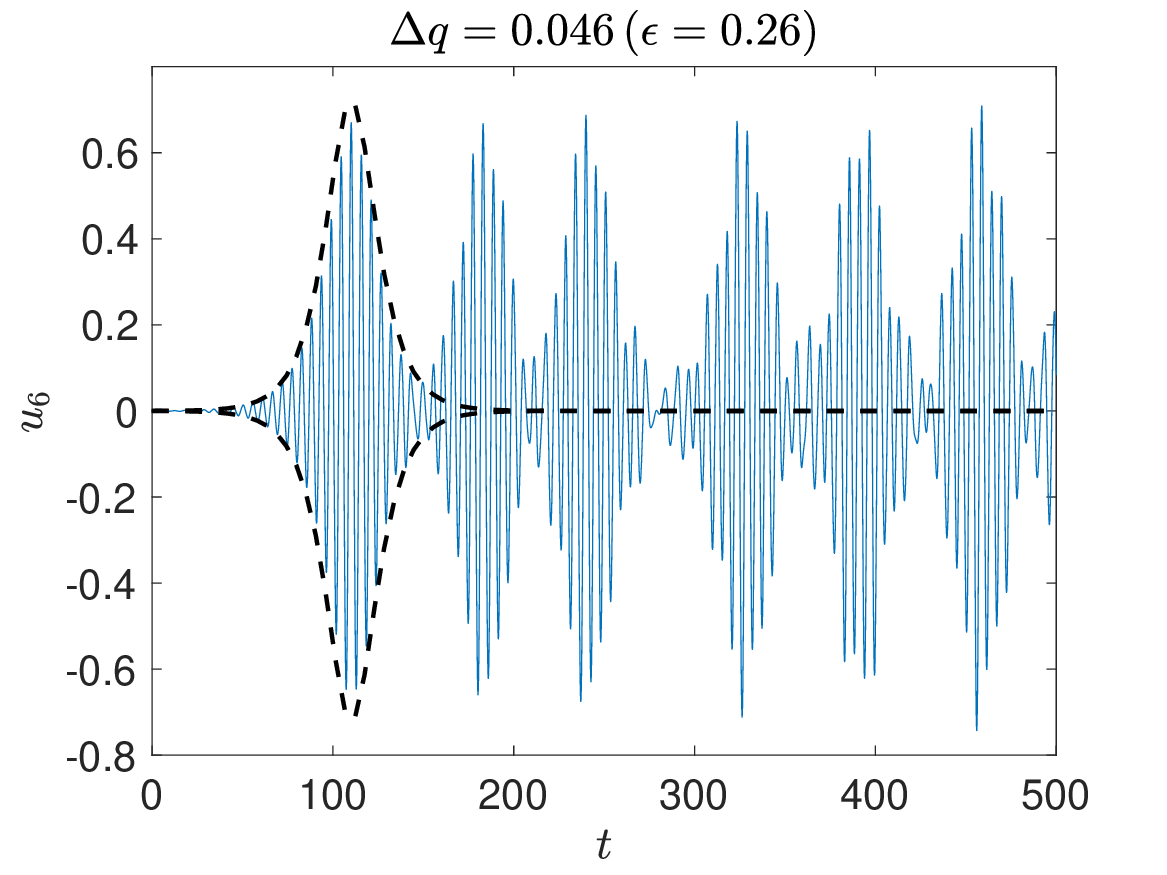} 
		\end{tabular}
	}
	\caption{Numerical simulation with all parameters as in
		Fig.~\ref{fig:simulation_K30} but larger values of $\epsilon$. 
		\textbf{(a)} With $\Delta q = 0.032$ ($\epsilon = 0.22)$
		a temporally localized structure is initially observed, but
		over a longer time interval a recurrence of breathing patterns emerges.  
		The dashed line shows an approximation of the envelope given by Eq.~\eqref{eq:soliton}.
		\textbf{(b)} With $\Delta q = 0.046$ ($\epsilon = 0.26)$ the second breather like structure
		emerges before the temporal localization of the first pulse can be achieved.
		The envelope approximation given by Eq.~\eqref{eq:soliton} is still quite good
		until the breakdown of the temporal localization.
	}
	\label{fig:longtime}
\end{figure}

For $\epsilon=0.1$, the distance of the wavenumber to the edge of the gap is $\Delta q = q_\ell - q_3 = 0.007$, where $q_\ell$ is the wavenumber
at the (left) edge of the gap. Recall from Eq.~\eqref{eq:Dq} that $\Delta q = \mathcal{O}(\epsilon^2)$.
Since the amplitude of the breather is $\mathcal{O}(\epsilon)$, see Eq.~\eqref{eq:soliton},
the amplitude grows like $\mathcal{O}(\sqrt{\Delta q})$. This observation was made numerically
and experimentally in \cite{chong2}, which we have now proved. 
This amplitude trend is consistent with the trend found for discrete breathers in space-periodic systems where it is well known that the breather amplitude grows like $\mathcal{O}\left(\sqrt{\Delta \omega}\right)$,
 where $\Delta \omega$ is the difference between the breather frequency 
 and the edge of the frequency spectrum \cite{flach_discrete_2008}.
Using \eqref{eq:Dq} and Eq.~\eqref{eq:soliton}
allows us to obtain an analytical prediction of the breather amplitude dependence
of the distance to the band edge, see Fig.~\ref{fig:simulation_K30}(c).
For small values of $\Delta q$ (and hence $\epsilon$) the agreement is very good.
The growth parameter gives an indication of how wide or narrow the breather will be,
with larger growth parameters corresponding to narrower solutions. The  prediction from the linear theory
is given by the real part of the Floquet exponent corresponding to mode $m_0$, and the
approximation from the perturbation analysis is \eqref{varepsilon-again}. 
A comparison of these two quantities are shown as the red markers and lines, respectively
of Fig.~\ref{fig:simulation_K30}(c). The trends in Fig.~\ref{fig:simulation_K30}(c)
demonstrate that the $q$-gap breathers becomes larger in amplitude and more narrow
as the wavenumber goes deeper into the gap.

In Fig.~\ref{fig:simulation_K30}(c) the $q$-gap breathers are generated
up until $\Delta q = 0.042$ (which corresponds to $\epsilon = 0.25$).
For $\Delta q = 0$, the width of the wavenumber bandgap is
$q_r - q_\ell \approx 0.28$ (where $q_r$, $q_\ell$ are the right
and left edges of the bandgap, respectively). Thus, the branch of solutions
shown Fig.~\ref{fig:simulation_K30}(c) extends to roughly 15\% of the width of the bandgap.
For $\Delta q > 0.042$ we did not observe a coherent temporal localization. 
Indeed, for all the breathers observed numerically, the localization is obtained for a finite interval of time. For longer time simulations the amplitude of the breather can grow again (leading to a repeated appearance of breathers), see Fig.~\ref{fig:longtime}(a). Similar observations
have been made for k-gap solitons in photonic systems \cite{Superluminal}.
While Theorem \ref{theorem-1} guarantees that a temporally localized structures exists
over a finite temporal interval, there is no statement about the dynamics beyond
this interval. The numerical simulations suggest the tail of the breather
can experience repeated growth. We observed as $\Delta q$ becomes larger,
the time between consecutive peaks of the pulses becomes smaller. In other words,
the emergence of the ``second" breather occurs faster as $\Delta q$ becomes larger.
Thus, for sufficiently large $\Delta q$ the structure is not temporally localized
since the second breather emerges ``too soon", see Fig.~\ref{fig:longtime}(b). This is
the reason why we only show $\Delta q \leq 0.042$ in Fig.~\ref{fig:simulation_K30}(c).

\begin{figure}[htbp]
  \centerline{
   \begin{tabular}{@{}p{0.33\linewidth}@{}p{0.33\linewidth}@{}p{0.33\linewidth}@{} }
     \rlap{\hspace*{5pt}\raisebox{\dimexpr\ht1-.1\baselineskip}{\bf (a)}}
 \includegraphics[height=4cm]{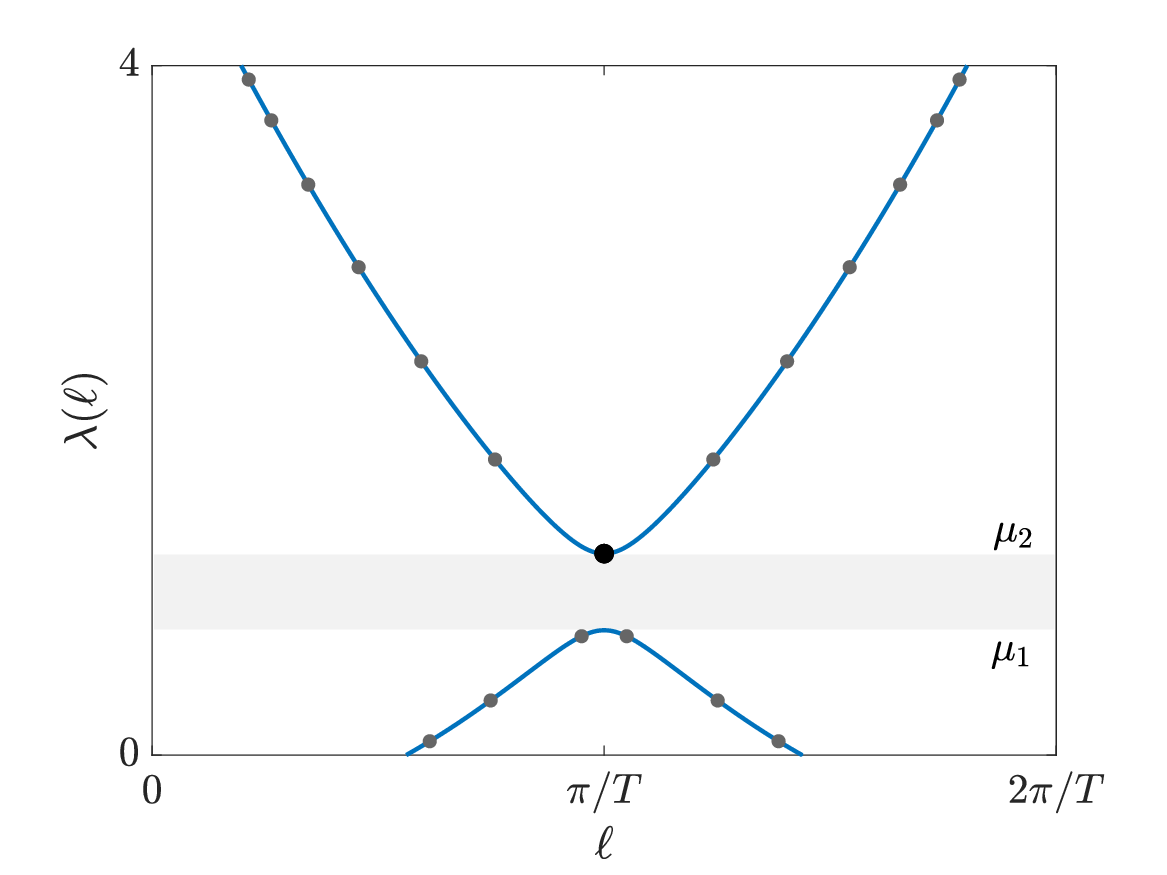}  &
  \rlap{\hspace*{5pt}\raisebox{\dimexpr\ht1-.1\baselineskip}{\bf (b)}}
 \includegraphics[height=4cm]{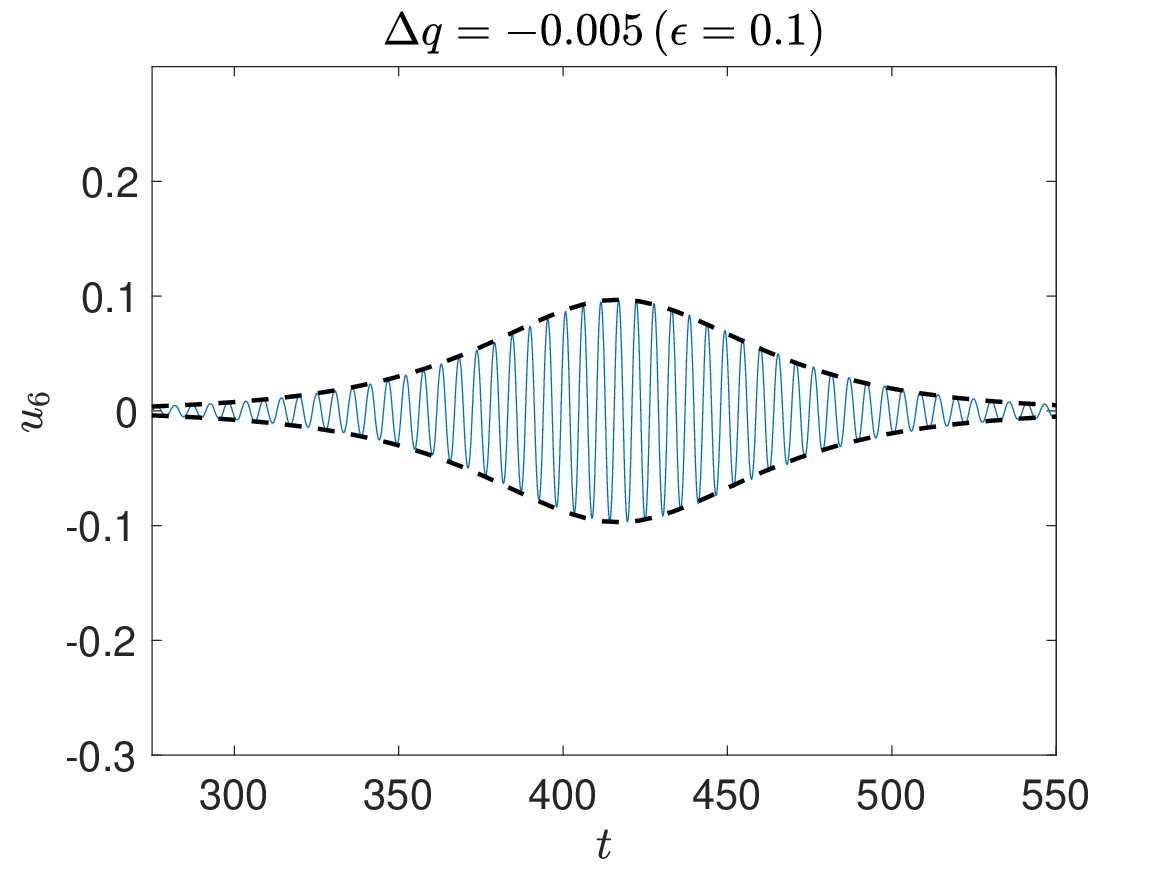} &
   \rlap{\hspace*{5pt}\raisebox{\dimexpr\ht1-.1\baselineskip}{\bf (c)}}
 \includegraphics[height=4cm]{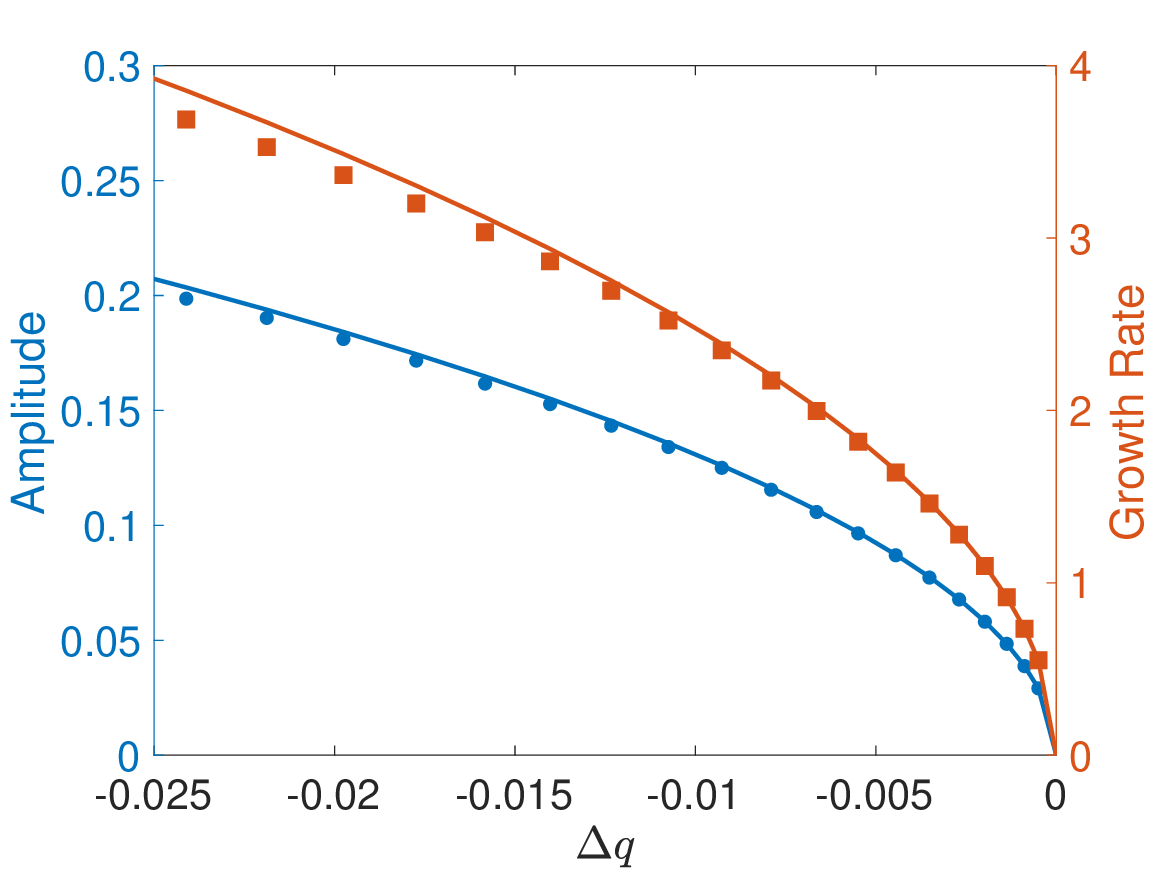} 
  \end{tabular}
  }

   \caption{\textbf{(a)} Spectral bands (blue curves) of the Schr\"odinger operator as a function of $\ell$, the imaginary part
   of the Floquet exponent. The gray dots shown the corresponding values
   in the finite lattice with $N=10$. The parameter values are $c=0$, $T=1/0.37$ and $K_2 = \underline{m} = 1$. The critical modulation amplitude parameters are $k_a^0 = 0.5$ and $k_b^0 =0.7487$. The critical mode number 
   is $m_0=4$, which lies at the bottom of the second spectral band, namely $\lambda(\ell(q_4)) = \mu_2$ (see larger black marker).
    The eigenvalues $\mu_1$ and $\mu_2$ define
   the edges of the band gap, shown as the gray shaded region. 
   \textbf{(b)} Example $q$-gap breather bifurcating from $\mu_2$ with nonlinear coefficients $K_3=0$ and $K_4 = 0.8$.
  The initial value is $u_n(0) = 10^{-4} \sin(q_4 n)$ and $\epsilon = 0.1$. The displacement of the 6th particle is shown as a function of time. The modulation amplitude is $k_a = k_a^0 + \epsilon^2 = 0.06$, $k_b = k_b^0 + \epsilon^2 = 0.7587$. The dashed line shows an approximation of the envelope given by Eq.~\eqref{eq:soliton}
   with $\chi =  1.7707$, $\lambda''(\ell_0) = 26.3821$ and $\epsilon = 0.1$.
For this value of $\epsilon$, the distance of the wavenumber to the edge of
the gap is $\Delta q = q_4 - q_r = - 0.005$.
   \textbf{(c)} Plot
   of the amplitude of the breather ($\mathrm{max}_t u_6(t)$) for the numerical simulation (blue dots)
   and prediction based on Eq.~\eqref{eq:soliton} (blue line) as a function of $\Delta q$.
   The real part of the Floquet exponent corresponding to $q_4$ (solid red squares) and asymptotic approximation
   $\epsilon \sqrt{\delta \lambda''(\ell_0)/2} $ (red line) are also shown, which indicate the growth (decay) rate of the q-gap breather.
  }
   \label{fig:simulation_K30_hard}
\end{figure}

Next, we consider an example where the breathers bifurcate from $\mu_2$. The spectral bands
corresponding to $c=0$, $T=1/0.37$, $K_2 = \underline{m} = 1$, $k_a^0 = 0.5$ and $k_b^0 =0.7487$
are shown in Fig.~\ref{fig:simulation_K30_hard}(a). For these parameter values
the mode $m_0=4$ lies at the bottom of the second spectral band, namely $\lambda(q_4) = \mu_2$.
Notice that the concavity of the second spectral band is opposite of the first band,
namely $\lambda_2''(\ell_0) > 0 $. Thus, in order to satisfy the \textbf{(Coeff)} condition,
we require $\chi > 0$. If $K_3 = 0$ this implies that $K_4 > 0$ and the sign parameter
is now $\delta = -1.$ Thus, for the next numerical simulations,
we fix $K_3=0$ and $K_4 = 0.8.$ The approximation given in Eq.~\eqref{eq:soliton}
is identical in this case, but we replace $\lambda_1''(\ell)$ with $\lambda_2''(\ell)$, and likewise for the underlying Bloch modes
(where $f_1$ should be replaced by $f_2$, etc.).

Figure~\ref{fig:simulation_K30_hard}(b) shows
an example of the generalized $q$-gap breather with $\epsilon = 0.1$, with corresponding envelope
prediction given by Eq.~\eqref{eq:soliton}. Qualitatively, the results are similar to the
example shown in Fig.~\ref{fig:simulation_K30}(a). Figure~\ref{fig:simulation_K30_hard}(c)
shows the dependence of the breather amplitude and growth rate on the parameter
$\Delta q = q_r - q_{m_0}$. Note, since the breather is bifurcating from the right
edge of the wavenumber bandgap, the quantity $\Delta q$ will be negative.
The breather amplitude grows like $\mathcal{O}(\sqrt{|\Delta q|})$.

\subsection{Examples with $c=0$ and $K_3\neq0$}

Here we will consider $K_3 \neq0$. In particular, we will chose
values of the nonlinear coefficients in \eqref{Taylor} that correspond to the modulated magnetic lattice
described in Sec.~\ref{sec:model} so that the results obtained here are directly relevant
for the experimental set-up described in \cite{chong2}. In the re-scaled variables the
interaction coefficients are $K_2= K_3 = 1$ and $K_4 = 0.8$. Since $K_3 \neq0$ the sign of $\chi$ must be computed
directly to see if the relevant eigenvalue to bifurcate from is $\mu_1$ or $\mu_2$.
$\chi$ will depend on the function $h_2(t)$,
which we can obtain by solving Eq.~\eqref{eq:h2}. It will be convenient to
estimate $h_2(t)$ numerically under the constraint that $h_2(t)$ is $T$ periodic,
which we achieve using a shooting method. In particular, we apply Newton iterations
on the map $F(h^0) =h_2(0;h^0) - h_2(T;h^0)$ where $h_2(t;h^0)$ is the solution of
Eq.~\eqref{eq:h2} with initial condition $h^0 = (h_2(0), \dot{h}_2(0) )^T$.
The Jacobian of the map $F$ is simply
$I-V(T)$, where $I$ is the 2x2 identity matrix and $V(T)$ is the solution to
the variational equation $\dot{V} = \frac{df}{dh} V$ with initial value $V(0) = I$ where $\frac{df}{dh}$ is the Jacobian corresponding to Eq.~\eqref{eq:h2} \cite{Smale}.

\begin{figure}[htbp]
	\centerline{
		\begin{tabular}{@{}p{0.33\linewidth}@{}p{0.33\linewidth}@{}p{0.33\linewidth}@{} }
			\rlap{\hspace*{5pt}\raisebox{\dimexpr\ht1-.1\baselineskip}{\bf (a)}}
			\includegraphics[height=4cm]{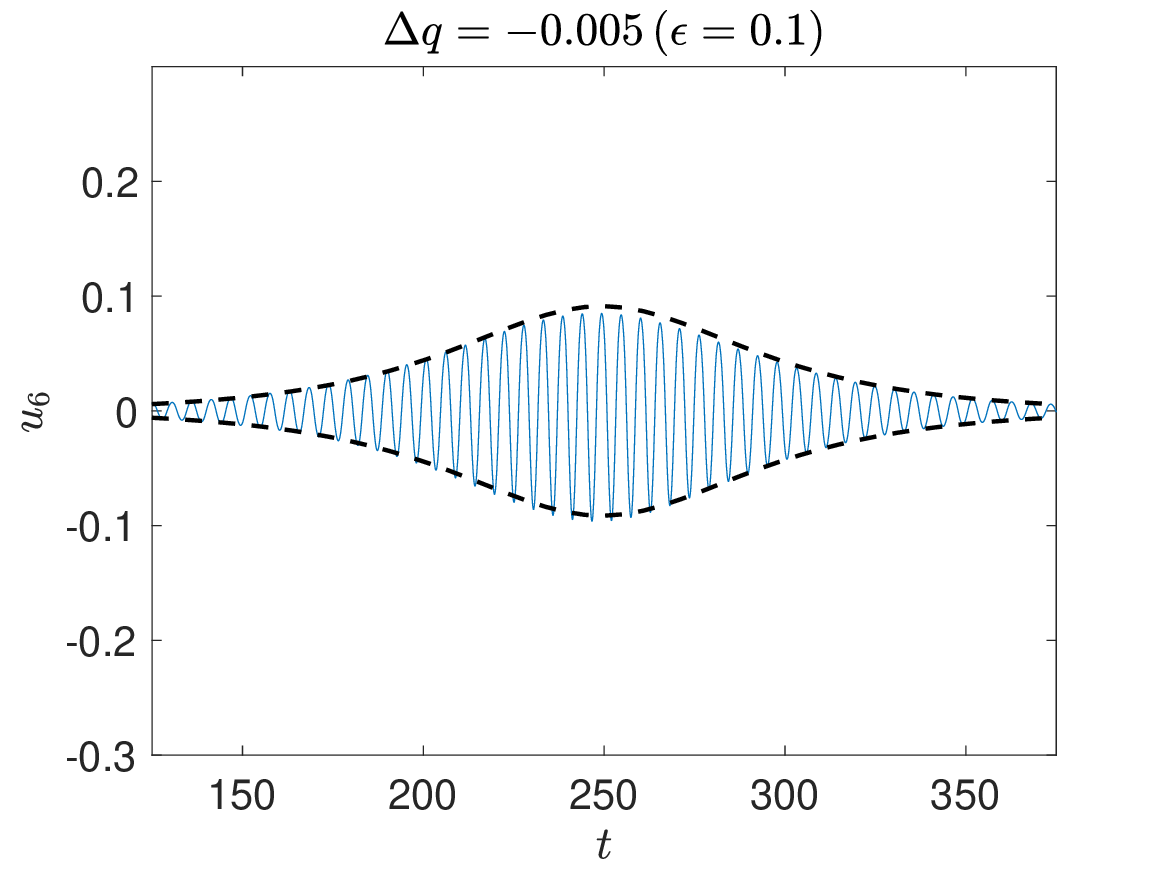}  &
			\rlap{\hspace*{5pt}\raisebox{\dimexpr\ht1-.1\baselineskip}{\bf (b)}}
			\includegraphics[height=4cm]{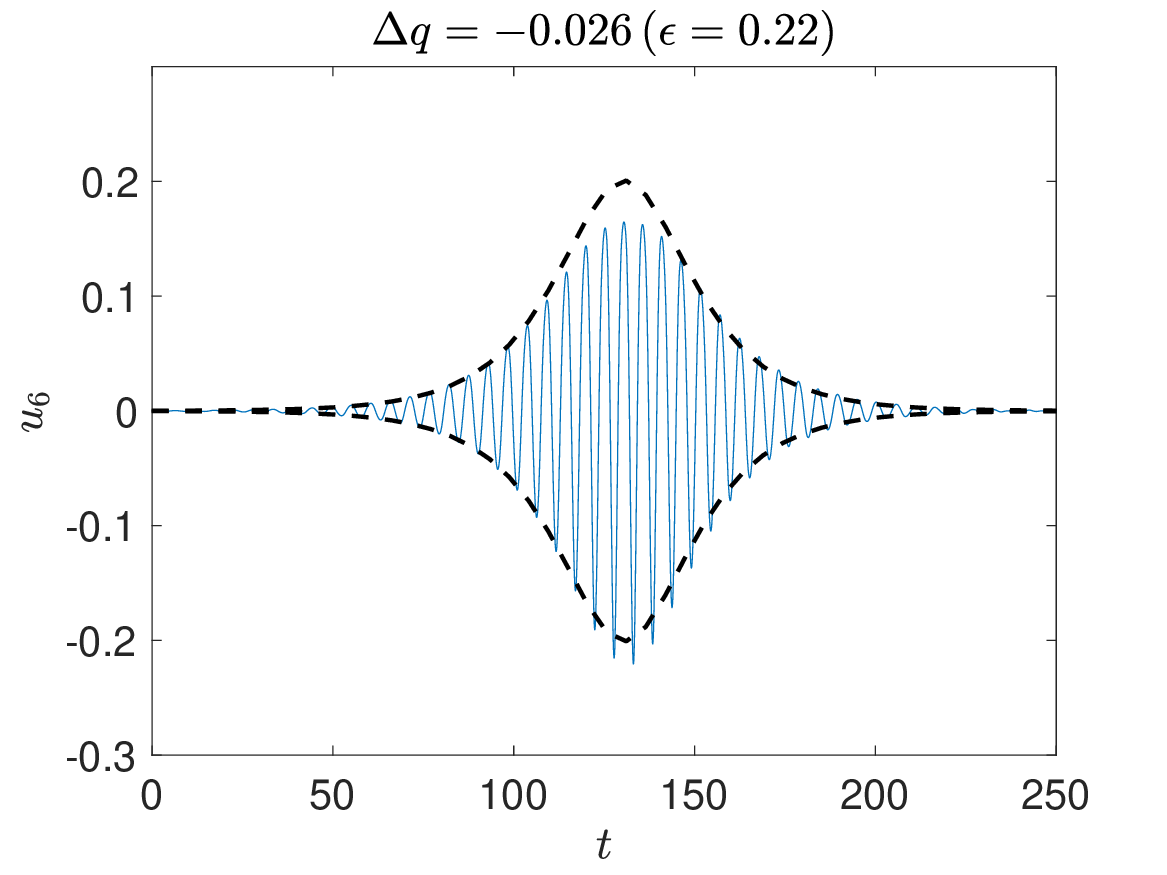} &
			\rlap{\hspace*{5pt}\raisebox{\dimexpr\ht1-.1\baselineskip}{\bf (c)}}
			\includegraphics[height=4cm]{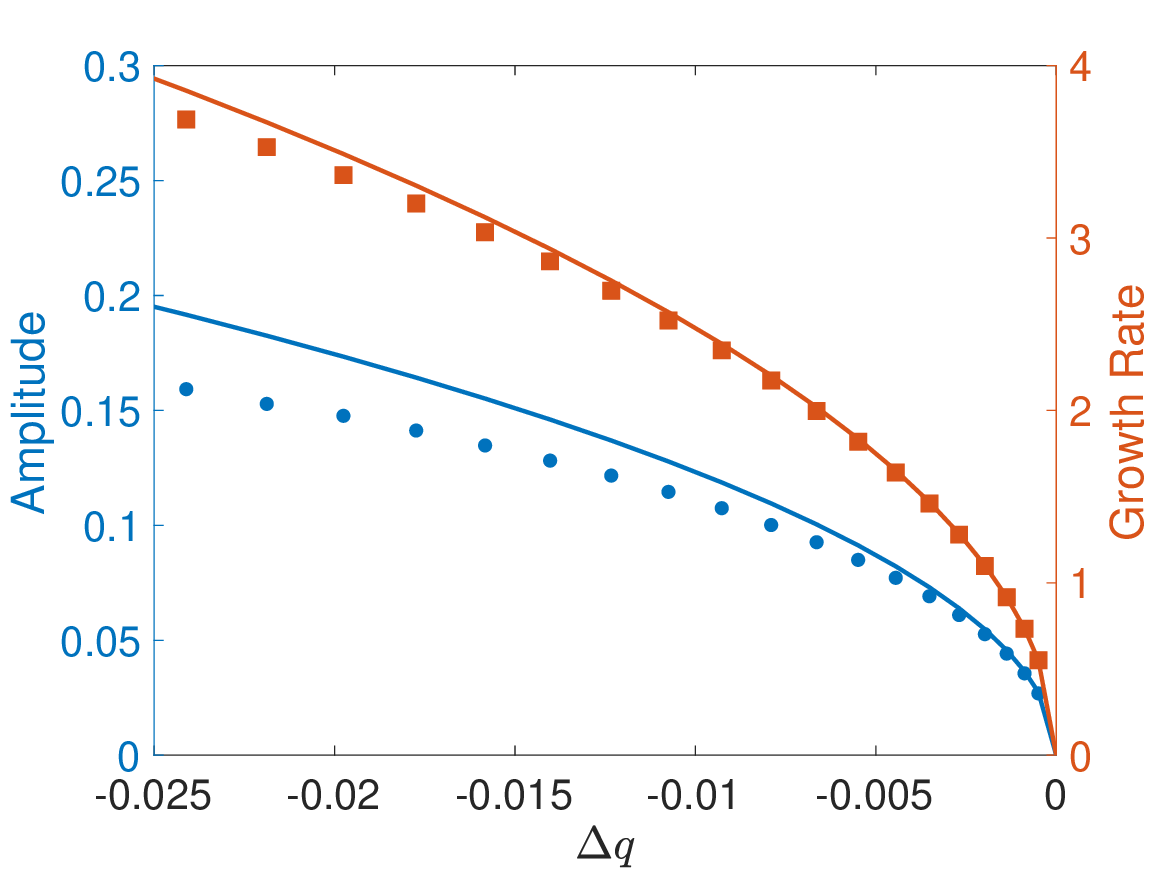} 
		\end{tabular}
	}
	
	\caption{\textbf{(a)} A generalized $q$-gap breather bifurcating from $\mu_2$ with nonlinear coefficients 
		corresponding to a magnetic lattice, namely $K_3=1$ and $K_4 = 0.8$. The other parameter values
		and spectral picture is identical to Fig.~\ref{fig:simulation_K30_hard}(a).
		The initial value is $u_n(0) = 10^{-4} \sin(q_4 n)$ and $\epsilon = 0.1$. The displacement of the 6th particle is shown as a function of time. The modulation amplitude is $k_a = k_a^0 + \epsilon^2 = 0.06$, $k_b = k_b^0 + \epsilon^2 = 0.7587$. The dashed line shows an approximation of the envelope given by Eq.~\eqref{eq:soliton}
		with $\chi = 1.997$, $\lambda''(\ell_0) = 26.3821$ and $\epsilon = 0.1$.
		For this value of $\epsilon$, the distance of the wavenumber to the edge of
		the gap is $\Delta q = q_4 - q_r = - 0.005$.
		\textbf{(b)} Same as panel (a) with $\epsilon = 0.22$. 
		\textbf{(c)} Plot
		of the amplitude of the breather ($\mathrm{max}_t u_6(t)$) for the numerical simulation (blue dots)
		and prediction based on Eq.~\eqref{eq:soliton} (blue line) as a function of $\Delta q$.
		The real part of the Floquet exponent corresponding to $q_4$ (solid red squares) and asymptotic approximation \eqref{varepsilon-again} (red line) are also shown, which indicate the growth (decay) rate of the breather.
	}
	\label{fig:simulation_magnetic}
\end{figure}

In this example, we consider the spectral situation as shown in Fig.~\ref{fig:simulation_K30_hard}(a), such that $m_0=4$ is the  critical mode
bifurcating from $\mu_2$. In this case, $\lambda''(\ell_0) > 0$, so we chose
$\delta=-1$ and we must have $\chi > 0$ to satisfy (\textbf{Coeff}).
Upon computing $h_2(t)$ with the shooting method and substituting into Eq.~\eqref{eq:chi} with $K_3 =1$ and $K_4 = 0.8$
we find $\chi = 1.997 > 0$, as desired. 

Figure \ref{fig:simulation_magnetic}(a) shows a numerical simulation of the lattice with initial displacement $u_n(0) = 10^{-4} \sin(q_4 n)$ and $\epsilon = 0.1$, and panel (b) shows a simulation with $\epsilon = 0.22$. By comparing
panels (a) and (b), we see once again that the $q$-gap breather becomes more narrow and larger in amplitude as $\Delta q$ (and thus $\epsilon$) increases in magnitude. What is also apparent,
especially in panel (b), is that the numerical simulation is asymmetric, namely, the maximum
is not simply the minimum reflected about the $u=0$ line. Evidently, the asymmetric nature
of the FPUT potential with $K_3 \neq 0$ is manifested through a lack of reflection symmetry in the $q$-gap breather profile.
Asymmetric breathing profiles are well known in space-periodic FPUT systems with quadratic nonlinearities \cite{Huang2}. The approximation given by Eq.~\eqref{eq:soliton} remains symmetric, however, and thus one would expect the approximation not to do as well
as in the $K_3 = 0$ case. Indeed, inspection of Fig.~\ref{fig:simulation_K30_hard}(c)
confirms this, where the difference in amplitude between simulation and theory is larger
than in the $K_3 = 0$ case (see e.g., Fig.~\ref{fig:simulation_K30_hard}(c)). Nonetheless,
the asymptotic behavior as $\Delta q \rightarrow 0$ is correct, and in particular
the breather amplitude grows like $\mathcal{O}(\sqrt{|\Delta q|})$.

\subsection{Examples with $c\neq0$ and $K_3\neq0$}

In our final example, we include the effect of damping and select $\tilde{c} = 0.1$.
By definition, $c = \epsilon \tilde{c}$, so the critical parameter set (when $\epsilon=0$) will
have $c=0$, like before. Thus, we consider once again the parameter set 
that corresponds to Fig.~\ref{fig:simulation_magnetic}(a). However,
the numerical solutions and asymptotic approximations will have non-zero damping
affect for $\epsilon > 0$. Since the bifurcation scenario is the same
as in Fig.~\ref{fig:simulation_magnetic}(a) the initial condition
for simulations will be of the same form, namely $u_n(0) = 10^{-4} \sin(q_4 n)$.
An example lattice simulation with $\epsilon = 0.1$ is shown in Fig.~\ref{fig:simulation_magnetic_damped}(a)
where the underlying damping constant is $c = \epsilon \tilde{c} = 0.01$. The solution
experiences an initial growth, with growth rate given by the real part
of the $m_0 = 4$ Floquet exponent, but rather then decaying to a near zero amplitude,
like in all the previous examples with $c=0$, the solution approaches steady
periodic motion with period $2T$. 
A longer time evolution of the same solution
is shown in Fig.~\ref{fig:longtime_damped}. In particular, panel (b)
shows the dynamics are essentially periodic for $t$ sufficiently larger.

In terms of the Poincar\'e
map $F_j = U(2 T j )$, the trivial solution $U(t)=0$ is clearly
a fixed point. The $2T$-periodic solution that is approached in the dynamic
simulation is another fixed point. Thus, the solution shown in 
Fig.~\ref{fig:simulation_magnetic_damped}(a) is a transition front, since it
connects two different fixed points.

Another example of the transition front for a larger value of $\epsilon$ is shown in
Fig.~\ref{fig:simulation_magnetic_damped}(b).  Despite the fact
that the structure is not temporally localized, the initial dynamics
still resemble the ``left" side of the $q$-gap breather. Indeed,
the homoclinic approximation from Eq.~\eqref{eq:soliton} is
quite close to the initial front dynamics (see the solid gray line
of Fig.~\ref{fig:simulation_magnetic_damped}(a)). For this reason,
it is still reasonable to measure the amplitude of the front
in the same way we measured the amplitude for the $q$-gap breathers.
A plot of the front amplitude
and real part of the $m_0 = 4$ Floquet exponent is shown in
 \ref{fig:simulation_magnetic_damped}(c). The amplitude trend is similar to the non-damped case,
but the magnitude of the amplitude is smaller, as expected
(compare panel (c) of Figs.~\ref{fig:simulation_magnetic} and \ref{fig:simulation_magnetic_damped}).

\begin{figure}[htbp] 
	\centerline{
		\begin{tabular}{@{}p{0.33\linewidth}@{}p{0.33\linewidth}@{}p{0.33\linewidth}@{} }
			\rlap{\hspace*{5pt}\raisebox{\dimexpr\ht1-.1\baselineskip}{\bf (a)}}
			\includegraphics[height=4cm]{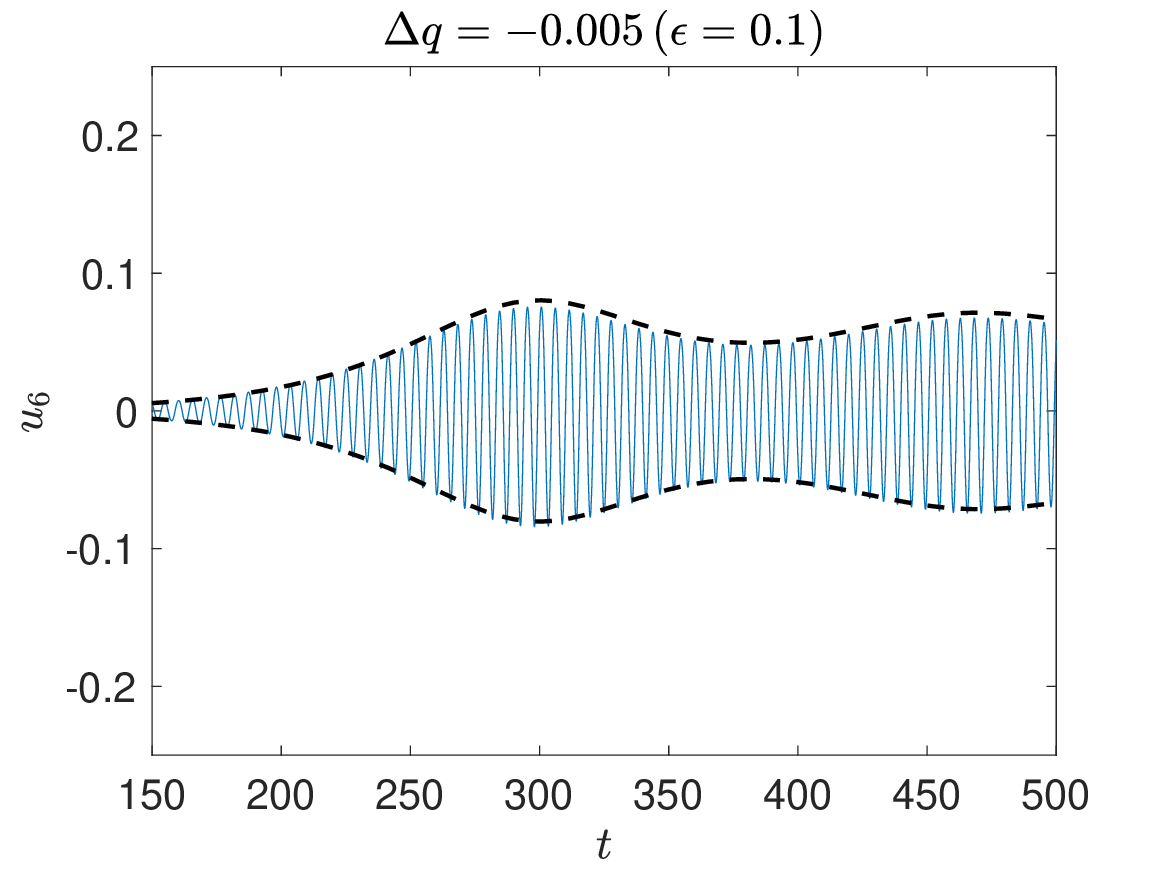}  &
			\rlap{\hspace*{5pt}\raisebox{\dimexpr\ht1-.1\baselineskip}{\bf (b)}}
			\includegraphics[height=4cm]{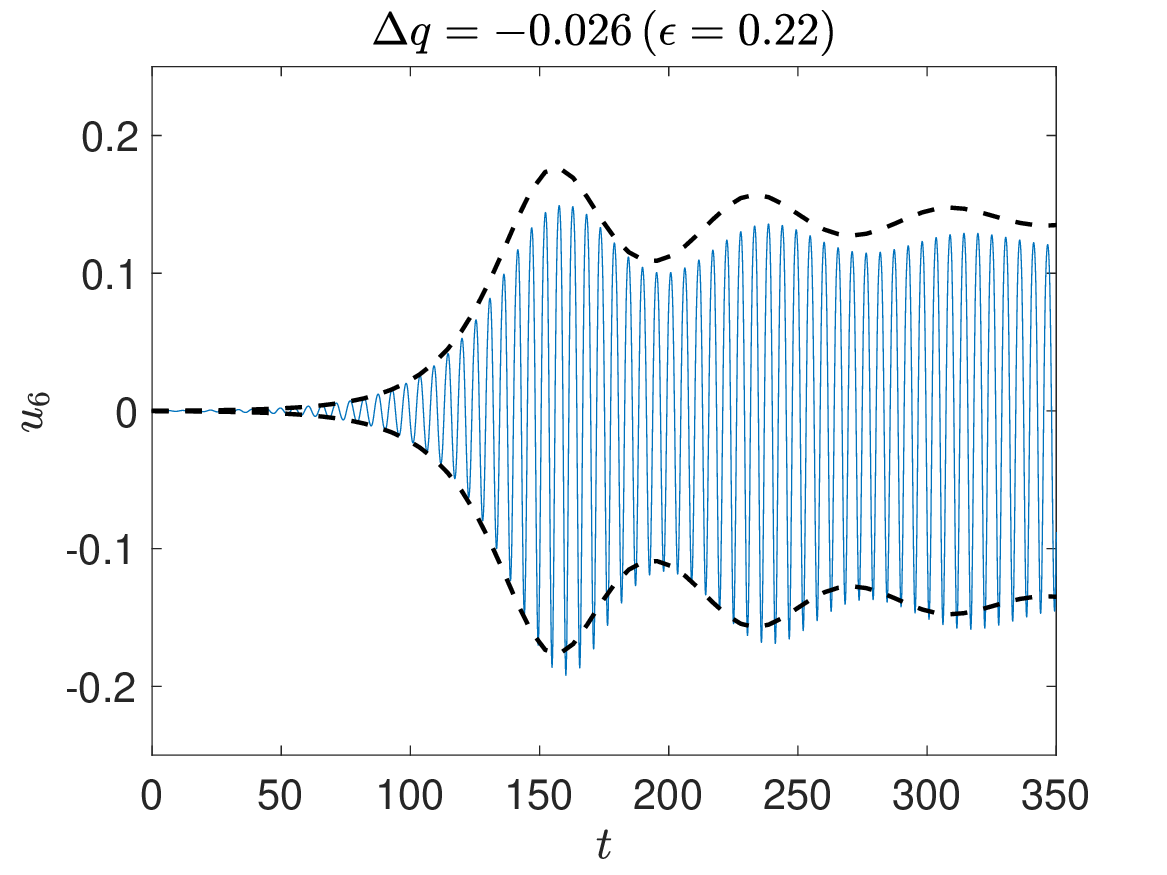} &
			\rlap{\hspace*{5pt}\raisebox{\dimexpr\ht1-.1\baselineskip}{\bf (c)}}
			\includegraphics[height=4cm]{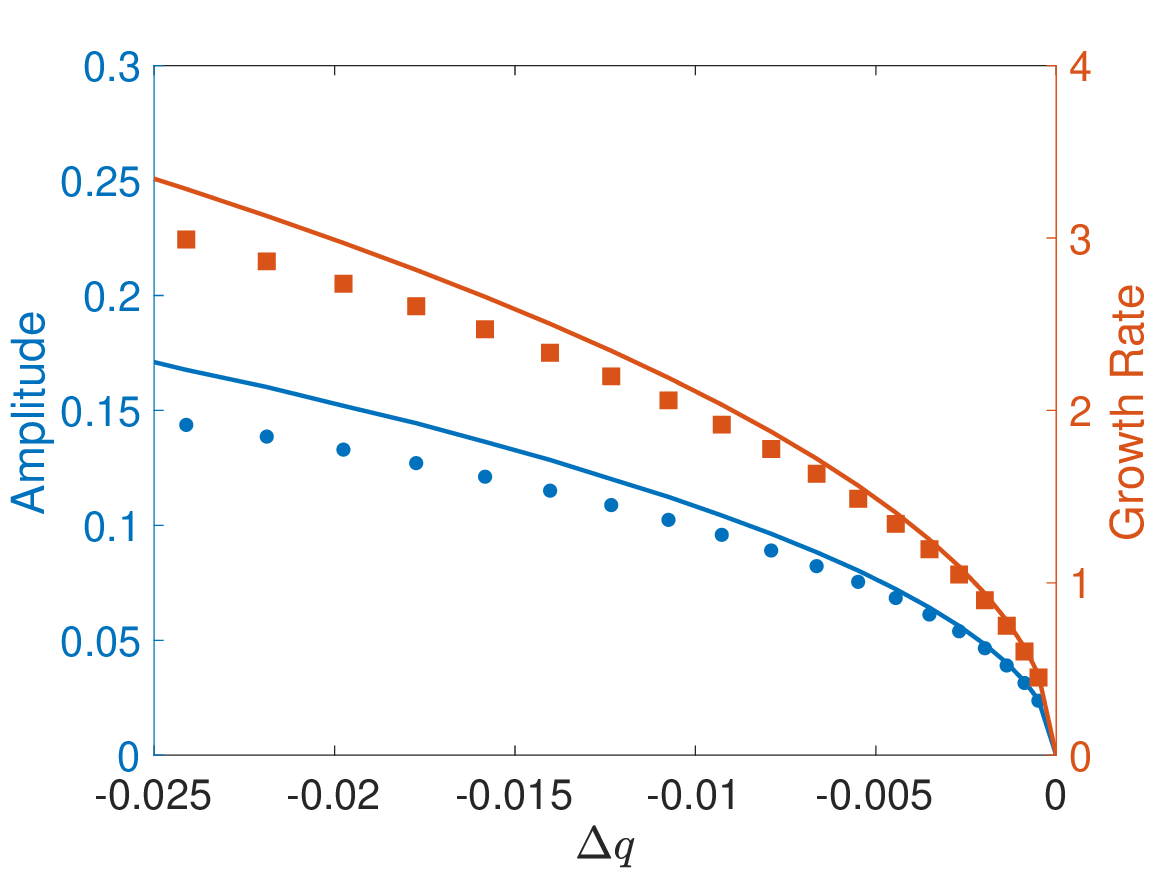} 
		\end{tabular}
	}
	\caption{\textbf{(a)} A transition front bifurcating from $\mu_2$ with nonlinear coefficients 
		corresponding to a damped magnetic lattice, namely $K_3=1$, $K_4 = 0.8$ and $\tilde{c} = 0.1$. The other parameter values
		and spectral picture is identical to Fig.~\ref{fig:simulation_K30_hard}(a).
		The initial value is $u_n(0) = 10^{-4} \sin(q_4 n)$ and $\epsilon = 0.1$. The displacement of the 6th particle is shown as a function of time. The modulation amplitude is $k_a = k_a^0 + \epsilon^2 = 0.06$, $k_b = k_b^0 + \epsilon^2 = 0.7587$ and the damping constant is $c = \epsilon \tilde{c} = 0.01$. The dashed line shows an approximation of the envelope given by Eq.~\eqref{eq:nls_ansatz}
		with $\chi = 1.997$, $\lambda''(\ell_0) = 26.3821$ and $\epsilon = 0.1$.
		For this value of $\epsilon$, the distance of the wavenumber to the edge of
		the gap is $\Delta q = q_4 - q_r = - 0.005$.
		\textbf{(b)} Same as panel (a) with $\epsilon = 0.22$. 
		\textbf{(c)} Plot
		of the amplitude of the front ($\mathrm{max}_t u_6(t)$) for the numerical simulation (blue dots)
		and prediction based on Eq.~\eqref{eq:nls_ansatz} (blue line) as a function of $\Delta q$.
		The real part of the Floquet exponent corresponding to $q_4$ (solid red squares) and asymptotic approximation $r_0^+ \epsilon$ (red line) are also shown, which indicates the initial growth rate.
	}
	\label{fig:simulation_magnetic_damped}
\end{figure}

\begin{figure} 
	\centerline{
		\begin{tabular}{@{}p{0.5\linewidth}@{}p{0.5\linewidth}@{}}
			\rlap{\hspace*{5pt}\raisebox{\dimexpr\ht1-.1\baselineskip}{\bf (a)}}
			\includegraphics[height=6cm]{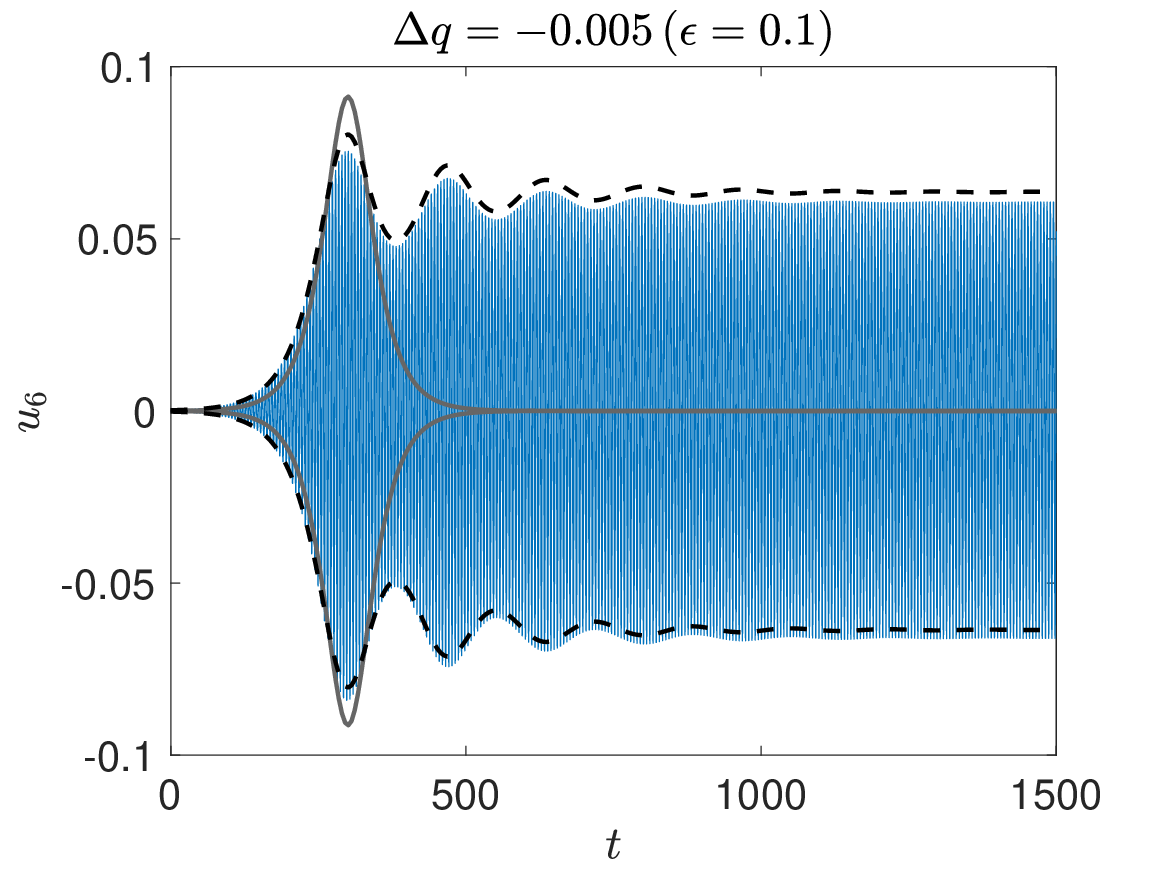}  &
			\rlap{\hspace*{5pt}\raisebox{\dimexpr\ht1-.1\baselineskip}{\bf (b)}}
			\includegraphics[height=6cm]{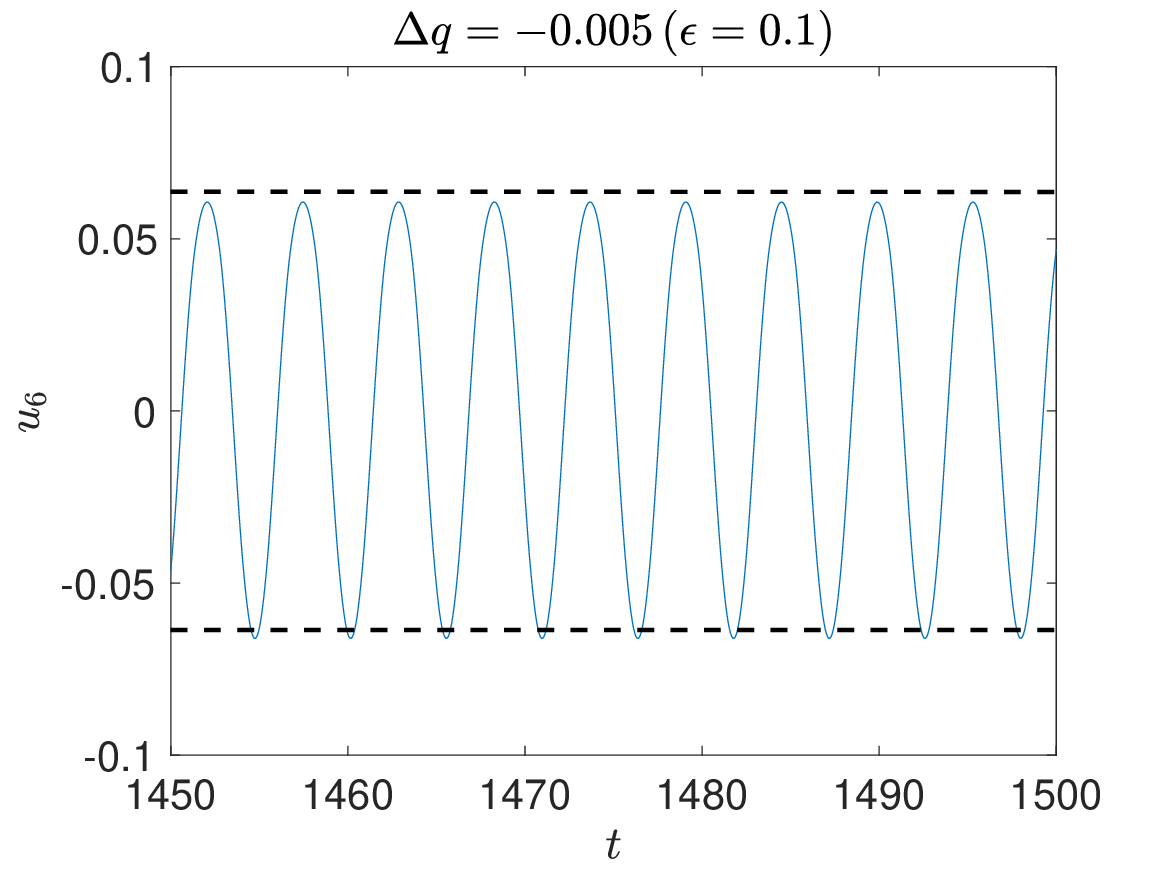} 
		\end{tabular}
	}
	\caption{
		\textbf{(a)} Same as Fig.~\ref{fig:simulation_magnetic_damped}(a) but over a longer
		time interval. The dashed line shows an approximation of the envelope given by Eq.~\eqref{eq:nls_ansatz},
		which accounts for damping. The solid gray line shows the homoclinic approximation (i.e., with no damping) of the envelope given by Eq.~\eqref{eq:soliton}.
		\textbf{(b)} A zoom of panel (a) for large values of $t$. Here it can be
		seen the dynamics are very close to periodic.
	}
	\label{fig:longtime_damped}
\end{figure}

The amplitude equation Eq.~\eqref{amplitude-nonlinear} can
 be used to approximate the front dynamics. However, the equation
does not yield an explicit solution in the presence of damping and so we will employ
a qualitative and numerical analysis of Eq.~\eqref{amplitude-nonlinear}.
A straightforward phase plane analysis shows that the trivial fixed point $(A,A') = (0,0)$ is a saddle
with corresponding eigenvalues $r_0^{\pm} = (-s_3 \pm \sqrt{s_3^2 + 4 s_1})/2$ and the fixed points
$(A,A') = (\pm \sqrt{s_1/s_2},0)$ are spiral-sinks with corresponding eigenvalues 
$r_{1}^{\pm}=(-s_3 \pm i \sqrt{8 s_1 - s_3^2 })/2$
where the $s_j$
are the coefficients of Eq.~\eqref{amplitude-nonlinear}, namely,
$$ s_1 =\left(\frac{\tilde{c}^2}{4\underline{m}} - \delta \right ) \frac{2}{\lambda''(\ell_0)} + \frac{\tilde{c}^2}{4\underline{m}^2} > 0,
\qquad s_2 =  \frac{2 \chi}{\lambda''(\ell_0)} > 0,
\qquad s_3 = \frac{\tilde{c}}{\underline{m}} > 0.$$
In the phase plane, there is a heteroclinic orbit that leaves the trivial fixed point along
the unstable eigenvector $(1,r_+^0)^T$ and approaches the $(\sqrt{s_1/s_2},0)$ fixed point.
There is another heteroclinic orbit that leaves the trivial fixed point along
the unstable eigenvector $(-1,-r_+^0)^T$ and approaches the $(-\sqrt{s_1/s_2},0)$ fixed point,
see Fig.~\ref{fig:heteroclinic}. To approximate the heteroclinic orbit, we numerically
solve Eq.~\eqref{amplitude-nonlinear} with initial condition $A(0) = 10^{-4}, A'(0) = 10^{-4} \, r_+^0$.
The resulting solution $A(\tau)$ is then used in Eq.~\eqref{eq:nls_ansatz} to generate the approximation of the lattice dynamics. 

Examples are shown in Fig.~\ref{fig:simulation_magnetic_damped}(a,b), where
the envelopes are shown for $\epsilon = 0.1$ and $\epsilon = 0.22$,
respectively. Once again, the envelope dynamics are captured
well by Eq.~\eqref{eq:nls_ansatz}, especially for small $\epsilon$.
The periodic oscillation of the envelope can be approximated by 
the imaginary part of the eigenvalue associated to the non-trivial
fixed point, namely $\tau_{\rm env} = 2\pi/( \sqrt{8 s_1 - s_3^2 }/2)$. In terms
of the original lattice variables this translates to $\tau_{\rm env}/\epsilon$.
For the example shown in Fig.~\ref{fig:longtime_damped}(a) with $\epsilon = 0.1$, 
the average peak-to-peak time of the envelope is $165.3$ time units,
where as $\tau_{\rm env}/\epsilon = 159.3$, which is quite close.

The front amplitude as a function of $\Delta q$ is shown as the solid blue line in panel (c),
and an approximation of the initial growth rate $r_+^0 \epsilon$ is shown as
the red line. Once again the asymptotic behavior as $\Delta q \rightarrow 0$ is correct. Despite the presence of damping, the front amplitude grows like $\mathcal{O}(\sqrt{|\Delta q|})$.

\section{Conclusions} \label{theend}

Generalized $q$-gap breathers are coherent structures that are localized in time, periodic in space, and have wavenumber in a $q$-gap. They are the natural counterpart of the discrete breathers of spatially periodic lattices, which themselves are of fundamental importance in a diverse range of fields. 

In the absence of damping, we proved rigorously the existence of generalized $q$-gap breathers in a time-periodic FPUT lattice using normal form theory. In particular, we proved
the existence of oscillating homoclinic solutions over a finite time interval with tails that can be made arbitrarily small, but finite. These solutions bifurcate from
one edge of the $q$-gap. Which of the edges is determined by the nonlinear coefficients $K_3,K_4$ and the concavity of the spectral band from which
the solution bifurcates. The amplitude of the $q$-gap breather grows like $\mathcal{O}(\sqrt{\Delta q})$,
where $\Delta q$ is the distance of the underlying wavenumber to the band edge.
This result makes rigorous the numerical and experimental observations of such $q$-gap breathers in \cite{chong2}.
We also provided a tractable analytical approximation of such solutions using a multiple-scale
analysis and corroborated results with direct numerical simulations. 

In the presence of damping we proved
the existence of solutions that connect the zero state to a time-periodic one,
which we called the transition fronts. The multiple-scale analysis
also provided an accurate description of the front solutions, although
the underlying amplitude equation needed to be solved numerically. The initial
stage of the front dynamics were well described by the undamped $q$-gap breather approximations.

Generalized $q$-gap breathers and transition fronts represent new types of nonlinear wave structures. This work
provided the first rigorous results in their study, complementing earlier experimental and numerical work.
Nonetheless, there are still many open questions regarding $q$-gap breathers and transition fronts.
This includes the possible existence of genuine $q$-gap breathers (i.e., with both tails decaying to zero),
the numerically exact computation of $q$-gap breathers (i.e., numerical roots of the appropriate map up to a user-prescribed tolerance) and
the exploration of such structures in higher spatial dimensions 
or in settings beyond the FPUT realm. Indeed, any system that is already described by a nonlinear wave equation that could be adapted to be time-varying
(in order to induce a $q$-gap) would be a candidate for the implementation of $q$-gap breathers. This suggests
that $q$-gap breathers' relevance, and hence the results of this work, could extend to a wide range of fields.

\bibliographystyle{siam}
\bibliography{GLbib,LatticeManuscript2}

\end{document}